\documentclass[11pt]{article}
\RequirePackage[OT1]{fontenc}
\RequirePackage{amsthm,amsmath,amssymb,amsfonts,mathtools,mathrsfs}
\RequirePackage{graphicx,lmodern}
\RequirePackage{microtype}
\RequirePackage[round,colon,authoryear]{natbib}
\usepackage{xcolor}
\usepackage{booktabs,array,tabularx}
\usepackage[shortlabels]{enumitem}
\usepackage{tikz}
\usetikzlibrary{patterns,decorations.pathreplacing,arrows.meta,calc}
\usepackage{algorithm}
\usepackage[noend]{algpseudocode}
\allowdisplaybreaks
\numberwithin{equation}{section}
\setlist{itemsep=0.3em,topsep=0.4em,parsep=0pt}
\newtheorem{theorem}{Theorem}[section]
\newtheorem{proposition}[theorem]{Proposition}
\newtheorem{lemma}[theorem]{Lemma}
\newtheorem{corollary}[theorem]{Corollary}
\theoremstyle{definition}
\newtheorem{definition}[theorem]{Definition}
\newtheorem{example}[theorem]{Example}

\theoremstyle{remark}
\newtheorem{remark}[theorem]{Remark}

\newcommand{\Cm}{\Gamma}
\newcommand{\sgn}{\operatorname{sgn}}
\tikzset{
known/.style={fill=black!18},
knownsel/.style={fill=black!45},
unknown/.style={pattern=north east lines, pattern color=black!45},
brace/.style={decorate,decoration={brace,amplitude=4pt}},
lab/.style={font=\small},
slab/.style={font=\footnotesize},
arr/.style={-{Stealth[length=5pt]},semithick},
}

\newcommand{\R}{\mathbb{R}}
\newcommand{\Q}{\mathbb{Q}}
\newcommand{\Rpos}{\mathbb{R}_{\ge 0}}
\newcommand{\Qpos}{\mathbb{Q}_{\ge 0}}
\newcommand{\rank}{\operatorname{rank}}
\newcommand{\rankp}{\operatorname{rank}_{+}}
\newcommand{\supp}{\operatorname{supp}}
\newcommand{\bits}{\operatorname{bits}}
\newcommand{\col}{\operatorname{col}}
\newcommand{\spanof}{\operatorname{span}}
\newcommand{\T}{\mathcal{T}}
\newcommand{\be}{\mathbf e}
\newcommand{\Id}{\operatorname{Id}}

\newcommand{\pfpart}[1]{\smallskip\noindent\emph{#1.}}
\newcounter{proofstepblock}
\newcounter{proofstep}[proofstepblock]

\newcommand{\proofstep}[2][]{\pfpart{\if\relax\detokenize{#1}\relax\stepcounter{proofstep}\else\refstepcounter{proofstep}\label{#1}\fi Step \theproofstep, #2}}
\newcounter{costpart}
\renewcommand{\thecostpart}{(\alph{costpart})}
\newcommand{\costpart}[1]{\textup{\refstepcounter{costpart}\label{#1}\thecostpart}}
\usepackage{float}
\usepackage{titlesec}
\RequirePackage[colorlinks,citecolor=blue,urlcolor=blue,linkcolor=blue]{hyperref}
\titlespacing*{\section}{0pt}{2.2ex plus .5ex minus .2ex}{1.0ex plus .2ex}
\titlespacing*{\subsection}{0pt}{1.6ex plus .4ex minus .2ex}{0.6ex plus .2ex}
\setlist{nosep}
\title{Structure and Exact Certification of Nonnegative Tensor Decompositions}
\author{Haoming Wang and Ming Yuan\\ Columbia University}
\date{}

\begin{document}
\maketitle

\begin{abstract}
We study the structure and exact certification of nonnegative tensor decompositions supplied through rational factor matrices. Every failure of minimality or essential uniqueness contains an irreducible positive exchange between $p$ supplied terms and $q$ competing terms. A dimension count based on the Lovitz--Petrov splitting theorem bounds the effective order of such an exchange, the sum over the modes of the joint factor-span dimension minus one, by $p+q-2$, and this bound is attained at every supplied length $R\ge2$. After eliminating ambient-row variables by the Farkas lemma, the resulting semialgebraic systems have fewer than $2R^2$ real variables. For input bit length $L$, this gives a deterministic $G(R)L^{O(R^2)}$-time algorithm that computes the nonnegative rank and an exact minimum-length decomposition, decides minimality and essential uniqueness, and returns exact witnesses when either property fails. The exponent is independent of the tensor order and the mode dimensions. For Lebesgue moment tensors, positive quadrature determines the nonnegative rank and the uniqueness behavior, and some rational length-$R$ inputs force a common field of degree at least $\lfloor R/2\rfloor!$ for any exact minimum-length output. As an illustration, an implementation resolves all 114 instances left inconclusive by earlier positive-scattering criteria, of which 95 are minimal and essentially unique and 19 are minimal but nonunique.
\end{abstract}



\section{Introduction}\label{sec:introduction}

Nonnegative matrix and tensor factorizations are widely used to represent data by a small number of interpretable nonnegative components.  Once a decomposition has been produced, however, a basic question remains: \emph{how can one certify what has been found?}  Suppose that we are given a nonnegative CP decomposition
\[
\T=\sum_{r=1}^{R}a_r^{(1)}\otimes\cdots\otimes a_r^{(d)},
\qquad a_r^{(m)}\in\Rpos^{n_m}\setminus\{0\}.
\]
There are three natural exact questions.  Is the supplied length $R$ minimal, that is, does $R=\rankp(\T)$?  If it is not minimal, can one construct a shorter nonnegative decomposition?  If it is minimal, is it essentially unique, and if uniqueness fails, can one construct an inequivalent decomposition of the same length?  Here essential uniqueness means uniqueness of the multiset of rank-one terms, equivalently up to permutation and within-term rescaling of the factors.  If $R$ is not minimal, the reported number of components is not intrinsic to $\T$; if the decomposition is not essentially unique, its individual rank-one components are not identified by $\T$.

These questions are distinct from the problem of \emph{finding} a decomposition. A large literature develops numerical and algebraic methods for computing tensor decompositions \citep[e.g.,][]{KoldaBader2009,AnandkumarEtAl2014}, while another line of work gives sufficient conditions for identifiability or uniqueness, beginning with Kruskal's theorem \citep{Kruskal1977} and including later deterministic and generic criteria \citep[e.g.,][]{DomanovDeLathauwer2013a,DomanovDeLathauwer2013b,LovitzPetrov2023, WangYuan2026}. Such conditions are valuable when their hypotheses hold, but failure of a sufficient condition is generally inconclusive. Here we ask for complete exact certification: given an arbitrary supplied decomposition, decide minimality and essential uniqueness in every case and return exact witnesses when either property fails. We use the term \emph{certification} in this complete exact decision-and-search sense, rather than merely for checking a sufficient condition.

The input model is important.  We are given the rational factor matrices of the supplied length-$R$ decomposition, with total bit length $L$, rather than the expanded tensor.  The factor list has $R\sum_m n_m$ entries, whereas the ambient tensor has $\prod_m n_m$ entries and can be exponentially larger.  Our algorithms work directly with the compact representation and never expand the tensor.  The model also reflects the form in which numerical factorization methods typically return their output: a finite list of floating-point factor entries, each of which admits an exact rational representation.  Thus the question studied here applies directly once such an output is converted to exact rational data.

The difficulty of complete certification is already visible for matrices.  When $d=2$, deciding whether the nonnegative rank of a $m\times n$ matrix is at most a prescribed value, say $r$, is a highly nontrivial exact problem.  \citet{AroraGeKannanMoitra2012} gave a polynomial-time algorithm for every fixed target rank, and \citet{Moitra2016} obtained an $(mn)^{O(r^2)}$ algorithm by replacing a naive $r(m+n)$-variable semialgebraic formulation by one with $2r^2$ real variables.  Thus even in order two, controlling the number of real variables is central to obtaining a useful exact algorithm.

A direct tensor formulation faces the same dimensionality barrier in a sharper form.  A length-$q$ competing decomposition has
\begin{equation}\label{eq:direct-vars}
q\sum_{m=1}^d n_m
\end{equation}
real factor variables.  Direct application of general real-algebraic decision procedures therefore yields a running-time exponent that grows with the tensor order and ambient mode dimensions.  The central question is whether one can certify the supplied decomposition without paying for those ambient dimensions in the exponent.

Our answer is affirmative. The main structural contribution is a bounded-dimensional description of the positive exchanges through which a competing decomposition can occur. Sharpness and necessity results explain the geometry of this reduction; exact certification and reconstruction follow from it. More specifically, our main result states:

\begin{theorem}[Exact certification]\label{thm:main}
There are a computable function $G$, depending only on $R$, and an absolute constant $C$ with the following property.  For every tensor order $d\ge2$, an explicitly listed rational supplied decomposition of length $R$ and total bit length $L$ can be processed deterministically in time
\begin{equation}\label{eq:running-time}
G(R)L^{CR^2},
\end{equation}
and the algorithm
\begin{enumerate}[label=\textup{(\roman*)}]
\item computes $\rankp(\T)$, which is at most $R$;
\item\label{item:main-decomposition} outputs an exact minimal nonnegative CP decomposition of $\T$;
\item\label{item:main-decisions} decides whether the supplied decomposition is minimal and whether it is essentially unique;
\item\label{item:main-witnesses} outputs an exact shorter competitor when minimality fails, and an exact inequivalent length-$R$ competitor when the supplied decomposition is minimal but not essentially unique. 
\end{enumerate}
\end{theorem}

Section~\ref{sec:one-shot} states Algorithm~\ref{alg:main} and proves Theorem~\ref{thm:main}.  The output entries can be irrational; they are real algebraic numbers in one number field and are represented exactly, see Remark~\ref{rem:main-format}.  The function $G$ is not optimized; see Remark~\ref{rem:main-format} and the open questions in Section~\ref{sec:discussion}.

Moreover, the field degree in the output cannot be avoided, and it cannot be bounded by a polynomial in $R$.

\begin{theorem}[Output field degree]\label{thm:intro-degree}
For every $R\ge4$ there is a rational supplied decomposition of length $R$ whose tensor has a unique minimal nonnegative decomposition, such that every field which contains all entries of the factor matrices of an exact minimum-length decomposition has degree at least $\lfloor R/2\rfloor!$ over $\Q$.  In particular, no polynomial in $R$ bounds the degree of the number field of the output in Theorem~\ref{thm:main}.
\end{theorem}

Theorem~\ref{thm:intro-degree} follows from Theorem~\ref{thm:degree} and Corollary~\ref{cor:output-degree-lower}, which is proved in Section~\ref{subsec:output-lower} for moment tensors of Laguerre zeros.
The lower bound also shows that the large degree of the common algebraic number field in our exact output representation is not merely an artifact of the algorithm. 

\paragraph{Contributions.}
This paper builds on the positive exchanges, the Lovitz--Petrov splitting inequality \citep{LovitzPetrov2023}, and the sufficient criteria of \citet{WangYuan2026}.  Its contributions are the following.
\begin{enumerate}[label=\textup{(\alph*)}]
\item \emph{Structure.}  Every failure of minimality or essential uniqueness is localized to irreducible exchanges on subsets of the supplied terms (Theorems~\ref{thm:localized-witness} and~\ref{thm:rank-localization}).  The effective order of such an exchange is at most $p+q-2$, independently of the tensor order (Theorem~\ref{thm:budget}).
\item \emph{Exact reduction and algorithm.}  Elimination of the ambient rows by the Farkas lemma gives polynomial systems with fewer than $2R^2$ real variables, and hence Theorem~\ref{thm:main}.
\item \emph{Sharpness and necessity.}  The effective-order and variable bounds are attained at every $R\ge2$ (Section~\ref{sec:sharpness}), and no ingredient of the reduction can be omitted (Appendix~\ref{app:necessity}).
\item \emph{Algebraic output.}  Exact outputs require algebraic numbers, and a common field can have degree growing factorially in $R$ (Theorem~\ref{thm:intro-degree}).  Moment tensors of positive quadrature rules supply explicit families with known rank and uniqueness behavior (Section~\ref{sec:moments}).
\end{enumerate}

The key reduction is from the ambient-dimensional formulation in \eqref{eq:direct-vars} to polynomial systems with fewer than $2R^2$ real variables.  Consequently the supplied-decomposition problem is in XP when parameterized by $R$.  Whether it is fixed-parameter tractable, with running time $G(R)L^{O(1)}$, remains open, see discussions in Section~\ref{sec:discussion}.  The reduction has four ingredients.

\paragraph{1. Localizing a global competitor.}
Comparing the supplied decomposition with a shorter or inequivalent decomposition gives a positive exchange between their rank-one terms.  Decomposing this exchange into irreducible blocks shows that every failure of minimality contains a \emph{shortening localized replacement}: some $p$ supplied terms can be replaced by $q<p$ nonnegative rank-one terms.  If the supplied decomposition is minimal but not essentially unique, some $p$ supplied terms admit an inequivalent replacement by $p$ terms.  Moreover, a minimum-length decomposition can be obtained by performing shortening replacements on pairwise disjoint subsets of the supplied terms.  Thus the global search reduces to localized competitors involving subsets of the $R$ given terms; see Theorems~\ref{thm:localized-witness} and~\ref{thm:rank-localization}.

\paragraph{2. Bounding the effective order.}
For an irreducible replacement, let $\delta_m$ be the dimension of the joint span of the supplied and competing mode-$m$ factors.  The dimension count behind the positive splitting inequality of \citet{WangYuan2026}, applied to these joint spans and stated without the scattering term, gives
\[
\kappa:=\sum_{m=1}^d(\delta_m-1)\le p+q-2;
\]
see Theorem~\ref{thm:budget}, which rests on the Lovitz--Petrov splitting theorem.  In particular, at most $p+q-2$ modes can be active, independently of the original tensor order, and $\sum_{m\in H}\delta_m\le2\kappa$, so the total active-mode dimension of the core is controlled solely by the number of terms in the exchange.  Section~\ref{sec:sharpness} shows that the effective-order bound and both resulting variable bounds are attained for shortening and balanced exchanges at every $R\ge2$.

\paragraph{3. Eliminating ambient-dimensional variables.}
For each localized replacement we enumerate a finite collection of profiles describing the active modes and selected rows.  After normalization, its core uses exactly $q\kappa+q-1$ real variables.  The omitted rows still impose nonnegativity constraints, but introducing their coefficient rows as variables would reintroduce ambient-dimensional dependence.  A Farkas-lemma argument eliminates these coefficients: feasibility is reduced to polynomial conditions on extreme rays of a polyhedral cone, represented by minors.  The ambient dimensions affect the number of constraints but introduce no additional real variables.  Theorem~\ref{thm:representation}, Section~\ref{sec:lifting}, and Proposition~\ref{prop:dimension} make this reduction precise and show that combining disjoint shortening blocks leaves fewer than $2R^2$ variables.

\paragraph{4. Exact semialgebraic search and reconstruction.}
We enumerate the resulting bounded-variable polynomial systems and solve them by exact algorithms from real algebraic geometry.  Searching possible resulting lengths computes $\rankp(\T)$; balanced term-disjoint replacements decide essential uniqueness.  A feasible core solution is then lifted back to the original dimensions by solving linear systems over the same real algebraic number field, yielding exact competing decompositions without numerical approximation.  Algorithm~\ref{alg:main} carries out this search and reconstruction.

\medskip
\medskip
The decisive algorithmic step is the structural reduction to $O(R^2)$ real variables. This connects our work to algebraic and geometric approaches to tensor ranks and decomposition \citep{LandsbergTeitler2010,FriedlandOttaviani2014,Blekherman2015,Nie2017,Mourrain2018}, and to polynomial-optimization approaches to factorization ranks \citep{GriblingDeLaatLaurent2019}.

The closest algorithmic precedent is Moitra's bounded-variable reduction for exact nonnegative matrix factorization \citep{Moitra2016}. Algorithms for unrestricted tensor rank \citep{BhargavaSarafVolkovich2021,PelegShpilkaVolk2024,BhargavaShringi2025} and work on nonnegative tensor geometry and identifiability \citep{QiComonLim2016,LovitzPetrov2023,WangYuan2026} provide complementary perspectives. Our task is complete exact decision and search for an arbitrary supplied nonnegative decomposition. Section~\ref{sec:discussion} compares the input models, guarantees, and structural criteria in detail.

\paragraph{Why matricization does not suffice.}
Every nonnegative CP decomposition of $\T$ induces an NMF of every matricization $\T_{(S)}$ obtained by grouping the modes into two sides, so
\begin{equation}\label{eq:unfolding-rank}
\rankp(\T_{(S)})\le \rankp(\T).
\end{equation}
The converse fails, because the two factors of a nonnegative rank-one matrix term of $\T_{(S)}$ need not split as Kronecker products across the grouped modes.  Neither the unfolding rank nor Moitra's formulation augmented by rank-one constraints gives a complete tensor certificate.  For the tensor $\T=\be_1\otimes\be_1\otimes\be_2+\be_1\otimes\be_2\otimes\be_1+\be_2\otimes\be_1\otimes\be_1$ with $\be_1,\be_2$ the standard basis of $\R^2$, the supplied decomposition is minimal and essentially unique with $\rankp(\T)=3$, every matricization has nonnegative rank two, and for $M=\T_{(1)}$ and every $q\le3$ the formulation of \citet[Section~2.4]{Moitra2016} with the added condition that every right factor reshapes to a rank-one matrix has no solution; see Example~\ref{ex:matricization} in Appendix~\ref{app:matricization}.  In particular, splitting the last term into two halves gives a nonminimal supplied decomposition of length four that this constrained formulation fails to detect.  Moreover, a balanced unfolding can have product-sized matrix dimensions and thereby forfeit the compact factor-list representation.
Thus the obstruction is not merely the size of an unfolding: matricization can lose precisely the multilinear structure that certification must preserve.

\medskip\noindent\emph{Organization.}
Section~\ref{sec:prelim} fixes notation and the input model.  Sections~\ref{sec:localization} and~\ref{sec:budgets} establish localization and the effective-order bound, and Section~\ref{sec:normalized} develops the exact semialgebraic reduction.  Section~\ref{sec:one-shot} presents the algorithm, together with the exact real-algebraic tools it uses, and proves Theorem~\ref{thm:main}.  Section~\ref{sec:sharpness} proves sharpness, Section~\ref{sec:moments} studies moment tensors and output complexity, Section~\ref{sec:experiments} gives computational evidence, and Section~\ref{sec:discussion} discusses related work and open problems.  The appendices contain the matricization analysis, the necessity examples, deferred proofs, the complexity analysis, and implementation details.

\section{The supplied decomposition}\label{sec:prelim}

\subsection{Notation and conventions}\label{subsec:notation}

Write $[n]=\{1,\ldots,n\}$.  All tensors and factors are real.  Multisets are written with braces, like sets, but an element may occur in a multiset several times.  Two multisets are equal if every element occurs in both the same number of times.  The letter $G$ denotes a computable function of its displayed arguments; its value may change from one occurrence to the next.  The letter $C$, with or without a subscript, denotes an absolute constant and may likewise change unless a specific occurrence is referenced later.  Constants hidden in $O(\cdot)$ are absolute unless a dependence is displayed explicitly.  Rational numbers are encoded by numerator and positive denominator in binary, and $\bits(\cdot)$ denotes total binary encoding length, including dimensions and separators.  For a matrix $M$ and index sets $S$ and $S'$, we write $M[S,S']$ for the submatrix with rows $S$ and columns $S'$, and $M[S,:]$ and $M[:,S']$ for row and column restrictions. For a set $J$ of column indices, we also write $M_J=M[:,J]$.  For a vector $\eta$ we write $\supp(\eta)$ for the set of indices of its nonzero entries and $\eta_J$ for its restriction to $J$.  A row basis is a set of row indices whose rows form a basis of the row space.  The left kernel of $M$ is the space of all row vectors $z$ with $zM=0$.

For $k\ge1$ we write
\[
\mathbb Z[y_1,\ldots,y_k]=\Bigl\{\sum_{\alpha\in E}c_\alpha\,y_1^{\alpha_1}\cdots y_k^{\alpha_k}\ :\ E\subseteq\mathbb Z_{\ge0}^{k}\text{ finite},\ c_\alpha\in\mathbb Z\Bigr\}
\]
for the ring of polynomials in the variables $y_1,\ldots,y_k$ with integer coefficients.  For $k=1$ and the variable $T$ we write $\mathbb Z[T]$.  The degree of a nonzero polynomial is the largest $\alpha_1+\cdots+\alpha_k$ with $c_\alpha\ne0$.  The \emph{bit length} of the polynomial is the largest bit length of its coefficients $c_\alpha$, as in \citet[Section~2]{BouzidiEtAl2015}. A real number $\theta$ is \emph{algebraic} if $f(\theta)=0$ for some nonzero $f\in\mathbb Z[T]$.  In this case
\[
\Q(\theta)=\Bigl\{\frac{g_1(\theta)}{g_0(\theta)} : g_0,g_1\in\mathbb Z[T],\ g_0(\theta)\ne0\Bigr\}
\]
is the smallest subfield of $\R$ which contains $\theta$, and $\Q(\theta)^k$ is the set of all vectors with $k$ coordinates in $\Q(\theta)$. The \emph{degree} of $\Q(\theta)$ is the smallest degree of a nonzero $f\in\mathbb Z[T]$ with $f(\theta)=0$ \citep[1.25 and p.~19]{MilneFT}.  Finitely many algebraic numbers \emph{generate} the smallest subfield of $\R$ which contains them, and this field equals $\Q(\theta)$ for some algebraic $\theta$ \citep[Proposition~1.30 and Aside~2.26]{MilneFT}. Throughout the paper, the word \emph{exact} means that all computations are carried out in exact arithmetic over $\Q$ or over a field $\Q(\theta)$, and that no numerical approximation is used.

\subsection{The supplied decomposition and its competitors}\label{sec:model}

Let $d\ge2$ and let
\begin{equation}\label{eq:given}
\T=\sum_{r=1}^{R}a_r^{(1)}\otimes\cdots\otimes a_r^{(d)},
\qquad
a_r^{(m)}\in\Rpos^{n_m}\setminus\{0\},
\end{equation}
be an explicitly \emph{supplied} nonnegative CP decomposition, with factor matrices
\[
A^{(m)}=[a_1^{(m)},\ldots,a_R^{(m)}]\in\Rpos^{n_m\times R}.
\]
The decomposition is \emph{minimal} if no shorter nonnegative CP decomposition of $\T$ exists.  It is \emph{essentially unique} if every nonnegative decomposition of $\T$ of length $R$ has the same multiset of rank-one terms.

Write
\[
P_r=a_r^{(1)}\otimes\cdots\otimes a_r^{(d)},
\qquad r\in[R],
\]
for the supplied rank-one terms, so that $\T=\sum_{r=1}^{R}P_r$.

A nonnegative rank-one tensor is a product
\[
x^{(1)}\otimes\cdots\otimes x^{(d)},
\qquad
x^{(m)}\in\Rpos^{n_m}\setminus\{0\}.
\]
A nonnegative CP decomposition is a finite sum of such tensors.  Its length is the number of nonzero terms, and the nonnegative rank $\rankp(\T)$ is the smallest possible length.

The algorithmic input is the explicit rational list
\begin{equation}\label{eq:input-factors}
A^{(m)}\in\Qpos^{n_m\times R},
\qquad m\in[d],
\end{equation}
with no zero column.  We write $A=(A^{(1)},\ldots,A^{(d)})$ for the full supplied factor data.  The input size $L=\bits(A)$ is the total number of bits in the explicit list.  The list contains all entries of every matrix $A^{(m)}$, including the zero entries.  Thus, up to the fixed conventions of the encoding, $d\le L$, $\sum_m n_m\le L$, and $R\sum_m n_m\le L$.  The factor list can be exponentially smaller than an explicit listing of the ambient tensor, which has $\prod_m n_m$ entries; our algorithm never needs to expand this tensor.  Our bounds do not directly apply to sparse factor formats, in which $n_m$ can be much larger than $L$.  All running-time bounds below are stated in terms of $L$, and they count bit operations.

We recall the following fact from \citet[Lemma 1(i)]{WangYuan2026}.
\begin{lemma}\label{lem:product-equality}
Suppose
\[
x^{(1)}\otimes\cdots\otimes x^{(d)}
=
y^{(1)}\otimes\cdots\otimes y^{(d)}\ne0.
\]
Then $y^{(m)}=\lambda_mx^{(m)}$ for nonzero scalars $\lambda_m$ satisfying $\prod_m\lambda_m=1$.  If all factors are nonnegative, every $\lambda_m$ is positive.
\end{lemma}

Two decompositions of the same length are \emph{equivalent} when their multisets of rank-one terms agree. By Lemma~\ref{lem:product-equality}, this holds if and only if the factors of the two decompositions agree after a permutation of the terms and positive rescalings of the factors of every term, where the product of the scalars of every term is one.

A length-$q$ \emph{competitor} of \eqref{eq:given} is a decomposition
\begin{equation}\label{eq:competitor}
\T=\sum_{j=1}^{q}b_j^{(1)}\otimes\cdots\otimes b_j^{(d)},
\qquad
b_j^{(m)}\in\Rpos^{n_m}\setminus\{0\}.
\end{equation}
Its factor matrix in mode $m$ is
\[
B^{(m)}=[b_1^{(m)},\ldots,b_q^{(m)}]\in\Rpos^{n_m\times q}.
\]
A competitor with $q<R$ shows that the supplied decomposition is not minimal. At $q=R$, a competitor is \emph{nontrivial} when it is not equivalent to the supplied decomposition.  A competitor may have arbitrary nonnegative real entries. Even for matrices, i.e. for $d=2$, there are rational inputs such that every minimal nonnegative factorization has irrational entries \citep{ChistikovEtAl2017,Shitov2016}. Therefore, the entries of an exact output must be allowed to be irrational real algebraic numbers.  Their representation is described in Section~\ref{subsec:exact-computation}.

\section{Positive exchanges and localization}\label{sec:localization}

We use the positive-exchange framework of \citet{WangYuan2026}. The recalled machinery below leads to the two localization theorems, whose proofs we retain here; the effective-order bound in Section~\ref{sec:budgets} begins the structural reduction used for exact certification. Routine proofs of Lemmas~\ref{lem:block-decomposition}, \ref{lem:localized-minimal-subsums}, and~\ref{lem:nonminimal-nonunique}, Proposition~\ref{prop:blocks}, and Lemma~\ref{lem:combine} are deferred to Appendix~\ref{app:proofs}.

\begin{definition}\label{def:exchange}
A positive exchange is an equality
\[
\sum_{r\in I}P_r=\sum_{j\in J}Q_j
\]
between finite sums of nonzero nonnegative rank-one tensors.  A subexchange is a pair $(I',J')$ with $I'\subseteq I$, $J'\subseteq J$, and
\[
\sum_{r\in I'}P_r=\sum_{j\in J'}Q_j.
\]
The exchange is \emph{irreducible} if it has no nonempty proper subexchange.
\end{definition}

Nonnegativity forces both sides of a nonempty subexchange to be nonempty \citep[Remark~1]{WangYuan2026}. We recall \citet[Lemma~3]{WangYuan2026}.

\begin{lemma}\label{lem:block-decomposition}
Every positive exchange decomposes as a disjoint union of irreducible positive exchanges, which we call its \emph{blocks}.  If both sides are minimal nonnegative decompositions of the same tensor, then the two sides of every block have the same number of terms.
\end{lemma}

\begin{definition}\label{def:localized-replacement}
For a nonempty set $I\subseteq[R]$ and nonzero nonnegative rank-one tensors $Q_1,\ldots,Q_q$, an equality
\begin{equation}\label{eq:localized-replacement}
\sum_{r\in I}P_r=\sum_{j=1}^{q}Q_j
\end{equation}
is a \emph{localized replacement} on $I$: \emph{shortening} if $q<|I|$, \emph{balanced} if $q=|I|$, and \emph{term-disjoint} if $\{P_r:r\in I\}$ and $\{Q_j:j\in[q]\}$ have empty multiset intersection. A balanced replacement is a \emph{localized alternative} if these multisets differ. Irreducibility refers to Definition~\ref{def:exchange}.
\end{definition}

\begin{lemma}\label{lem:localized-minimal-subsums}
Let $Z_1,\ldots,Z_n$ be nonzero nonnegative rank-one tensors with $\rankp\bigl(\sum_{i=1}^{n}Z_i\bigr)=n$.  Then
\begin{equation}\label{eq:localized-subsum-rank}
\rankp\!\left(\sum_{i\in K}Z_i\right)=|K|
\end{equation}
for every nonempty $K\subseteq[n]$.
\end{lemma}

\begin{lemma}\label{lem:nonminimal-nonunique}
If the supplied length-$R$ decomposition is nonminimal, then it is not essentially unique.
\end{lemma}

Both localization theorems follow from the following block structure.

\begin{proposition}[Blocks relative to a minimal decomposition]\label{prop:blocks}
Compare the supplied decomposition $\T=\sum_{r=1}^{R}P_r$ with a minimal decomposition $\T=\sum_{j=1}^{\rho^*}Q_j$, where $\rho^*=\rankp(\T)$. Lemma~\ref{lem:block-decomposition} gives irreducible blocks
\begin{equation}\label{eq:blocks-minimal}
\sum_{r\in I_\nu}P_r=\sum_{j\in J_\nu}Q_j,
\qquad\nu\in[c],
\end{equation}
where $[R]=I_1\sqcup\cdots\sqcup I_c$ and $[\rho^*]=J_1\sqcup\cdots\sqcup J_c$.  Put $p_\nu=|I_\nu|$ and $q_\nu=|J_\nu|$.  Then
\begin{enumerate}[label=\textup{(\alph*)}]
\item\label{item:blocks-rank} $1\le q_\nu\le p_\nu$ and $\rankp\bigl(\sum_{r\in I_\nu}P_r\bigr)=q_\nu$ for every $\nu\in[c]$;
\item\label{item:blocks-count} $\sum_{\nu=1}^{c}(p_\nu-q_\nu)=R-\rho^*$.
\end{enumerate}
In particular, if $\rho^*<R$, then some block satisfies $q_\nu<p_\nu$, and for it $p_\nu\ge2$.
\end{proposition}

\begin{theorem}[Localization of obstructions]\label{thm:localized-witness}
Let $\T=\sum_{r=1}^{R}P_r$ be the supplied decomposition.
\begin{enumerate}[label=\textup{(\roman*)}]
\item\label{item:localized-minimality}
The supplied decomposition is nonminimal if and only if it admits a shortening localized replacement.  If it is nonminimal, such a replacement can be chosen irreducible with
\[
2\le p:=|I|\le R,
\qquad
1\le q<p,
\qquad
\rankp\!\left(\sum_{r\in I}P_r\right)=q.
\]
\item\label{item:localized-uniqueness}
Assume the supplied decomposition is minimal.  It is not essentially unique if and only if it admits a balanced localized alternative.  Whenever it is not essentially unique, such an alternative can be chosen irreducible and term-disjoint with $2\le p:=|I|=q\le R$.
\end{enumerate}
\end{theorem}

\begin{proof}
\pfpart{Part \ref{item:localized-minimality}, sufficiency}
Replacing the terms indexed by $I$ in \eqref{eq:localized-replacement} gives $\T=\sum_{r\notin I}P_r+\sum_{j=1}^{q}Q_j$, a nonnegative decomposition of length $R-|I|+q<R$.

\pfpart{Part \ref{item:localized-minimality}, necessity}
Suppose that $\rho^*=\rankp(\T)<R$, and choose a minimal decomposition $\T=\sum_{j=1}^{\rho^*}Q_j$.  By Proposition~\ref{prop:blocks}, some block \eqref{eq:blocks-minimal} satisfies $p:=p_\nu\ge2$ and $q:=q_\nu<p$, and $\rankp\bigl(\sum_{r\in I_\nu}P_r\bigr)=q$.  This block is an irreducible shortening localized replacement on $I=I_\nu$, with $1\le q<p\le R$.

\pfpart{Part \ref{item:localized-uniqueness}, sufficiency}
A balanced localized alternative on $I$ gives the length-$R$ decomposition $\T=\sum_{r\notin I}P_r+\sum_{j=1}^{p}Q_j$. If it were equivalent to the supplied decomposition, canceling the common multiset $\{P_r:r\notin I\}$ would give $\{P_r:r\in I\}=\{Q_j:j\in[p]\}$, contrary to the definition of an alternative.

\pfpart{Part \ref{item:localized-uniqueness}, necessity}
Assume that the supplied decomposition is minimal and not essentially unique.  Choose an inequivalent length-$R$ decomposition $\T=\sum_{j=1}^{R}Q_j$.  Since $\rankp(\T)=R$, it is minimal as well.  Cancel from the two term multisets every common rank-one tensor, counted with multiplicity, and let $I_0$ and $J_0$ be the surviving index sets.  Both are nonempty, since otherwise the two multisets would be equal.  The surviving sums remain equal,
\begin{equation}\label{eq:localized-cancelled-exchange}
\sum_{i\in I_0}P_i=\sum_{j\in J_0}Q_j,
\end{equation}
and the two surviving term multisets are disjoint.  By Lemma~\ref{lem:localized-minimal-subsums}, applied to each of the two minimal decompositions, both sides of \eqref{eq:localized-cancelled-exchange} are minimal decompositions of the same tensor.  Lemma~\ref{lem:block-decomposition} therefore decomposes \eqref{eq:localized-cancelled-exchange} into irreducible blocks $\sum_{i\in I_\nu}P_i=\sum_{j\in J_\nu}Q_j$, and every block is balanced, $|I_\nu|=|J_\nu|$.

A size-one block would give $P_i=Q_j$, contradicting cancellation. Any block therefore gives the required irreducible, term-disjoint alternative with $2\le p=q\le R$.
\end{proof}

\subsection{Disjoint localization of the nonnegative rank}

Several shortening replacements on disjoint supplied subsets suffice to attain $\rankp(\T)$.

\begin{definition}\label{def:pattern}
A \emph{disjoint shortening pattern} is a finite family $\mathfrak D=((I_\nu,q_\nu))_{\nu=1}^{c}$ of \emph{blocks}, where the $I_\nu\subseteq[R]$ are pairwise disjoint, $p_\nu=|I_\nu|\ge2$, and $1\le q_\nu<p_\nu$. The empty family is allowed. Its \emph{resulting length} is
\begin{equation}\label{eq:resulting-length}
\rho(\mathfrak D)
=R-\sum_{\nu=1}^{c}(p_\nu-q_\nu).
\end{equation}
A \emph{realization} consists of localized replacements
\begin{equation}\label{eq:local-replacements}
\sum_{r\in I_\nu}P_r
=\sum_{j=1}^{q_\nu}Q_{\nu j},
\qquad\nu\in[c]
\end{equation}
by nonzero nonnegative rank-one tensors $Q_{\nu j}$; $\mathfrak D$ is \emph{realizable} if such tensors exist.
\end{definition}

\begin{lemma}\label{lem:combine}
Every realization of a disjoint shortening pattern $\mathfrak D$ gives a nonnegative decomposition of $\T$ of length $\rho(\mathfrak D)$, namely
\begin{equation}\label{eq:combined-decomposition}
\T
=\sum_{r\notin\bigcup_{\nu}I_\nu}P_r
+\sum_{\nu=1}^{c}\sum_{j=1}^{q_\nu}Q_{\nu j}.
\end{equation}
\end{lemma}

The construction is illustrated in Figure~\ref{fig:pattern} in Appendix~\ref{app:proofs}.

\begin{theorem}[Disjoint localization of nonnegative rank]\label{thm:rank-localization}
Let $\T=\sum_{r=1}^{R}P_r$ be a supplied nonnegative decomposition.  Then
\begin{equation}\label{eq:rank-minimum}
\rankp(\T)
=\min\bigl\{\rho(\mathfrak D):
\mathfrak D\text{ is a realizable disjoint shortening pattern}\bigr\},
\end{equation}
where the empty pattern is realizable and has resulting length $R$.  Moreover, for every minimal decomposition $\T=\sum_{j=1}^{\rho^*}Q_j$ with $\rho^*=\rankp(\T)$, there are a pattern $\mathfrak D$ with $\rho(\mathfrak D)=\rho^*$ and a realization of $\mathfrak D$ such that every localized replacement in \eqref{eq:local-replacements} is irreducible and its right-hand side is a minimal decomposition of $\sum_{r\in I_\nu}P_r$.  The terms $Q_{\nu j}$ of this realization are terms of the given minimal decomposition.
\end{theorem}

\begin{proof}
\pfpart{Every realizable pattern gives an upper bound on nonnegative rank}
By Lemma~\ref{lem:combine}, a realizable pattern $\mathfrak D$ gives a nonnegative decomposition of $\T$ of length $\rho(\mathfrak D)$.  Hence $\rankp(\T)\le \rho(\mathfrak D)$ for every realizable pattern.  The empty pattern is realizable and gives the supplied length $R$, so the minimum in \eqref{eq:rank-minimum} is taken over a nonempty finite set of patterns and is at least $\rankp(\T)$.

\pfpart{A minimal decomposition gives a realizable pattern with resulting length $\rho^*$}
Choose a minimal decomposition $\T=\sum_{j=1}^{\rho^*}Q_j$ and use the blocks \eqref{eq:blocks-minimal} of Proposition~\ref{prop:blocks}.  Let $S\subseteq[c]$ be the set of indices with $q_\nu<p_\nu$.  The sets $I_\nu$, $\nu\in S$, are pairwise disjoint and each has size $p_\nu\ge2$.  Relabeling the right-hand terms of each such block from $1$ to $q_\nu$ gives a realization of the disjoint shortening pattern $\mathfrak D=\bigl((I_\nu,q_\nu)\bigr)_{\nu\in S}$.  Every localized replacement of this realization is irreducible, and its right-hand side is a minimal decomposition of $\sum_{r\in I_\nu}P_r$ by Proposition~\ref{prop:blocks}\ref{item:blocks-rank}.  Since $p_\nu=q_\nu$ for $\nu\notin S$, Proposition~\ref{prop:blocks}\ref{item:blocks-count} gives
\[
\rho(\mathfrak D)
=R-\sum_{\nu\in S}(p_\nu-q_\nu)
=R-\sum_{\nu=1}^{c}(p_\nu-q_\nu)
=\rho^*.
\]
Thus the minimum in \eqref{eq:rank-minimum} is at most $\rho^*=\rankp(\T)$.  Together with the first part, equality follows.
\end{proof}

\section{Effective order of irreducible positive exchanges}\label{sec:budgets}

Section~\ref{sec:localization} shows that every obstruction to minimality or essential uniqueness is witnessed by a localized replacement, and that such a replacement decomposes into irreducible blocks.  This section bounds how many modes an irreducible block can actually use.  The bound depends only on the numbers of terms in the block, not on the order of the tensor.

Throughout this section we consider an irreducible positive exchange
\begin{equation}\label{eq:block-exchange}
\sum_{r\in I}P_r=\sum_{j\in J}Q_j
\end{equation}
between supplied terms $P_r$, $r\in I$, and nonzero nonnegative rank-one tensors $Q_j=b_j^{(1)}\otimes\cdots\otimes b_j^{(d)}$, $j\in J$.  We write $p=|I|$ and $q=|J|$.  Every irreducible localized replacement is of this form.  For every mode $m$ let
\[
U_m=\spanof\{a_r^{(m)}:r\in I\},
\qquad
V_m=\spanof\{b_j^{(m)}:j\in J\}
\]
be the mode-$m$ factor spans of the two sides of \eqref{eq:block-exchange}, and put
\begin{equation}\label{eq:effective-order}
\delta_m=\dim(U_m+V_m),
\qquad
\kappa=\sum_{m=1}^{d}(\delta_m-1),
\qquad
H=\{m:\delta_m>1\}.
\end{equation}
We call $\kappa$ the \emph{effective order} and $H$ the \emph{active mode set} of the exchange.  For a mode $m\notin H$ we have $\delta_m=1$, and all mode-$m$ factors of both sides are positive multiples of one supplied factor.  Theorem~\ref{thm:budget} below bounds $\kappa\le p+q-2$.  It is the dimension count behind the positive splitting inequality of \citet[Theorem~4]{WangYuan2026}, applied to the joint spans $U_m+V_m$ and stated without the scattering term.  Consequently, the exchange has at most $p+q-2$ active modes, and $\sum_{m\in H}\delta_m\le2(p+q-2)$.  These bounds do not depend on the order $d$.

The count rests on the splitting theorem of \citet{LovitzPetrov2023}, which generalizes the theorem of \citet{Kruskal1977}.  We recall the notion of splitting and the theorem in the form used in \citet[Definition~2 and Theorem~2]{WangYuan2026}.

\begin{definition}\label{def:splitting}
A finite multiset $\mathcal E$ of nonzero vectors splits if it has a nonempty proper submultiset $\mathcal F$ such that
\[
\spanof\mathcal E
=
\spanof\mathcal F\oplus\spanof(\mathcal E\setminus\mathcal F).
\]
If no such $\mathcal F$ exists, $\mathcal E$ is called connected.
\end{definition}

\begin{theorem}[Lovitz--Petrov splitting theorem]\label{thm:LP-splitting}
Let $\mathcal E$ be a finite multiset of nonzero product tensors $x_1\otimes\cdots\otimes x_d$ over $\R$, and for each mode $m$ let $r_m$ be the dimension of the span of the mode-$m$ factors $x_m$ of the members of $\mathcal E$.  If
\[
\dim\spanof\mathcal E
\le
\sum_{m=1}^{d}(r_m-1),
\]
then $\mathcal E$ splits.
\end{theorem}

We apply Theorem~\ref{thm:LP-splitting} through its contrapositive to the signed multiset
\[
\mathcal E=\{P_r:r\in I\}\cup\{-Q_j:j\in J\}
\]
of the exchange \eqref{eq:block-exchange}, where the minus sign is absorbed into one factor of each $Q_j$, so that every member is a real product tensor.

\begin{theorem}[Effective-order bound]\label{thm:budget}
For the irreducible positive exchange \eqref{eq:block-exchange}, the following two statements hold.
\begin{enumerate}[label=\textup{(\roman*)}]
\item\label{item:budget-effective-order}
\begin{equation}\label{eq:master-joint-budget}
\kappa\le p+q-2.
\end{equation}
\item
\begin{equation}\label{eq:master-active-budget}
|H|\le\kappa,
\qquad
\sum_{m\in H}\delta_m\le2\kappa.
\end{equation}
\end{enumerate}
\end{theorem}

\begin{proof}
Steps~\ref{step:budget-connected} and~\ref{step:budget-bound} are Steps~1 and~3 in the proof of Theorem~4 of \citet{WangYuan2026}, with the scattering term omitted.  That paper states the resulting inequality only inside its proof, so we include the short argument.

\stepcounter{proofstepblock}
\proofstep[step:budget-connected]{irreducibility implies that $\mathcal E$ is connected}
Suppose $\spanof\mathcal E=\spanof\mathcal E_1\oplus\spanof\mathcal E_2$ for a partition $\mathcal E=\mathcal E_1\sqcup\mathcal E_2$ into nonempty parts.  The sum of all elements of $\mathcal E$ is zero, so the sums of the elements of $\mathcal E_1$ and of $\mathcal E_2$ add to zero, and directness forces each to vanish.  Each part contains terms from both sides of the exchange, because a nonempty sum of nonzero nonnegative tensors does not vanish.  Hence either part is a nonempty proper subexchange, contradicting irreducibility.

\proofstep[step:budget-bound]{the bound \eqref{eq:master-joint-budget}}
Let $\mathscr U=\spanof\{P_r:r\in I\}$, $\mathscr V=\spanof\{Q_j:j\in J\}$, and $h=\dim(\mathscr U\cap\mathscr V)$.  The mode-$m$ factor span of $\mathcal E$ is $U_m+V_m$ and therefore has dimension $\delta_m$, because the multiplication of one factor of each $Q_j$ by $-1$ does not change factor spans.  The total span of $\mathcal E$ is $\mathscr U+\mathscr V$.  Since $\mathcal E$ is connected, the contrapositive of Theorem~\ref{thm:LP-splitting}, together with integrality of the dimensions, gives
\[
\sum_{m=1}^{d}(\delta_m-1)
\le
\dim(\mathscr U+\mathscr V)-1.
\]
The left side is $\kappa$, while
\[
\dim(\mathscr U+\mathscr V)
=
\dim\mathscr U+\dim\mathscr V-h.
\]
Hence $\kappa\le\dim\mathscr U+\dim\mathscr V-h-1$.  The common block sum is a nonzero tensor belonging to $\mathscr U\cap\mathscr V$, so $h\ge1$.  Moreover
\[
\dim\mathscr U\le |I|,
\qquad
\dim\mathscr V\le |J|.
\]
Substitution gives \eqref{eq:master-joint-budget}.

\proofstep{active modes}
Every $m\in H$ satisfies $\delta_m-1\ge1$, so
\[
|H|
\le
\sum_{m\in H}(\delta_m-1)
\le
\kappa.
\]
Since the inactive modes contribute zero to $\kappa$,
\[
\sum_{m\in H}\delta_m
=
\sum_{m\in H}(\delta_m-1)+|H|
=
\kappa+|H|
\le2\kappa.
\]
This proves \eqref{eq:master-active-budget}.
\end{proof}

Theorem~\ref{thm:budget} is the step that removes the ambient tensor order from the subsequent variable count: every irreducible obstruction has a core whose total effective dimension is controlled only by the numbers of terms on its two sides.  Irreducibility is essential: for a reducible exchange the effective order is not bounded by any function of $p$ and $q$ (Example~\ref{ex:reducible} in Appendix~\ref{app:necessity}).
Although the dimension inequality used in Theorem~\ref{thm:budget} is implicit in the proof of \citet{WangYuan2026}, the present paper isolates it as an independent structural result and uses it in a different role: it bounds the dimension of every irreducible obstruction searched in complete certification. Combined with localization and the profile representation below, it yields a bounded-dimensional exact decision-and-search procedure, extending the earlier sufficient criteria to complete certification. The sharpness results in Section~\ref{sec:sharpness} show that this dimension bound is attained by explicit families.

\begin{remark}\label{rem:scattering}
Let $\sigma_m=\dim U_m$, and let $\tau(I)$ be the positive scattering term of the supplied terms $P_r$, $r\in I$, defined in \citet[Definition~4]{WangYuan2026}.  By Step~2 in the proof of Theorem~4 of \citet{WangYuan2026}, the irreducible exchange \eqref{eq:block-exchange} satisfies
\[
\sum_{m=1}^{d}(\delta_m-\sigma_m)\ge\tau(I).
\]
Since $\kappa=\beta(I)+\sum_m(\delta_m-\sigma_m)$ with $\beta(I)=\sum_m(\sigma_m-1)$, it gives $\kappa\ge\beta(I)+\tau(I)$.  Together with \eqref{eq:master-joint-budget}, this shows that there is no irreducible positive exchange \eqref{eq:block-exchange} with $p+q-2<\beta(I)+\tau(I)$.  This is the criterion of \citet{WangYuan2026}, see Section~\ref{sec:discussion}.
\end{remark}

\section{Exact reduction to bounded-dimensional semialgebraic geometry}\label{sec:normalized}

In this section we fix the set $I$ of supplied terms and the number $q$ of competing terms, and we describe the localized replacements on $I$ with $q$ terms by polynomial systems in the sense of Definition~\ref{def:polysystem} below.  The number of variables in these systems does not depend on $d$ or on $n_1,\ldots,n_d$.  The construction has three parts.  In every mode $m$ with $\delta_m>1$, only $\delta_m$ rows of the competing factor matrix are variables, and every other row is a linear combination of these rows (Figure~\ref{fig:rows}).  The modes with $\delta_m=1$ contribute only scalars (Figure~\ref{fig:core}).  Finally, every column is normalized to coordinate sum one.  Section~\ref{subsec:profiles} defines profiles and the core, and Remark~\ref{rem:why-added-rows} explains why a profile must allow added rows.  Section~\ref{subsec:derivation} derives the conditions which an irreducible localized replacement satisfies, and these conditions define the system in Section~\ref{subsec:system}.  Section~\ref{subsec:equivalence} proves that the system is equivalent to the existence of a localized replacement (Theorem~\ref{thm:representation}).  The system still contains one coefficient row for each ambient row, and Section~\ref{sec:lifting} eliminates these variables.

\paragraph{Polynomial systems.} All systems in Sections~\ref{sec:normalized} and~\ref{sec:one-shot} are of the following kind.
\begin{definition}\label{def:polysystem}
Let $F=(F_1,\ldots,F_s)$ be a family of polynomials in $\mathbb Z[y_1,\ldots,y_k]$.
\begin{enumerate}[label=\textup{(\roman*)}]
\item A \emph{basic condition} on $y$ is a finite system of equations $F_i=0$ and of inequalities $F_i>0$, $F_i\ge0$, $F_i<0$, or $F_i\le0$.
\item A \emph{polynomial system} is a finite list of conditions on $y$.  Every condition of the list is either a basic condition or an implication between two basic conditions $\mathsf A$ and $\mathsf B$, i.e. a condition which holds at $y$ if $\mathsf A$ does not hold at $y$ or $\mathsf B$ holds at $y$.
\item A \emph{solution} of a polynomial system is a point $y\in\R^k$ at which all conditions of the list hold.
\item The \emph{length} of a polynomial system is the total number of equations and inequalities which occur in the list.
\end{enumerate}
\end{definition}

A system whose polynomials have rational coefficients becomes a polynomial system after every polynomial has been multiplied by a positive common denominator of its coefficients. A solution of a polynomial system can have irrational coordinates, see Section~\ref{sec:model}.  Exact representations of such solutions are defined in Section~\ref{subsec:exact-computation}.

Let $I\subseteq[R]$ be nonempty, write $p=|I|$, fix an ordering of $I$, and put
\[
A_I^{(m)}=[a_r^{(m)}:r\in I]\in\Qpos^{n_m\times p}.
\]
Fix an integer $q\ge1$.

\subsection{Profiles, selected rows, and the core}\label{subsec:profiles}

\begin{definition}\label{def:profile}
Fix $I\subseteq[R]$, $p=|I|$, and $q\ge1$.  For each mode put
\[
\sigma_m=\rank A_I^{(m)}
\]
and fix one row basis $S_m^{0}\subseteq[n_m]$ of $A_I^{(m)}$.  To make the choice unique, we take the first row basis in the lexicographic order.  A \emph{profile} of the pair $(I,q)$ consists of integers $\delta_m$, $m\in[d]$, and, for every mode with $\delta_m>1$, of a row set $K_m$ satisfying
\begin{align}
&\sigma_m\le\delta_m\le\min\{\sigma_m+q,\,n_m\},
\label{eq:profile-dimension-range}\\
&\kappa:=\sum_{m=1}^{d}(\delta_m-1)\le p+q-2,
\label{eq:profile-budget}\\
&K_m\subseteq[n_m]\setminus S_m^{0},\qquad
|K_m|=\delta_m-\sigma_m.
\label{eq:profile-rows}
\end{align}
The modes in $H=\{m:\delta_m>1\}$ are called \emph{active}, and the other modes are called \emph{inactive}.  For $m\in H$ put
\[
S_m=S_m^{0}\cup K_m,
\qquad
X_m=A_I^{(m)}[S_m,:]\in\Qpos^{\delta_m\times p}.
\]
The rows in $S_m$ are the \emph{selected rows} of mode $m$, and those in $[n_m]\setminus S_m$ are its \emph{omitted rows}.  For $m\notin H$ we have $\delta_m=1$ by the definition of $H$, and hence $\sigma_m=1$ by \eqref{eq:profile-dimension-range}.  In such a mode choose a fixed reference index $r_0\in I$, put $u_m=a_{r_0}^{(m)}$, and write uniquely
\begin{equation}\label{eq:inactive-ratios}
a_r^{(m)}=c_{mr}u_m,
\qquad r\in I,
\end{equation}
where every $c_{mr}$ is a positive rational number.
\end{definition}

\begin{remark}\label{rem:tau-profile}
By Remark~\ref{rem:scattering}, the condition $\sum_{m\in H}|K_m|\ge\tau(I)$ can be added to \eqref{eq:profile-dimension-range}--\eqref{eq:profile-rows} without the loss of any irreducible localized replacement, because $|K_m|=\delta_m-\sigma_m$.  This reduces the number of profiles.
\end{remark}

\begin{remark}[Why added rows are needed]\label{rem:why-added-rows}
The rows $K_m$ cannot be omitted in general, because the competing factors need not lie in the column space of the supplied factor matrices.  For the row decomposition of the cycle matrix $C_R$ with $R\ge4$ even, every term-disjoint balanced replacement has $|K_2|=1$, and a search restricted to $K_m=\varnothing$ would miss the nonuniqueness of the supplied decomposition (Corollary~\ref{cor:added-rows} and Remark~\ref{rem:compression} in Appendix~\ref{app:necessity}).
\end{remark}

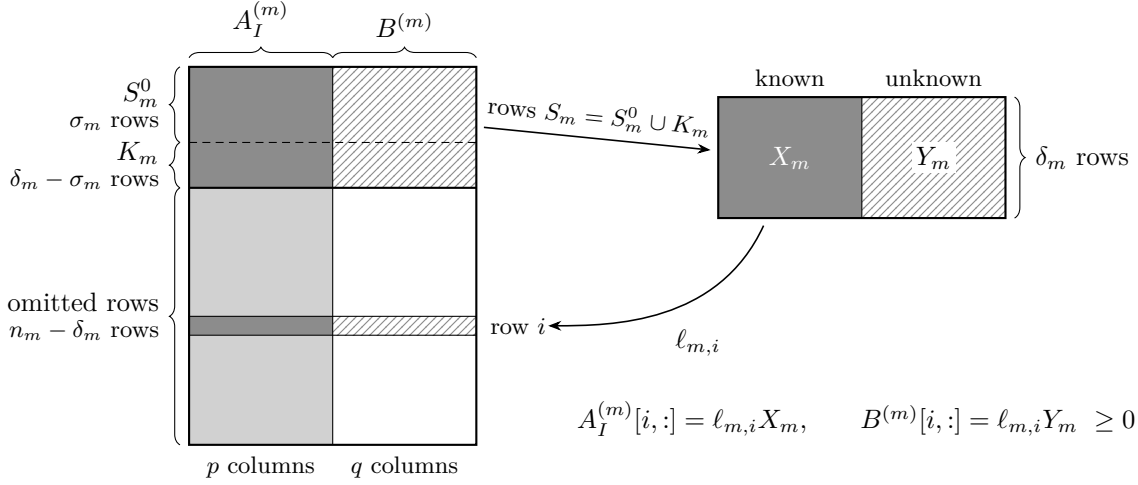
\begin{figure}[t]
\centering
\begin{tikzpicture}[x=1cm,y=1cm]
\path[known]   (0,0) rectangle (1.9,5.0);
\path[knownsel](0,3.4) rectangle (1.9,5.0);
\path[unknown] (1.9,3.4) rectangle (3.8,5.0);
\path[knownsel](0,1.45) rectangle (1.9,1.70);
\path[unknown] (1.9,1.45) rectangle (3.8,1.70);
\draw (0,1.45) rectangle (3.8,1.70);
\draw[thick] (0,0) rectangle (3.8,5.0);
\draw (1.9,0) -- (1.9,5.0);
\draw[densely dashed] (0,4.0) -- (3.8,4.0);
\draw[thick] (0,3.4) -- (3.8,3.4);
\draw[brace] (0,5.12) -- node[above=4pt,lab]{$A_I^{(m)}$} (1.9,5.12);
\draw[brace] (1.9,5.12) -- node[above=4pt,lab]{$B^{(m)}$} (3.8,5.12);
\node[slab] at (0.95,-0.3) {$p$ columns};
\node[slab] at (2.85,-0.3) {$q$ columns};
\draw[brace] (-0.12,4.0) -- node[left=4pt,lab,align=right]{$S_m^{0}$\\[-2pt]{\footnotesize $\sigma_m$ rows}} (-0.12,5.0);
\draw[brace] (-0.12,3.4) -- node[left=4pt,lab,align=right]{$K_m$\\[-2pt]{\footnotesize $\delta_m-\sigma_m$ rows}} (-0.12,4.0);
\draw[brace] (-0.12,0) -- node[left=4pt,lab,align=right]{omitted rows\\[-2pt]{\footnotesize $n_m-\delta_m$ rows}} (-0.12,3.4);
\node[slab,anchor=west] at (3.85,1.575) {row $i$};
\begin{scope}[shift={(7.0,3.0)}]
\path[knownsel](0,0) rectangle (1.9,1.6);
\path[unknown] (1.9,0) rectangle (3.8,1.6);
\draw[thick] (0,0) rectangle (3.8,1.6);
\draw (1.9,0) -- (1.9,1.6);
\node[lab,fill=black!45,inner sep=1.5pt] at (0.95,0.8) {\color{white}$X_m$};
\node[lab,fill=white,inner sep=1.5pt] at (2.85,0.8) {$Y_m$};
\draw[brace] (3.92,1.6) -- node[right=4pt,lab]{$\delta_m$ rows} (3.92,0);
\node[slab] at (0.95,1.85) {known};
\node[slab] at (2.85,1.85) {unknown};
\end{scope}
\draw[arr] (3.9,4.2) -- node[above,slab,sloped]{rows $S_m=S_m^{0}\cup K_m$} (6.9,3.9);
\draw[arr] (7.6,2.9) .. controls (7.0,1.6) and (5.6,1.575) .. (4.75,1.575);
\node[lab,anchor=west] at (6.3,1.35) {$\ell_{m,i}$};
\node[lab,anchor=west,align=left] at (5.0,0.35)
{$A_I^{(m)}[i,:]=\ell_{m,i}X_m$,\qquad $B^{(m)}[i,:]=\ell_{m,i}Y_m\ \ge 0$};
\end{tikzpicture}
\caption{One active mode $m$ of a localized replacement whose competing terms have the factor matrix $B^{(m)}$.  The left block of $[\,A_I^{(m)}\ \,B^{(m)}\,]$ is known, and the right block is unknown.  The rows in $S_m$ form a row basis, and every omitted row $i$ is the combination $\ell_{m,i}[\,X_m\ \,Y_m\,]$ of the selected rows, where $\ell_{m,i}Y_m\ge0$.}
\label{fig:rows}
\end{figure}

\begin{lemma}\label{lem:core}
For a profile of $(I,q)$, every column of $X_m$, $m\in H$, is nonzero.  Define
\[
s_{mr}=\mathbf1^{\mathsf T}X_m[:,r]>0,
\qquad
\bar X_m[:,r]=s_{mr}^{-1}X_m[:,r]
\]
and
\begin{equation}\label{eq:supplied-weight}
\alpha_r
=
\left(\prod_{m\notin H}c_{mr}\right)
\left(\prod_{m\in H}s_{mr}\right)>0.
\end{equation}
Then all $\bar X_m$ and all $\alpha_r$ are rational, and they are computed from the input.  The number $\Omega_I=\sum_{r\in I}\alpha_r$ is positive.
\end{lemma}

\begin{proof}
The rows indexed by $S_m^{0}$ span every row of $A_I^{(m)}$.  If the restriction of some supplied column to those rows were zero, every row of the full matrix would have zero in that column, contradicting the assumption that every supplied factor is nonzero.  The same column remains nonzero after the additional rows in $K_m$ are appended.  Since the column is nonnegative, its sum $s_{mr}$ over the selected rows is strictly positive.

The matrices $X_m$ have rational entries, so the $s_{mr}$ and the normalized matrices $\bar X_m$ are rational.  In an inactive mode, the one-dimensional space $\col A_I^{(m)}$ is generated by the nonzero nonnegative rational vector $u_m$.  Every other nonzero nonnegative rational column on that line is a unique positive rational multiple of $u_m$, proving that the numbers $c_{mr}$ in \eqref{eq:inactive-ratios} are positive and rational.  Formula \eqref{eq:supplied-weight} and the positivity of its factors prove the remaining assertions.
\end{proof}

We call the tensor
\begin{equation}\label{eq:core}
\sum_{r\in I}\alpha_r\bigotimes_{m\in H}\bar X_m[:,r]
\ \in\ \bigotimes_{m\in H}\R^{\delta_m}
\end{equation}
the \emph{core} of $\sum_{r\in I}P_r$ for the given profile. If $H=\varnothing$, the empty tensor product is interpreted as the scalar one.  By \eqref{eq:inactive-ratios} and \eqref{eq:supplied-weight}, the restriction of $\sum_{r\in I}P_r$ to the selected rows of the active modes is the tensor product of the core and of the vectors $u_m$, $m\notin H$, see Figure~\ref{fig:core}.  The core has $\prod_{m\in H}\delta_m$ entries.  This number is bounded by a function of $p$ and $q$, because $\sum_{m\in H}\delta_m=|H|+\kappa\le2\kappa\le2(p+q-2)$.

\subsection{The equations of a localized replacement}\label{subsec:derivation}

Let \eqref{eq:localized-replacement} be an irreducible localized replacement on $I$ with $q$ terms.  Write
\[
Q_j=b_j^{(1)}\otimes\cdots\otimes b_j^{(d)},
\qquad
B^{(m)}=[b_1^{(m)},\ldots,b_q^{(m)}].
\]
For each mode let
\[
U_m=\col A_I^{(m)},\qquad
V_m=\col B^{(m)},\qquad
\delta_m=\dim(U_m+V_m).
\]
Theorem~\ref{thm:budget}\ref{item:budget-effective-order} gives
\[
\kappa=\sum_m(\delta_m-1)\le p+q-2,
\]
which is \eqref{eq:profile-budget}.  The $q$ columns of $B^{(m)}$ increase the rank by at most $q$, so $\sigma_m\le\delta_m\le\sigma_m+q$.  Since $\delta_m\le n_m$, this gives \eqref{eq:profile-dimension-range}.

If $\delta_m>1$, extend the fixed row basis $S_m^{0}$ of $A_I^{(m)}$ to a row basis $S_m=S_m^{0}\cup K_m$ of the joint matrix $[\,A_I^{(m)}\ B^{(m)}\,]$ (Figure~\ref{fig:rows}).  The number of added rows is $\delta_m-\sigma_m$, so the resulting data satisfy \eqref{eq:profile-rows}.  Hence the integers $\delta_m$ and the row sets $K_m$ form a profile of $(I,q)$.  Below, $H$, $S_m$, $X_m$, $s_{mr}$, $\bar X_m$, $u_m$, $c_{mr}$, $\alpha_r$, and $\Omega_I$ are the data of this profile.  Put
\[
Y_m=B^{(m)}[S_m,:].
\]
Every column of $X_m$ and $Y_m$ is nonzero.  Indeed, the selected rows span every row of the joint matrix.  Therefore, a column which vanishes on all selected rows would vanish on every row, and this contradicts the assumption that all factors are nonzero.

We use the following facts on the normalization of columns.  Let $X\in\Rpos^{\delta\times p}$ and $Y\in\Rpos^{\delta\times q}$ have no zero column.  Let $D_X$ and $D_Y$ be the diagonal matrices of the column sums $s_r=\mathbf1^{\mathsf T}X[:,r]$ and $t_j=\mathbf1^{\mathsf T}Y[:,j]$, which are positive, and put $\bar X=XD_X^{-1}$ and $\bar Y=YD_Y^{-1}$.  Then every column of $\bar X$ and $\bar Y$ is nonnegative and has coordinate sum one.  Moreover $\rank[\,X\ Y\,]=\rank[\,\bar X\ \bar Y\,]$, and for every row $a\in\R^{1\times p}$ and every row vector $\ell$,
\begin{equation}\label{eq:normalized-rows}
\ell X=a\ \text{ and }\ \ell Y\ge0
\quad\Longleftrightarrow\quad
\ell\bar X=aD_X^{-1}\ \text{ and }\ \ell\bar Y\ge0,
\end{equation}
because $D_X$ and $D_Y$ are diagonal matrices with positive diagonal entries. 

We apply this normalization to $X_m$ and $Y_m$ and write
\[
\bar X_m=X_mD_{X,m}^{-1},
\qquad
\bar Y_m=Y_mD_{Y,m}^{-1},
\]
where the diagonal entries of $D_{X,m}$ and $D_{Y,m}$ are the column sums $s_{mr}$ and $t_{mj}$, respectively.  Then the columns of the normalized matrices are nonnegative and have sum one, and $\rank[\,\bar X_m\ \bar Y_m\,]=\delta_m$.  By Lemma~\ref{lem:product-equality}, the factors of a term $Q_j$ are unique only up to positive scalars, and the normalization chooses these scalars.

If $\delta_m=1$, every supplied and every competing factor in this mode is a positive multiple of the same vector.  With the reference vector $u_m$ of Definition~\ref{def:profile}, write
\[
a_r^{(m)}=c_{mr}u_m,
\qquad
b_j^{(m)}=z_{mj}u_m,
\qquad c_{mr},z_{mj}>0.
\]
Let $\alpha_r$ be as in \eqref{eq:supplied-weight}, and define the weights
\begin{equation}\label{eq:actual-weights}
\omega_j
=
\left(\prod_{m\notin H}z_{mj}\right)
\left(\prod_{m\in H}t_{mj}\right).
\end{equation}
We restrict \eqref{eq:localized-replacement} to the selected rows of the active modes (Figure~\ref{fig:core}).  Both sides of the restricted equation are tensor products with the nonzero tensor $\bigotimes_{m\notin H}u_m$, and we remove this common factor.  After we insert $X_m[:,r]=s_{mr}\bar X_m[:,r]$ and $Y_m[:,j]=t_{mj}\bar Y_m[:,j]$, we obtain the \emph{core equation}
\begin{equation}\label{eq:core-equation}
\sum_{r\in I}\alpha_r
\bigotimes_{m\in H}\bar X_m[:,r]
=
\sum_{j=1}^{q}\omega_j
\bigotimes_{m\in H}\bar Y_m[:,j]
\end{equation}
with the weights \eqref{eq:actual-weights}.  Taking the total entry sum of that equation gives
\[
\sum_j\omega_j
=\sum_r\alpha_r
=\Omega_I.
\]
All weights $\omega_j$ are strictly positive.

Every omitted row of the joint matrix is a linear combination of the selected rows.  If $\ell_{m,i}$ is the row of its coefficients, then
\[
\ell_{m,i}X_m=A_I^{(m)}[i,:],
\qquad
\ell_{m,i}Y_m=B^{(m)}[i,:]\ge0.
\]
Put
\[
\bar a_{m,i}=A_I^{(m)}[i,:]
\operatorname{diag}(s_{mr}:r\in I)^{-1}.
\]
By \eqref{eq:normalized-rows}, the same row $\ell_{m,i}$ satisfies
\begin{equation}\label{eq:row-lift}
\ell_{m,i}\bar X_m=\bar a_{m,i},
\qquad
\ell_{m,i}\bar Y_m\ge0.
\end{equation}

Thus the matrices $\bar Y_m$ have nonnegative columns with sum one, the weights $\omega_j$ are positive with sum $\Omega_I$, the core equation \eqref{eq:core-equation} holds, and for every $m\in H$ and every omitted row $i\notin S_m$ there is a row vector $\ell_{m,i}$ with \eqref{eq:row-lift}.

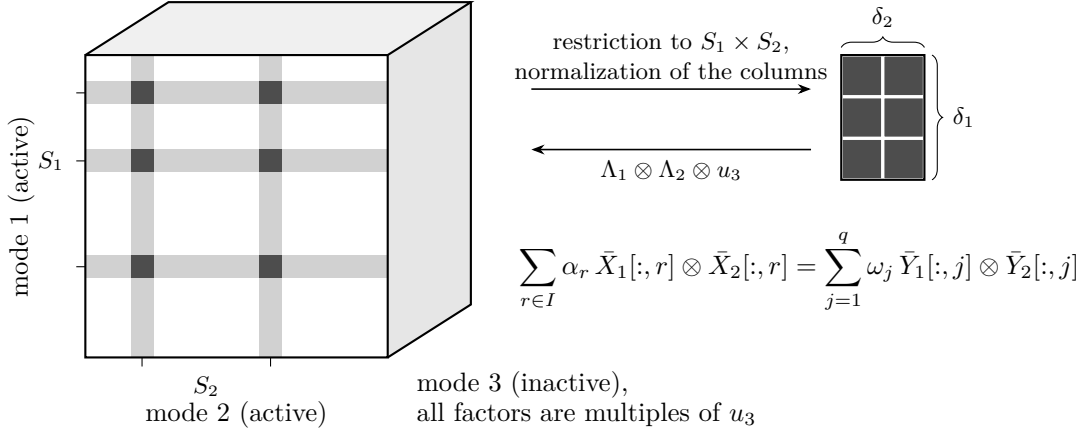
\begin{figure}[t]
\centering
\begin{tikzpicture}[x=1cm,y=1cm]
\def\dx{1.1}\def\dy{0.7}
\path[fill=black!6] (0,4) -- ++(\dx,\dy) -- ++(4,0) -- (4,4) -- cycle;
\path[fill=black!10] (4,0) -- ++(\dx,\dy) -- ++(0,4) -- (4,4) -- cycle;
\path[fill=white] (0,0) rectangle (4,4);
\foreach \y in {3.35,2.45,1.05}{ \path[fill=black!18] (0,\y) rectangle (4,\y+0.3); }
\foreach \x in {0.6,2.3}{ \path[fill=black!18] (\x,0) rectangle (\x+0.3,4); }
\foreach \y in {3.35,2.45,1.05}{ \foreach \x in {0.6,2.3}{ \path[fill=black!70] (\x,\y) rectangle (\x+0.3,\y+0.3); } }
\draw[thick] (0,0) rectangle (4,4);
\draw[thick] (0,4) -- ++(\dx,\dy) -- ++(4,0) -- (4,4);
\draw[thick] (4,0) -- ++(\dx,\dy) -- ++(0,4);
\node[lab,rotate=90] at (-0.85,2) {mode $1$ (active)};
\node[lab] at (2,-0.75) {mode $2$ (active)};
\node[lab,anchor=north west,align=left] at (4.25,-0.05) {mode $3$ (inactive),\\[-1pt] all factors are multiples of $u_3$};
\foreach \y in {3.5,2.6,1.2}{ \draw (-0.12,\y) -- (0,\y); }
\node[slab,anchor=east] at (-0.15,2.6) {$S_1$};
\foreach \x in {0.75,2.45}{ \draw (\x,-0.12) -- (\x,0); }
\node[slab,anchor=north] at (1.6,-0.12) {$S_2$};
\begin{scope}[shift={(10.0,2.35)}]
\foreach \i in {0,1,2}{ \foreach \j in {0,1}{ \path[fill=black!70,draw=white,line width=1.2pt] (0.55*\j,0.55*\i) rectangle ++(0.55,0.55); } }
\draw[thick] (0,0) rectangle (1.1,1.65);
\draw[brace] (1.2,1.65) -- node[right=4pt,slab]{$\delta_1$} (1.2,0);
\draw[brace] (0,1.75) -- node[above=4pt,slab]{$\delta_2$} (1.1,1.75);
\end{scope}
\draw[arr] (5.9,3.55) -- node[above,slab,align=center]{restriction to $S_1\times S_2$,\\ normalization of the columns} (9.6,3.55);
\draw[arr] (9.6,2.75) -- node[below,slab]{$\Lambda_1\otimes\Lambda_2\otimes u_3$} (5.9,2.75);
\node[lab,anchor=west] at (5.6,1.15)
{$\displaystyle\sum_{r\in I}\alpha_r\,\bar X_1[:,r]\otimes\bar X_2[:,r]
=\sum_{j=1}^{q}\omega_j\,\bar Y_1[:,j]\otimes\bar Y_2[:,j]$};
\end{tikzpicture}
\caption{The tensor $\sum_{r\in I}P_r$ for $d=3$ and $H=\{1,2\}$.  In mode $3$ all factors are positive multiples of $u_3$, and in modes $1$ and $2$ only the rows $S_1$ and $S_2$ are used.  The dark entries form the core, and $\Lambda_1\otimes\Lambda_2\otimes u_3$ maps a solution of the core equation \ref{cond:core} back to the format $n_1\times n_2\times n_3$, see Theorem~\ref{thm:representation}.}
\label{fig:core}
\end{figure}

\subsection{The system}\label{subsec:system}

The competing terms are unknown.  We now replace the matrices $\bar Y_m$ and the weights $\omega_j$ of Section~\ref{subsec:derivation} by variables, which we denote by the same symbols.

For $m\in H$, introduce $q(\delta_m-1)$ variables $y_{mij}$, $i\in[\delta_m-1]$, $j\in[q]$, and define the matrix $\bar Y_m$ with $\delta_m$ rows and $q$ columns by
\begin{equation}\label{eq:simplex}
\bar Y_m(i,j)=y_{mij}\quad\text{for $i<\delta_m$},\qquad \text{and}
\qquad
\bar Y_m(\delta_m,j)=1-\sum_{i=1}^{\delta_m-1}y_{mij}.
\end{equation}
Thus the inequalities $\bar Y_m(i,j)\ge0$ for all $i,j$ describe exactly the matrices with nonnegative columns of coordinate sum one.

Introduce $q-1$ weight variables $\omega_1,\ldots,\omega_{q-1}$ and define
\begin{equation}\label{eq:last-weight}
\omega_q
=\Omega_I-\sum_{j=1}^{q-1}\omega_j.
\end{equation}
When $q=1$, this means simply $\omega_1=\Omega_I$ and no weight variable is introduced.

\begin{definition}\label{def:system}
For a profile of $(I,q)$, the system $\mathcal N_{I,q}$ consists of the following conditions on the variables in \eqref{eq:simplex} and \eqref{eq:last-weight}.
\begin{enumerate}[label=\textup{(N\arabic*)},leftmargin=13mm]
\item\label{cond:simplex}
Every entry of every $\bar Y_m$, $m\in H$, is nonnegative.
\item\label{cond:weights}
$\omega_j>0$ for every $j\in[q]$.
\item\label{cond:core}
The variables satisfy the core equation
\begin{equation*}\tag{\ref{eq:core-equation}}
\sum_{r\in I}\alpha_r
\bigotimes_{m\in H}\bar X_m[:,r]
=
\sum_{j=1}^{q}\omega_j
\bigotimes_{m\in H}\bar Y_m[:,j].
\end{equation*}

\item\label{cond:lifting}
For every $m\in H$ and every omitted row $i\notin S_m$, put
\[
\bar a_{m,i}=A_I^{(m)}[i,:]
\operatorname{diag}(s_{mr}:r\in I)^{-1}.
\]
There is a row vector $\ell_{m,i}\in\R^{1\times\delta_m}$, which we call a \emph{coefficient row}, with
\begin{equation*}\tag{\ref{eq:row-lift}}
\ell_{m,i}\bar X_m=\bar a_{m,i},
\qquad
\ell_{m,i}\bar Y_m\ge0.
\end{equation*}
\end{enumerate}
When $q=p$, the \emph{term-disjoint system} $\mathcal N^{\mathrm{td}}_{I,p}$ consists of the conditions of $\mathcal N_{I,p}$ and of the inequalities
\begin{equation}\label{eq:td}
(\omega_j-\alpha_r)^2
+
\sum_{m\in H}
\|\bar Y_m[:,j]-\bar X_m[:,r]\|_2^2
>0
\end{equation}
for every $r\in I$ and $j\in[p]$.  For $H=\varnothing$, the sum over $m$ is empty.
\end{definition}

Condition \ref{cond:lifting} contains the additional unknowns $\ell_{m,i}$, which are eliminated in Section~\ref{sec:lifting}.  The equation $\rank[\,\bar X_m\ \bar Y_m\,]=\delta_m$ of Section~\ref{subsec:derivation} is not a condition of the system, see Remark~\ref{rem:rank-of-solution}.  Table~\ref{tab:notation} lists the notation of this section.

\begin{table}[t]
\centering
\small
\begin{tabularx}{\textwidth}{@{}lXl@{}}
\toprule
Symbol & Meaning & Status\\
\midrule
$I$, $p=|I|$, $q$ & supplied terms which are replaced, number of competing terms & fixed\\
$A_I^{(m)}$, $\sigma_m$ & supplied factor matrix of mode $m$ with columns in $I$, its rank & input\\
$S_m^{0}$ & fixed row basis of $A_I^{(m)}$ & computed from the input\\
$\delta_m$, $K_m$ & number of selected rows, added rows & profile\\
$H$, $\kappa$ & active modes $\{m:\delta_m>1\}$, effective order $\sum_m(\delta_m-1)$ & given by the profile\\
$S_m=S_m^{0}\cup K_m$ & selected rows of an active mode & given by the profile\\
$X_m$, $\bar X_m$, $s_{mr}$ & $A_I^{(m)}[S_m,:]$, its normalized columns, its column sums & rational data\\
$u_m$, $c_{mr}$ & inactive mode, reference vector and ratios $a_r^{(m)}=c_{mr}u_m$ & rational data\\
$\alpha_r$, $\Omega_I$ & supplied weights and their sum & rational data\\
$\bar Y_m$, $\omega_j$ & normalized competing columns on $S_m$, competing weights & variables\\
$\ell_{m,i}$, $\Lambda_m$ & coefficient row of an omitted row $i$, coefficient matrix & eliminated in Section~\ref{sec:lifting}\\
\bottomrule
\end{tabularx}
\caption{Notation of Section~\ref{sec:normalized} for a fixed pair $(I,q)$ and a fixed profile.}
\label{tab:notation}
\end{table}

\subsection{Equivalence with localized replacements}\label{subsec:equivalence}

\begin{theorem}[Polynomial representation of localized replacements]\label{thm:representation}
Fix a nonempty set $I\subseteq[R]$, $p=|I|$, and $q\ge1$.
\begin{enumerate}[label=\textup{(\roman*)}]
\item\label{item:completeness}
Suppose that
\begin{equation}\label{eq:actual-exchange}
\sum_{r\in I}P_r=\sum_{j=1}^{q}Q_j
\end{equation}
is an irreducible localized replacement.  Then there is a profile of $(I,q)$ with $\kappa\le p+q-2$ such that the system $\mathcal N_{I,q}$ of this profile has a solution.  If $q=p$ and the replacement is term-disjoint, then this solution also satisfies \eqref{eq:td}, i.e. it is a solution of $\mathcal N^{\mathrm{td}}_{I,p}$.
\item\label{item:soundness}
Conversely, for every profile of $(I,q)$, every solution of $\mathcal N_{I,q}$ gives a localized replacement
\begin{equation}\label{eq:lifted-exchange}
\sum_{r\in I}P_r=\sum_{j=1}^{q}\widetilde Q_j
\end{equation}
by nonzero nonnegative rank-one tensors.  If the solution satisfies \eqref{eq:td}, then this replacement is term-disjoint.
\end{enumerate}
\end{theorem}

\begin{proof}
\pfpart{Step 1, every irreducible localized replacement gives a solution}
By Section~\ref{subsec:derivation}, the replacement \eqref{eq:actual-exchange} has a profile of $(I,q)$ with $\kappa\le p+q-2$, and its matrices $\bar Y_m$ and weights $\omega_j$ have the form \eqref{eq:simplex} and \eqref{eq:last-weight} and satisfy \ref{cond:simplex}--\ref{cond:lifting}.  Hence they are a solution of $\mathcal N_{I,q}$.

Assume finally that $q=p$, that the replacement \eqref{eq:actual-exchange} is term-disjoint, and, for contradiction, that some polynomial in \eqref{eq:td} vanishes at this solution, say
\[
\omega_j=\alpha_r,
\qquad
\bar Y_m[:,j]=\bar X_m[:,r],\quad m\in H.
\]
The selected rows $S_m$ form a row basis of the joint matrix, so in every active mode there is one coefficient matrix $\Lambda_m$ with $A_I^{(m)}=\Lambda_mX_m$ and $B^{(m)}=\Lambda_mY_m$.  Then, for every active mode,
\[
b_j^{(m)}
=\Lambda_mY_m[:,j]
=t_{mj}\,\Lambda_m\bar Y_m[:,j]
=t_{mj}\,\Lambda_m\bar X_m[:,r]
=\frac{t_{mj}}{s_{mr}}\,a_r^{(m)},
\]
and for every inactive mode $b_j^{(m)}=z_{mj}u_m=(z_{mj}/c_{mr})\,a_r^{(m)}$.  Multiplying over all modes,
\begin{align*}
Q_j
&=\Biggl(\prod_{m\notin H}\frac{z_{mj}}{c_{mr}}\Biggr)
\Biggl(\prod_{m\in H}\frac{t_{mj}}{s_{mr}}\Biggr)P_r
=\frac{\omega_j}{\alpha_r}\,P_r
=P_r,
\end{align*}
which contradicts the assumption that the replacement is term-disjoint.  Thus it is also a solution of $\mathcal N^{\mathrm{td}}_{I,p}$.

\pfpart{Step 2, every solution gives a localized replacement}
Fix a solution.  For an active mode $m$, choose one row $\ell_{m,i}$ satisfying \eqref{eq:row-lift} for every omitted row $i\notin S_m$.  Let $\Lambda_m\in\R^{n_m\times\delta_m}$ be the matrix whose rows on $S_m$ are the rows of the identity matrix and whose row $i\notin S_m$ is $\ell_{m,i}$.  We call $\Lambda_m$ the \emph{coefficient matrix} of the mode $m$.  Then
\begin{equation}\label{eq:normalized-supplied-factor}
\Lambda_m\bar X_m
=A_I^{(m)}\operatorname{diag}(s_{mr}:r\in I)^{-1}
=:\bar A_m,
\end{equation}
and
\begin{equation}\label{eq:normalized-competing-factor}
\bar B_m:=\Lambda_m\bar Y_m\ge0.
\end{equation}
Every column of $\bar B_m$ is nonzero because its restriction to the selected rows is a column of $\bar Y_m$, whose coordinate sum is one. For an inactive mode we keep the fixed nonzero nonnegative vector $u_m$.  Let
\[
u_{H^c}=\bigotimes_{m\notin H}u_m.
\]
Apply the linear map
\[
\Lambda_H:
\bigotimes_{m\in H}\R^{\delta_m}
\longrightarrow
\bigotimes_{m=1}^{d}\R^{n_m},
\qquad
Z\longmapsto
\left(\bigotimes_{m\in H}\Lambda_m\right)Z
\otimes u_{H^c},
\]
with the factors inserted in their original mode positions, to the core equation \eqref{eq:core-equation}.  On its supplied side, the $r$th term becomes
\begin{align*}
&\alpha_r
\left(\bigotimes_{m\in H}\bar A_m[:,r]\right)
\otimes u_{H^c}\\
&\quad=
\left(\prod_{m\notin H}c_{mr}\right)
\left(\prod_{m\in H}s_{mr}\right)
\left(\bigotimes_{m\in H}s_{mr}^{-1}a_r^{(m)}\right)
\otimes
\left(\bigotimes_{m\notin H}u_m\right)\\
&\quad=
\bigotimes_{m=1}^{d}a_r^{(m)}
=P_r.
\end{align*}
On the competing side, define
\begin{equation}\label{eq:reconstructed-term}
\widetilde Q_j
=
\omega_j
\left(\bigotimes_{m\in H}\bar B_m[:,j]\right)
\otimes u_{H^c}.
\end{equation}
Every $\widetilde Q_j$ is a nonzero nonnegative rank-one tensor.  The scalar $\omega_j$ may be absorbed into any one of its nonzero mode factors.  Therefore the image of the core equation under $\Lambda_H$ is exactly \eqref{eq:lifted-exchange}.

\pfpart{Step 3, comparison of terms}
The map $\Lambda_H$ is injective.  If $H\ne\varnothing$, each $\Lambda_m$ is injective because it contains the identity rows on $S_m$; hence their tensor product is injective.  Tensoring with the fixed nonzero tensor $u_{H^c}$ preserves injectivity.  For $H=\varnothing$, $\Lambda_H$ is the map from $\R$ sending a scalar $z$ to $zu_{H^c}$, which is injective because $u_{H^c}\ne0$.

By construction, $P_r$ and $\widetilde Q_j$ are the images under $\Lambda_H$ of the two tensors $\alpha_r\bigotimes_{m\in H}\bar X_m[:,r]$ and $\omega_j\bigotimes_{m\in H}\bar Y_m[:,j]$.  Injectivity gives
\begin{equation}\label{eq:core-term-equivalence}
P_r=\widetilde Q_j
\quad\Longleftrightarrow\quad
\alpha_r\bigotimes_{m\in H}\bar X_m[:,r]
=\omega_j\bigotimes_{m\in H}\bar Y_m[:,j].
\end{equation}

All columns of $\bar X_m$ and $\bar Y_m$ have coordinate sum one.  Hence the two sides of the equation on the right of \eqref{eq:core-term-equivalence} have entry sums $\alpha_r$ and $\omega_j$.  Equality therefore forces $\omega_j=\alpha_r$.  After cancelling this positive scalar, equality of the two nonzero product tensors and Lemma~\ref{lem:product-equality} give positive scalars $\lambda_m$ with
\[
\bar Y_m[:,j]=\lambda_m\bar X_m[:,r],
\quad m\in H,
\qquad
\prod_{m\in H}\lambda_m=1.
\]
Taking coordinate sums gives $1=\lambda_m$ for every active mode.  Thus
\begin{equation}\label{eq:term-equality}
P_r=\widetilde Q_j
\quad\Longleftrightarrow\quad
\omega_j=\alpha_r
\ \text{and}\ 
\bar Y_m[:,j]=\bar X_m[:,r]
\text{ for every }m\in H.
\end{equation}
For $H=\varnothing$, the same conclusion reduces to equality of the two weights.  Consequently the left side of \eqref{eq:td} vanishes exactly when $P_r=\widetilde Q_j$.  If all these expressions are positive, then the replacement \eqref{eq:lifted-exchange} is term-disjoint.  This completes the proof of part \ref{item:soundness}.
\end{proof}

\begin{corollary}\label{cor:criteria}
Let $\T=\sum_{r=1}^{R}P_r$ be the supplied decomposition.
\begin{enumerate}[label=\textup{(\roman*)}]
\item The supplied decomposition is nonminimal if and only if there are a set $I\subseteq[R]$, $p=|I|\ge2$, an integer $q<p$, and a profile of $(I,q)$ for which $\mathcal N_{I,q}$ has a solution.
\item\label{item:criteria-uniqueness} Assume that the supplied decomposition is minimal.  It is not essentially unique if and only if there are a set $I\subseteq[R]$, $p=|I|\ge2$, and a profile of $(I,p)$ for which $\mathcal N^{\mathrm{td}}_{I,p}$ has a solution.
\end{enumerate}
\end{corollary}

\begin{proof}
Suppose first that the supplied decomposition is nonminimal.  By Theorem~\ref{thm:localized-witness}\ref{item:localized-minimality}, it has an irreducible shortening localized replacement on some $I$ with $q<p$ terms.  By Theorem~\ref{thm:representation}\ref{item:completeness}, the system $\mathcal N_{I,q}$ of some profile has a solution.  Conversely, a solution of $\mathcal N_{I,q}$ gives by Theorem~\ref{thm:representation}\ref{item:soundness} a localized replacement with $q<p$ terms, and the supplied decomposition is nonminimal by Theorem~\ref{thm:localized-witness}\ref{item:localized-minimality}.

Now assume that the supplied decomposition is minimal.  If it is not essentially unique, Theorem~\ref{thm:localized-witness}\ref{item:localized-uniqueness} gives an irreducible, term-disjoint, balanced localized replacement on some $I$ with $p\ge2$.  By Theorem~\ref{thm:representation}\ref{item:completeness}, the system $\mathcal N^{\mathrm{td}}_{I,p}$ of some profile has a solution.  Conversely, a solution of $\mathcal N^{\mathrm{td}}_{I,p}$ gives a term-disjoint balanced localized replacement.  Because both sides are nonempty and term-disjoint, their multisets of rank-one terms are unequal; hence the replacement is a localized alternative.  The supplied decomposition is therefore not essentially unique by Theorem~\ref{thm:localized-witness}\ref{item:localized-uniqueness}.
\end{proof}

\begin{remark}\label{rem:rank-of-solution}
A solution of $\mathcal N_{I,q}$ need not satisfy $\rank[\,\bar X_m\ \bar Y_m\,]=\delta_m$.  The integer $\delta_m$ of a profile is only the number of selected rows, and the rank can be smaller.  This does not affect Theorem~\ref{thm:representation}.  Part~\ref{item:soundness} holds for every solution, because its proof uses only the conditions \ref{cond:simplex}--\ref{cond:lifting} and not the rank of $[\,\bar X_m\ \bar Y_m\,]$.  The rank is used only in part~\ref{item:completeness}.  There $S_m$ is a row basis of the joint matrix, and therefore every omitted row is a linear combination of the selected rows.  Thus a profile whose integers $\delta_m$ are larger than necessary remains valid; it merely yields a system with more variables.
\end{remark}

\subsection{Eliminating ambient coordinates}\label{sec:lifting}

Condition \ref{cond:lifting} of the system $\mathcal N_{I,q}$ contains the unknown coefficient rows $\ell_{m,i}$.  Their number grows with the dimensions $n_m$, while the running time in Theorem~\ref{thm:sampling} is exponential in the number of variables.  We therefore replace \ref{cond:lifting} by conditions involving only the matrices $\bar Y_m$.  We use the Farkas lemma (Proposition~\ref{prop:lifting}) and the vectors of Lemma~\ref{lem:cofactor} (Lemma~\ref{lem:circuit-elimination}).  Condition~\ref{cond:lifting} cannot be dropped: in Example~\ref{ex:lifting} of Appendix~\ref{app:necessity} the decomposition is essentially unique, but without this condition Phase~III would report that it is not.

\paragraph{Tools.}
We use three standard facts.  First, a polyhedral cone contained in the nonnegative orthant is pointed, and therefore it is generated by its finitely many extreme rays \citep[Section~8.8]{Schrijver1986}.  Second, we use the following form of the Farkas lemma.

\begin{lemma}[Farkas lemma]\label{lem:farkas}
Let $\Cm\in\R^{e\times q}$ and $g\in\R^{1\times q}$.  Exactly one of the following two statements is true.
\begin{enumerate}[label=\textup{(\alph*)}]
\item\label{item:farkas-primal} There is $z\in\R^{1\times e}$ with $g+z\Cm\ge0$.
\item\label{item:farkas-dual} There is $\eta\in\Rpos^{q}$ with $\Cm\eta=0$ and $g\eta<0$.
\end{enumerate}
\end{lemma}

\begin{proof}
If \ref{item:farkas-primal} and \ref{item:farkas-dual} were both true, then we would have $0\le(g+z\Cm)\eta=g\eta<0$, which is impossible.  If \ref{item:farkas-primal} is false, then \ref{item:farkas-dual} is true by the variant of the Farkas lemma in \citet[Corollary~7.1e]{Schrijver1986}, applied to the system $-\Cm^{\mathsf T}x\le g^{\mathsf T}$ with $x=z^{\mathsf T}$.
\end{proof}

The third fact is the following consequence of the Laplace expansion of determinants \citep[Section~0.3.1]{HornJohnson2013}.

\begin{lemma}\label{lem:cofactor}
Let $N\in\R^{(k-1)\times k}$ have rank $k-1$, and for $j\in[k]$ put $D_j=\det N[:,[k]\setminus\{j\}]$ and $\Delta_j=(-1)^jD_j$.  Then $\Delta=(\Delta_1,\ldots,\Delta_k)^{\mathsf T}$ is nonzero and $\ker N=\R\Delta$.
\end{lemma}

\begin{proof}
Since $N$ has rank $k-1$, its kernel has dimension one and some $D_j$ is nonzero, hence $\Delta\ne0$.  Fix a row index $i$, and let $N'$ be the $k\times k$ matrix which is obtained from $N$ by adding the $i$th row of $N$ as a new first row.  Two rows of $N'$ are equal, so $\det N'=0$.  The Laplace expansion of $\det N'$ along the first row gives $\sum_{j=1}^{k}(-1)^{1+j}N[i,j]D_j=0$, i.e. $\sum_{j=1}^{k}N[i,j]\Delta_j=0$.  This holds for every $i$.  Hence $N\Delta=0$, and $\Delta$ spans $\ker N$.
\end{proof}

\subsubsection{Reduction to the extreme rays of a cone}\label{subsec:cone}

Fix matrices
\[
X\in\R^{\delta\times p},
\qquad
Y\in\R^{\delta\times q},
\]
and let $a\in\R^{1\times p}$ belong to the row space of $X$.  We ask whether there is a row vector $\ell\in\R^{1\times\delta}$ with
\begin{equation}\label{eq:basic-lift}
\ell X=a,
\qquad
\ell Y\ge0.
\end{equation}
Choose one solution $w$ of $wX=a$.  Let the rows of $N$ form a basis of the left kernel of $X$, and put
\begin{equation}\label{eq:g-C}
g=wY\in\R^{1\times q},
\qquad
\Cm=NY\in\R^{e\times q},
\end{equation}
where $e=\delta-\rank X$.  We write
\[
\mathcal K(\Cm)=\ker\Cm\cap\Rpos^q .
\]
This is a polyhedral cone in the nonnegative orthant.

\begin{proposition}\label{prop:lifting}
The following statements hold.
\begin{enumerate}[label=\textup{(\roman*)}]
\item\label{item:lifting-dual} There is a row vector $\ell$ with \eqref{eq:basic-lift} if and only if
\begin{equation}\label{eq:dual-lift}
g\eta\ge0
\quad\text{for every }\eta\in\mathcal K(\Cm).
\end{equation}
\item\label{item:lifting-rays} Condition \eqref{eq:dual-lift} holds if and only if $g\eta\ge0$ for every vector $\eta$ which generates an extreme ray of $\mathcal K(\Cm)$.
\item\label{item:lifting-support} Let $\eta$ generate an extreme ray of $\mathcal K(\Cm)$, and let $J=\supp(\eta)$.  Then $\ker\Cm_J=\R\eta_J$.  In particular, $\rank\Cm_J=|J|-1$ and $|J|\le\rank\Cm+1$.
\end{enumerate}
\end{proposition}

\begin{proof}
Every solution of $\ell X=a$ has the form
\[
\ell=w+zN.
\]
Thus there is $\ell$ with \eqref{eq:basic-lift} if and only if there is $z\in\R^{1\times e}$ with
\begin{equation}\label{eq:z-lift}
g+z\Cm\ge0.
\end{equation}
By Lemma~\ref{lem:farkas}, there is no such $z$ exactly when there is an $\eta\ge0$ with
\[
\Cm\eta=0,
\qquad
g\eta<0.
\]
This proves part \ref{item:lifting-dual}.  For $e=0$ the statement is trivial, because then $\ell=w$ and both conditions say that $g\ge0$.

The cone $\mathcal K=\mathcal K(\Cm)$ is pointed and polyhedral, and it is generated by its extreme rays (see the facts recalled at the beginning of this subsection).  Therefore, a linear functional is nonnegative on $\mathcal K$ exactly when it is nonnegative on its extreme rays.  This proves part \ref{item:lifting-rays}. Let $\eta$ generate an extreme ray and put $J=\supp(\eta)$.  Every entry of $\eta_J$ is positive.  If $\dim\ker \Cm_J\ge2$, choose $v\in\ker \Cm_J$ not proportional to $\eta_J$.  For sufficiently small $\varepsilon>0$, both $\eta_J\pm\varepsilon v$ remain nonnegative, and their zero extensions give two nonproportional vectors of $\mathcal K$ whose average is $\eta$, contradicting extremality.  Hence $\dim\ker \Cm_J=1$, so $\ker\Cm_J=\R\eta_J$,
\[
\rank \Cm_J=|J|-1
\]
and $|J|\le\rank \Cm+1$.
\end{proof}

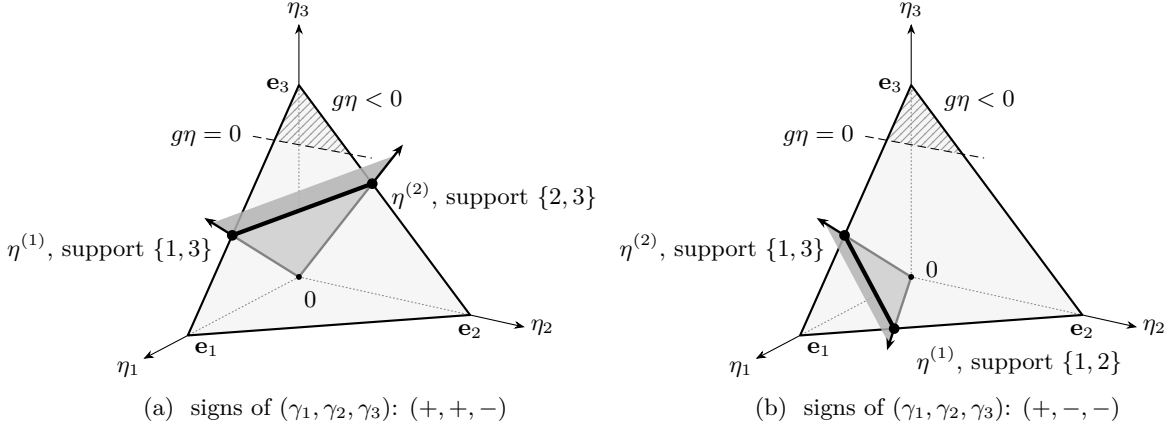
\begin{figure}[t]
\centering
\begin{tikzpicture}[x={(-1.470cm,-0.774cm)},y={(2.264cm,-0.503cm)},z={(0cm,2.537cm)},line join=round,line cap=round]
\def\conepanel#1#2{%
\coordinate (O)  at (0,0,0);
\coordinate (e1) at (1,0,0); \coordinate (e2) at (0,1,0); \coordinate (e3) at (0,0,1);
\coordinate (p1) at #1; \coordinate (p2) at #2;
\coordinate (P1) at (0.22,0,0.78); \coordinate (P2) at (0,0.30,0.70);
\draw[densely dotted] (O)--(e1) (O)--(e2) (O)--(e3);
\path[fill=black!30] (O)--(p1)--(p2)--cycle;
\draw[line width=0.9pt] (O)--(p1) (O)--(p2);
\path[fill=black!6,fill opacity=0.55] (e1)--(e2)--(e3)--cycle;
\path[pattern=north east lines,pattern color=black!40] (e3)--(P1)--(P2)--cycle;
\draw[thick] (e1)--(e2)--(e3)--cycle;
\draw[densely dashed] ($(P1)!-0.28!(P2)$)--($(P2)!-0.28!(P1)$);
\path[fill=black!30,fill opacity=0.85] (p1)--($(O)!1.30!(p1)$)--($(O)!1.30!(p2)$)--(p2)--cycle;
\draw[line width=0.9pt,-{Stealth[length=4.5pt]}] (p1)--($(O)!1.42!(p1)$);
\draw[line width=0.9pt,-{Stealth[length=4.5pt]}] (p2)--($(O)!1.42!(p2)$);
\draw[-{Stealth[length=4pt]}] (e1)--(1.40,0,0);
\draw[-{Stealth[length=4pt]}] (e2)--(0,1.32,0);
\draw[-{Stealth[length=4pt]}] (e3)--(0,0,1.32);
\node[slab,anchor=north east,inner sep=1pt] at (1.40,0,0) {$\eta_1$};
\node[slab,anchor=west,inner sep=2pt] at (0,1.32,0) {$\eta_2$};
\node[slab,anchor=south,inner sep=2pt] at (0,0,1.32) {$\eta_3$};
\draw[line width=1.6pt] (p1)--(p2);
\fill (p1) circle (2pt); \fill (p2) circle (2pt);
\fill (O) circle (1.1pt);
\node[slab,anchor=north west,inner sep=2.5pt] at (e1) {$\mathbf e_1$};
\node[slab,anchor=north,inner sep=3.5pt] at (e2) {$\mathbf e_2$};
\node[slab,anchor=east,inner sep=3pt] at (e3) {$\mathbf e_3$};
\node[slab,anchor=east] at ($(P1)!-0.30!(P2)$) {$g\eta=0$};
\node[slab,anchor=west] at ($(e3)+(0.30cm,-0.20cm)$) {$g\eta<0$};
}
\begin{scope}
\conepanel{(0.6,0,0.4)}{(0,0.4286,0.5714)}
\node[slab,anchor=east]  at ($(p1)+(-0.12cm,-0.16cm)$) {$\eta^{(1)}$, support $\{1,3\}$};
\node[slab,anchor=west]  at ($(p2)+(0.12cm,-0.16cm)$) {$\eta^{(2)}$, support $\{2,3\}$};
\node[slab,anchor=north] at ($(O)+(0.15cm,-0.06cm)$) {$0$};
\node[slab] at (0.35cm,-1.75cm) {(a)\ \ signs of $(\gamma_1,\gamma_2,\gamma_3)$: $(+,+,-)$};
\end{scope}
\begin{scope}[shift={(8.1cm,0cm)}]
\conepanel{(0.6667,0.3333,0)}{(0.6,0,0.4)}
\node[slab,anchor=north west] at ($(p1)+(0.16cm,-0.10cm)$) {$\eta^{(1)}$, support $\{1,2\}$};
\node[slab,anchor=east]  at ($(p2)+(-0.12cm,-0.16cm)$) {$\eta^{(2)}$, support $\{1,3\}$};
\node[slab,anchor=west] at ($(O)+(0.06cm,0.10cm)$) {$0$};
\node[slab] at (0.35cm,-1.75cm) {(b)\ \ signs of $(\gamma_1,\gamma_2,\gamma_3)$: $(+,-,-)$};
\end{scope}
\end{tikzpicture}
\caption{The cone $\mathcal K(\Cm)$ in the nonnegative orthant for $q=3$ and one row $\Cm=(\gamma_1,\gamma_2,\gamma_3)$.  The triangle is the simplex $\{\eta\ge0,\ \eta_1+\eta_2+\eta_3=1\}$, and the thick segment is its section with $\mathcal K(\Cm)$. The endpoints of the segment generate the extreme rays.  The condition $g\eta\ge0$ holds on $\mathcal K(\Cm)$ if and only if it holds at the two endpoints.  The supports of the endpoints depend on the signs of $\gamma_1,\gamma_2,\gamma_3$.}
\label{fig:cone}
\end{figure}

\subsubsection{Explicit conditions}\label{subsec:explicit}

Let $\Cm\in\R^{e\times q}$ and $g\in\R^{1\times q}$ be fixed.  For an integer $k\ge1$, a set $J=\{j_1<\cdots<j_k\}\subseteq[q]$, and a row set $E\subseteq[e]$ of size $k-1$, define
\begin{equation}\label{eq:circuit-cofactor}
\Delta^{E,J}_t
=(-1)^t\det \Cm[E,J\setminus\{j_t\}],
\qquad t\in[k].
\end{equation}
By Lemma~\ref{lem:cofactor}, the vector $\Delta^{E,J}=(\Delta^{E,J}_1,\ldots,\Delta^{E,J}_k)^{\mathsf T}$ spans the kernel of the $(k-1)\times k$ matrix $\Cm[E,J]$ whenever this matrix has rank $k-1$.  In all determinants the rows and the columns are taken in increasing order.  For $k=1$ we have $E=\varnothing$.  The determinant of the empty matrix is one, and hence $\Delta^{\varnothing,J}_1=-1$.

\begin{lemma}\label{lem:circuit-elimination}
Let $\Cm\in\R^{e\times q}$ and $g\in\R^{1\times q}$.  The condition
\begin{equation}\label{eq:farkas-condition}
g\eta\ge0
\quad\text{for every }\eta\in\mathcal K(\Cm)
\end{equation}
holds if and only if the following two implications hold for every integer $k$ with $1\le k\le\min\{q,e+1\}$, every set $J=\{j_1<\cdots<j_k\}\subseteq[q]$, and every set $E\subseteq[e]$ with $|E|=k-1$.
\begin{align}
&\text{if all minors of order $k$ of $\Cm_J$ are zero and } \Delta^{E,J}_t>0\text{ for every } t\in[k],\text{ then }
\sum_{t=1}^{k}g_{j_t}\Delta^{E,J}_t\ge0,
\label{eq:positive-circuit}\\
&\text{if all minors of order $k$ of $\Cm_J$ are zero and }\Delta^{E,J}_t<0\text{ for every }t\in[k],\text{ then }
\sum_{t=1}^{k}g_{j_t}\Delta^{E,J}_t\le0.
\label{eq:negative-circuit}
\end{align}
For $k=e+1$ the matrix $\Cm_J$ has no minors of order $k$, and the condition on the minors is empty. Both \eqref{eq:positive-circuit} and \eqref{eq:negative-circuit} are implications between basic conditions in the sense of Definition~\ref{def:polysystem}.  The number of their polynomials depends only on $(e,q)$, and the degrees of these polynomials in the entries of $(g,\Cm)$ are at most $e+1$.
\end{lemma}

If the hypothesis of \eqref{eq:positive-circuit} holds, then $\Delta^{E,J}$ is a strictly positive vector which spans $\ker\Cm_J$, and if the hypothesis of \eqref{eq:negative-circuit} holds, then $-\Delta^{E,J}$ is such a vector (Step~\ref{step:circuit-forward} of the proof).  The extension $\eta$ of this vector by zero outside $J$ belongs to $\mathcal K(\Cm)$, and the conclusion of the implication is the inequality $g\eta\ge0$.  Every extreme ray of $\mathcal K(\Cm)$ is generated by such a vector $\eta$ (Step~\ref{step:circuit-reverse} of the proof).  For $e=1$ the conditions of the lemma are given in Example~\ref{ex:one-row}.  The proof is given in Appendix~\ref{app:proofs}.

\begin{example}\label{ex:one-row}
Let $e=1$.  Then $\Cm=(\gamma_1,\ldots,\gamma_q)$ is a row and $z$ in \eqref{eq:z-lift} is a scalar.  The inequality $g_j+z\gamma_j\ge0$ is the lower bound $z\ge-g_j/\gamma_j$ if $\gamma_j>0$, the upper bound $z\le-g_j/\gamma_j$ if $\gamma_j<0$, and the condition $g_j\ge0$ if $\gamma_j=0$.  Hence $z$ exists if and only if $g_j\ge0$ whenever $\gamma_j=0$, and $g_j|\gamma_{j'}|+g_{j'}\gamma_j\ge0$ whenever $\gamma_j>0>\gamma_{j'}$, i.e. every lower bound is at most every upper bound. In Figure~\ref{fig:cone}, the endpoint with support $\{j,j'\}$ is the vector $\eta=(|\gamma_{j'}|\mathbf e_j+\gamma_j\mathbf e_{j'})/(\gamma_j+|\gamma_{j'}|)$, and the inequality $g\eta\ge0$ at this endpoint is $g_j|\gamma_{j'}|+g_{j'}\gamma_j\ge0$. These are the conditions of Lemma~\ref{lem:circuit-elimination} for $k=1$ and $k=2$.  

In \eqref{eq:g-C-system} below, the entries of $g$ and $\Gamma$ are polynomials in the variables $y_{mij}$.  The signs of the $\gamma_j$ then depend on the values of the variables, and therefore the conditions are stated as implications for every $j$ and every pair $(j,j')$, see Figure~\ref{fig:cone}.
\end{example}

We now apply Lemma~\ref{lem:circuit-elimination} to condition \ref{cond:lifting}.  Fix a profile of $(I,q)$.  For every $m\in H$ fix a rational matrix $N_m$ whose rows form a basis of the left kernel of $\bar X_m$, and for every omitted row $i\notin S_m$ fix a rational row $w_{m,i}$ with $w_{m,i}\bar X_m=\bar a_{m,i}$.  Put
\begin{equation}\label{eq:g-C-system}
g=w_{m,i}\bar Y_m,
\qquad
\Cm=N_m\bar Y_m .
\end{equation}
The matrix $\Cm$ has $e_m=\delta_m-\sigma_m$ rows, and \eqref{eq:profile-dimension-range} gives $e_m\le q$.  Thus the supports and minors appearing in Lemma~\ref{lem:circuit-elimination} have size bounded in terms of the number of competing terms, not the ambient mode dimension.  The entries of $g$ and $\Cm$ are polynomials of degree at most one in the variables $y_{mij}$.  By Proposition~\ref{prop:lifting}\ref{item:lifting-dual} and Lemma~\ref{lem:circuit-elimination}, condition \ref{cond:lifting} for the row $i$ holds if and only if the conditions of Lemma~\ref{lem:circuit-elimination} hold for \eqref{eq:g-C-system}.  These are conditions on $\bar Y_m$ only.  From now on, $\mathcal N_{I,q}$ and $\mathcal N^{\mathrm{td}}_{I,p}$ denote the systems in which \ref{cond:lifting} has been replaced in this way.  They are polynomial systems in the sense of Definition~\ref{def:polysystem}.  Theorem~\ref{thm:representation} and Corollary~\ref{cor:criteria} remain true for them, because the two forms of \ref{cond:lifting} are equivalent.  Thus the ambient dimensions now affect the number of polynomial conditions but not the number of real variables.

\subsubsection{The number of variables}\label{subsec:count}

\begin{corollary}\label{cor:count}
For every profile of $(I,q)$, the polynomial system $\mathcal N_{I,q}$ has exactly
\begin{equation}\label{eq:exact-count}
q\kappa+q-1
\end{equation}
variables, and therefore at most
\begin{equation}\label{eq:variable-bound}
q(p+q-2)+q-1
\end{equation}
variables.  The same holds for $\mathcal N^{\mathrm{td}}_{I,p}$ with $q=p$.  No variable belongs to an inactive mode or to an omitted row.
\end{corollary}

\begin{proof}
For an active mode $m$, the matrix $\bar Y_m$ contains exactly $q(\delta_m-1)$ variables by \eqref{eq:simplex}. Summing over the active modes gives
\[
q\sum_{m\in H}(\delta_m-1)=q\kappa
\]
variables.  By \eqref{eq:last-weight}, there are exactly $q-1$ weight variables.  The conditions \ref{cond:simplex}--\ref{cond:core}, the inequalities \eqref{eq:td}, and the conditions which replace \ref{cond:lifting} contain no further variable.  The bound \eqref{eq:variable-bound} follows from \eqref{eq:profile-budget}.
\end{proof}

\begin{remark}\label{rem:counts}
For a balanced localized replacement with $q=p$, the bound \eqref{eq:variable-bound} is
\[
p(p+p-2)+p-1=2p^2-p-1.
\]
For a shortening localized replacement with $q<p$, the function
\[
q\longmapsto q(p+q-2)+q-1=q(p+q-1)-1
\]
is strictly increasing for $q\ge1$, so its maximum over $q<p$ occurs at $q=p-1$ and equals
\[
(p-1)(2p-3)+(p-2)=2p^2-4p+1
=2(p-1)^2-1.
\]
Thus the largest numbers of variables are
\[
\begin{array}{c|rrrr}
p&2&3&4&5\\ \hline
\text{balanced, $q=p$}&5&14&27&44\\
\text{shortening, $q=p-1$}&1&7&17&31
\end{array}
\]
\end{remark}

\section{Exact certification and reconstruction}\label{sec:one-shot}

By Theorem~\ref{thm:rank-localization}, a realization of some disjoint shortening pattern gives a decomposition of $\T$ of length $\rankp(\T)$.  We therefore combine the block systems of a pattern into a single system.  For reference, Table~\ref{tab:objects} collects the objects of Sections~\ref{sec:localization} and~\ref{sec:normalized} which are used in this section.

\subsection{Exact real-algebraic computation}\label{subsec:exact-computation}

We represent the solutions of polynomial systems exactly and use the following algorithm for deciding their feasibility.

\begin{definition}\label{def:exact-representation}
Let $\theta$ be a real algebraic number, and let $f\in\mathbb Z[T]$ be a nonzero polynomial with $f(\theta)=0$.
\begin{enumerate}[label=\textup{(\roman*)}]
\item The \emph{Thom encoding} of $\theta$ with respect to $f$ is the list of the signs of the derivatives at $\theta$, 
\[
\big( \sgn f'(\theta),\sgn f''(\theta),\ldots, \sgn f^{(\deg f)}(\theta) \big).
\] 
These signs determine $\theta$ among the real roots of $f$ \citep{CosteRoy1988}, \citep[Proposition~2.28 and Definition~2.29]{BasuPollackRoy2006}.  We say that $\theta$ is \emph{given by} $f$ and its Thom encoding.
\item An \emph{exact representation} of a point $y\in\Q(\theta)^k$ is a list $(f,\theta,g_0,g_1,\ldots,g_k)$, in which $\theta$ is given by $f$ and its Thom encoding, and $g_0,g_1,\ldots,g_k\in\mathbb Z[T]$ are polynomials with 
\[
g_0(\theta)\ne0,\qquad \text{and}\qquad y_i=\frac{g_i(\theta)}{g_0(\theta)}\quad \text{for every $i\in[k]$}.
\]
\item\label{item:exact-representation-bits} The \emph{bit length} of an exact representation is the sum of the bit lengths of all coefficients of $f,g_0,g_1,\ldots,g_k$ and of all entries of the Thom encoding of $\theta$.
\item\label{item:quotient-bits} Every element of $\Q(\theta)$ is a quotient $g_1(\theta)/g_0(\theta)$ with $g_0,g_1\in\mathbb Z[T]$ and $g_0(\theta)\ne0$.  The \emph{bit length} of such a quotient is the sum of the bit lengths of all coefficients of $g_0$ and $g_1$.
\end{enumerate}
\end{definition}

Unlike the bit length of a polynomial, the bit lengths in Definition~\ref{def:exact-representation}\ref{item:exact-representation-bits} and \ref{item:quotient-bits} are sums over all coefficients, so they also bound the degrees of the polynomials.  The implications allowed in Definition~\ref{def:polysystem} cause no additional quantifiers: each implication $\mathsf A\Rightarrow\mathsf B$ is the quantifier-free Boolean condition $\neg\mathsf A\vee\mathsf B$ in the same family of polynomial sign conditions.  Thus the number of variables is unchanged and the size of the Boolean formula is bounded by a constant multiple of the length defined above.  A polynomial in $k$ variables of degree at most $D$ has at most $\binom{D+k}{k}$ monomials, so the dense encoding size of a system of length $s$ is bounded by $s\binom{D+k}{k}(b+O(1))$, which is itself $(s+D+b+2)^{O(k+1)}$.  The following theorem is a form of the sampling and existential-theory results of \citet{BasuPollackRoy2006}, which we recall precisely.  Their Algorithm~13.2 (Sampling) takes $s$ polynomials of degree at most $D$ in $k$ variables, computes a set of real univariate representations whose points meet every semialgebraically connected component of every realizable sign condition, and reports the signs of all the polynomials at these points.  By \citet[Theorem~13.11]{BasuPollackRoy2006}, it performs $s^{k+1}D^{O(k)}$ arithmetic operations in the ring generated by the coefficients, the univariate representations have degree at most $O(D)^k$, and for integer coefficients of bit size at most $b$ their coefficients have bit size at most $bD^{O(k)}$.  A sentence $(\exists y)\,F(y)$ with $F$ a quantifier-free Boolean combination of the sign conditions is decided by reading these signs, as in \citet[Theorem~13.13]{BasuPollackRoy2006}.  The passage from the arithmetic complexity to a bit complexity of the form \eqref{eq:sampling-complexity} uses the bit-size analyses of the subroutines of Algorithm~13.2 in the same reference; one bounds the bit size of every intermediate integer by a quantity of the form $(s+D+b+2)^{O(k+1)}$, so that each arithmetic operation costs at most polynomially many bit operations in that quantity.

\begin{theorem}\label{thm:sampling}
There are an absolute constant $C_{\mathrm{sa}}$ and a deterministic algorithm with the following property.  For a polynomial system of length $s$ in $k$ variables whose polynomials have degree at most $D$ and bit length at most $b$, the algorithm decides whether the system has a solution, and it computes one solution in exact representation if a solution exists.  The number of bit operations is at most
\begin{equation}\label{eq:sampling-complexity}
(s+D+b+2)^{C_{\mathrm{sa}}(k+1)}.
\end{equation}
The degree of the polynomial $f$ of the computed solution is at most
\begin{equation}\label{eq:sampling-degree}
(D+1)^{C_{\mathrm{sa}}(k+1)},
\end{equation}
and the bit length of the exact representation of the computed solution is at most \eqref{eq:sampling-complexity}.
\end{theorem}

For $k=0$, feasibility is decided by evaluating the constant sign conditions. A feasible system has the empty coordinate tuple, represented, for example, by $f(T)=T$ and its root $\theta=0$. This case satisfies the same bounds.

\subsection{Combined systems}\label{subsec:combined}

Fix a disjoint shortening pattern $\mathfrak D=((I_\nu,q_\nu))_{\nu=1}^{c}$, and for every block fix a profile of $(I_\nu,q_\nu)$ of Definition~\ref{def:profile}, with effective order $\kappa_\nu$.  Let $\mathcal N_{I_\nu,q_\nu}$ be the polynomial system of Section~\ref{sec:lifting} for this profile. 

\begin{table}[t]
\centering
\small
\begin{tabularx}{\textwidth}{@{}llX@{}}
\toprule
Object & Defined in & Meaning\\
\midrule
block $(I,q)$ & Definition~\ref{def:pattern} & the supplied terms $P_r$, $r\in I$, are replaced by $q$ terms\\
pattern $\mathfrak D=((I_\nu,q_\nu))_{\nu=1}^{c}$ & Definition~\ref{def:pattern} & blocks with pairwise disjoint $I_\nu$ and $q_\nu<p_\nu$\\
realization of $\mathfrak D$ & Definition~\ref{def:pattern} & the new terms $Q_{\nu j}$ of all blocks, see \eqref{eq:local-replacements}\\
profile of $(I,q)$ & Definition~\ref{def:profile} & the numbers $\delta_m$ and the row sets $K_m$ of one block\\
system $\mathcal N_{I,q}$ & Definition~\ref{def:system}, Section~\ref{sec:lifting} & polynomial system of block $(I,q)$ for one profile\\
\bottomrule
\end{tabularx}
\caption{The objects of Sections~\ref{sec:localization} and \ref{sec:normalized} used in Section~\ref{sec:one-shot}. }
\label{tab:objects}
\end{table}

\begin{definition}\label{def:combined}
The \emph{combined system} $\mathcal N_{\mathfrak D}$ is the polynomial system whose variables are the variables of the systems $\mathcal N_{I_\nu,q_\nu}$, $\nu\in[c]$, where different blocks have different variables, and whose conditions are all conditions of these systems.  A solution of $\mathcal N_{\mathfrak D}$ consists of one solution of every $\mathcal N_{I_\nu,q_\nu}$, so $\mathcal N_{\mathfrak D}$ has a solution if and only if every $\mathcal N_{I_\nu,q_\nu}$ has a solution.  For the empty pattern, $\mathcal N_{\mathfrak D}$ has no variable and no condition.
\end{definition}

\begin{proposition}\label{prop:one-shot}
For a disjoint shortening pattern $\mathfrak D=((I_\nu,q_\nu))_{\nu=1}^{c}$ put $p_\nu=|I_\nu|$.
\begin{enumerate}[label=\textup{(\roman*)}]
\item\label{item:one-shot-completeness}
If $\rankp(\T)<R$, then there are a pattern $\mathfrak D$ with $\rho(\mathfrak D)=\rankp(\T)$ and, for every block of $\mathfrak D$, a profile with $\kappa_\nu\le p_\nu+q_\nu-2$, such that $\mathcal N_{\mathfrak D}$ has a solution.
\item\label{item:one-shot-soundness}
For every pattern $\mathfrak D$ and every choice of profiles, every solution of $\mathcal N_{\mathfrak D}$ gives a nonnegative decomposition of $\T$ of length $\rho(\mathfrak D)$.
\item\label{item:one-shot-rank-formula}
Consequently,
\begin{equation}\label{eq:feasible-rank-formula}
\rankp(\T)
=\min\Bigl(\{R\}\cup
\bigl\{\rho(\mathfrak D):
\mathcal N_{\mathfrak D}\text{ has a solution for some choice of profiles}\bigr\}
\Bigr).
\end{equation}
\end{enumerate}
\end{proposition}

\begin{proof}
For part~\ref{item:one-shot-completeness}, fix a minimal decomposition of $\T$.  By the second statement of Theorem~\ref{thm:rank-localization}, there are a pattern $\mathfrak D$ with $\rho(\mathfrak D)=\rankp(\T)$ and a realization of $\mathfrak D$ in which every localized replacement is irreducible.  By Theorem~\ref{thm:representation}\ref{item:completeness}, for every block there is a profile with $\kappa_\nu\le p_\nu+q_\nu-2$ such that the system of this profile has a solution.  Different blocks have different variables, so these solutions together form a solution of $\mathcal N_{\mathfrak D}$.

For part~\ref{item:one-shot-soundness}, fix a pattern $\mathfrak D$, profiles of its blocks, and a solution of $\mathcal N_{\mathfrak D}$.  It consists of one solution of every $\mathcal N_{I_\nu,q_\nu}$.  By Theorem~\ref{thm:representation}\ref{item:soundness}, the solution of $\mathcal N_{I_\nu,q_\nu}$ gives a localized replacement on $I_\nu$ with $q_\nu$ terms.  These replacements form a realization of $\mathfrak D$, and Lemma~\ref{lem:combine} gives a nonnegative decomposition of $\T$ of length $\rho(\mathfrak D)$.

For part~\ref{item:one-shot-rank-formula}, the supplied decomposition has length $R$, and by part~\ref{item:one-shot-soundness} every number $\rho(\mathfrak D)$ in the set on the right-hand side of \eqref{eq:feasible-rank-formula} is the length of a nonnegative decomposition of $\T$.  Hence $\rankp(\T)$ is at most the minimum in \eqref{eq:feasible-rank-formula}.  Conversely, the minimum is at most $R$, and if $\rankp(\T)<R$, then by part~\ref{item:one-shot-completeness} the set contains $\rho(\mathfrak D)=\rankp(\T)$.  In both cases the minimum is at most $\rankp(\T)$.
\end{proof}

\begin{lemma}\label{lem:patterns}
There are at most $G(R)$ disjoint shortening patterns on $[R]$, and they can be listed in time $G(R)$.
\end{lemma}

\begin{proof}
A pattern is finite combinatorial data on the $R$ labels: it consists of pairwise disjoint subsets $I\subseteq[R]$ of size at least two, together with one integer $q\in[|I|-1]$ for each chosen subset.  Hence the number of patterns depends only on $R$.  Exhaustive enumeration and the checks in Definition~\ref{def:pattern} therefore take $G(R)$ time.
\end{proof}

\begin{proposition}\label{prop:dimension}
Let $\mathfrak D=((I_\nu,q_\nu))_{\nu=1}^{c}$ be a disjoint shortening pattern, put $p_\nu=|I_\nu|$ and $\rho=\rho(\mathfrak D)$, and fix profiles of its blocks with effective orders $\kappa_\nu$.  The number of variables of the combined system $\mathcal N_{\mathfrak D}$ is
\begin{equation}\label{eq:exact-total-variables}
\sum_{\nu=1}^{c}(q_\nu\kappa_\nu+q_\nu-1).
\end{equation}
If $\kappa_\nu\le p_\nu+q_\nu-2$ for every $\nu$, then this number is at most
\begin{equation}\label{eq:target-rank-variable-bound}
\rho(R+\rho-1)-1.
\end{equation}
If moreover $\mathfrak D$ is nonempty, then $\rho\le R-1$, and the number of variables is at most
\begin{equation}\label{eq:uniform-shorter-variable-bound}
2R^2-4R+1<2R^2.
\end{equation}
\end{proposition}

\begin{proof}
By Corollary~\ref{cor:count}, the system of the block $\nu$ has $q_\nu\kappa_\nu+q_\nu-1$ variables, and different blocks have different variables.  This gives \eqref{eq:exact-total-variables}.

The sets $I_\nu$ are pairwise disjoint subsets of $[R]$, so $\sum_\nu p_\nu\le R$, and by \eqref{eq:resulting-length} $\sum_\nu q_\nu=\rho-R+\sum_\nu p_\nu\le\rho$.  Hence $p_\nu+q_\nu-1\le R+\rho-1$ for every $\nu$.  If $\kappa_\nu\le p_\nu+q_\nu-2$ for every $\nu$, then
\[
\sum_{\nu=1}^{c}(q_\nu\kappa_\nu+q_\nu-1)
\le\sum_{\nu=1}^{c}q_\nu(p_\nu+q_\nu-1)-c
\le(R+\rho-1)\sum_{\nu=1}^{c}q_\nu-c
\le\rho(R+\rho-1)-c.
\]
For $c\ge1$ this proves \eqref{eq:target-rank-variable-bound}, and for $c=0$ the number of variables is $0$.

A nonempty pattern has at least one block, and every block has $q_\nu<p_\nu$, so $\rho\le R-1$.  The number $\rho(R+\rho-1)$ increases with $\rho$, and for $\rho=R-1$ it equals $(R-1)(2R-2)$.  Hence the number of variables is at most $(R-1)(2R-2)-1=2R^2-4R+1$.
\end{proof}

\subsection{The size of the systems}\label{subsec:size}

\begin{proposition}\label{prop:one-shot-format}
Let the supplied decomposition have length $R$ and total bit length $L$.
\begin{enumerate}[label=\textup{(\roman*)}]
\item\label{item:format-single}
Fix $I\subseteq[R]$, $p=|I|\ge2$, and $q\ge1$.  There are at most $G(p,q)L^{O(p+q)}$ profiles of $(I,q)$, and they can be listed in time $G(p,q)L^{O(p+q)}$.  For every profile, $\mathcal N_{I,q}$ and, for $q=p$, $\mathcal N^{\mathrm{td}}_{I,p}$ are polynomial systems of length at most $G(p,q)(L+1)$ in at most $q(p+q-2)+q-1$ variables.  Their polynomials have degree at most $p+q-1$ and bit length at most $G(p,q)(L+1)^{C}$.
\item\label{item:format-combined}
There are at most $G(R)L^{O(R)}$ choices of a nonempty disjoint shortening pattern $\mathfrak D$ on $[R]$ and of profiles for its blocks, and they can be listed in time $G(R)L^{O(R)}$.  For every such choice, $\mathcal N_{\mathfrak D}$ is a polynomial system of length at most $G(R)(L+1)$ in fewer than $2R^2$ variables.  Its polynomials have degree at most $2R-2$ and bit length at most $G(R)(L+1)^{C}$.
\end{enumerate}
\end{proposition}

The proof is routine bookkeeping and is given in Appendix~\ref{app:complexity}.

\subsection{The algorithm}\label{subsec:algorithm}

Algorithm~\ref{alg:main} is the algorithm of Theorem~\ref{thm:main}.  It lists sets, patterns, and profiles in lexicographic order.  For this purpose, a set is written as the increasing sequence of its elements, and a disjoint shortening pattern as the sequence of its blocks $(I_\nu,q_\nu)$, ordered by the smallest element of $I_\nu$.  A profile is written as the increasing sequence of its active modes $m$, followed by the pairs $(\delta_m,K_m)$ in the same order.  In a reconstructed decomposition, the terms $Q_{\nu j}$ of the blocks come first, in the order $\nu=1,\ldots,c$ and, within the block $\nu$, in the order $j=1,\ldots,q_\nu$.  The unchanged supplied terms $P_r$ follow in increasing order of $r$.

The algorithm uses two procedures.  For a polynomial system, $\textsc{Solve}$ decides whether the system has a solution and, if it has one, returns one solution in exact representation, see Theorem~\ref{thm:sampling}.  The coordinates of this solution lie in a field $\Q(\theta)$.  From a solution of $\mathcal N_{\mathfrak D}$ or of $\mathcal N^{\mathrm{td}}_{I,p}$, $\textsc{Reconstruct}$ computes the corresponding decomposition of $\T$.  Since the coefficient rows $\ell_{m,i}$ were eliminated in Section~\ref{sec:lifting}, they are not part of the solution, and $\textsc{Reconstruct}$ first computes them with entries in $\Q(\theta)$ by Lemma~\ref{lem:field-lifting}\ref{item:field-linear-systems} in Appendix~\ref{app:complexity}.  Then it forms the localized replacement of every block by Theorem~\ref{thm:representation}\ref{item:soundness}, and it replaces the supplied terms of every block by the new terms, as in \eqref{eq:combined-decomposition}.

\begin{algorithm}[t]
\small
\caption{The algorithm of Theorem~\ref{thm:main}}
\label{alg:main}
\begin{algorithmic}[1]
\Require an explicit rational list $A^{(m)}\in\Qpos^{n_m\times R}$, $m\in[d]$, with no zero column
\Ensure $\rankp(\T)$; an exact minimal decomposition $\mathcal D^*$; the two decisions of Theorem~\ref{thm:main}\ref{item:main-decisions}; the exact competitors of Theorem~\ref{thm:main}\ref{item:main-witnesses}
\If{$R=1$}\label{line:base}
\State \Return $\rankp(\T)=1$, the supplied decomposition, and the decisions \textsc{minimal} and \textsc{unique}
\EndIf
\State $\rho^*\gets R$\Comment{Phase I: computation of the rank}\label{line:init}
\For{$\rho=1,2,\ldots,R-1$}\label{line:loop-rho}
\For{every nonempty disjoint shortening pattern $\mathfrak D$ with $\rho(\mathfrak D)=\rho$}\label{line:loop-pattern}
\For{every choice of profiles for the blocks of $\mathfrak D$}\label{line:loop-profile}
\State form the combined system $\mathcal N_{\mathfrak D}$\Comment{Definition~\ref{def:combined} and Proposition~\ref{prop:one-shot-format}}\label{line:build}
\If{$\textsc{Solve}(\mathcal N_{\mathfrak D})$ returns a solution $y$ with coordinates in $\Q(\theta)$}\label{line:sample}
\State $\rho^*\gets\rho$ and \textbf{go to} line~\ref{line:reconstruct}\label{line:accept}
\EndIf
\EndFor
\EndFor
\EndFor\label{line:rank-R}
\If{$\rho^*=R$}\Comment{Phase II: exact reconstruction}\label{line:reconstruct}
\State $\mathcal D^*\gets{}$the supplied decomposition
\Else
\State $\mathcal D^*\gets\textsc{Reconstruct}(y)$, a nonnegative decomposition of length $\rho^*$ with entries in $\Q(\theta)$
\EndIf
\State \textbf{output} $\rankp(\T)=\rho^*$, the decomposition $\mathcal D^*$, and \textsc{minimal} exactly if $\rho^*=R$\label{line:decide-min}
\If{$\rho^*=R$}\Comment{Phase III: essential uniqueness at length $R$}\label{line:unique-branch}
\For{every $I\subseteq[R]$ with $|I|\ge2$, by increasing size and then lexicographically}\label{line:loop-subset}
\For{every profile of the pair $(I,|I|)$}\label{line:loop-td-profile}
\State form the system $\mathcal N^{\mathrm{td}}_{I,|I|}$\Comment{Definition~\ref{def:system} and Corollary~\ref{cor:criteria}}\label{line:build-td}
\If{$\textsc{Solve}(\mathcal N^{\mathrm{td}}_{I,|I|})$ returns a solution $y'$}\label{line:sample-td}
\State \textbf{output} \textsc{not unique} and \Return $\textsc{Reconstruct}(y')$\label{line:witness-unique}
\EndIf
\EndFor
\EndFor
\State \textbf{output} \textsc{unique} and \Return\label{line:unique}
\Else\label{line:nonmin-branch}
\State \textbf{output} \textsc{not unique}, with $\mathcal D^*$ as the shorter competitor, and \Return\label{line:short-witness}
\EndIf
\end{algorithmic}
\end{algorithm}

\begin{proof}[Proof of Theorem~\ref{thm:main}]
\stepcounter{proofstepblock}
\proofstep{the case $R=1$}
This is line~\ref{line:base} of Algorithm~\ref{alg:main}.  A nonzero nonnegative rank-one tensor has nonnegative rank one, and its one-term decomposition is essentially unique by Lemma~\ref{lem:product-equality}.  The algorithm returns the supplied decomposition; henceforth assume $R\ge2$.

\proofstep[step:main-rank]{computation of the rank}
This is Phase~I, lines~\ref{line:init}--\ref{line:rank-R} of Algorithm~\ref{alg:main}.  For $\rho=1,\ldots,R-1$, the algorithm goes through the nonempty disjoint shortening patterns $\mathfrak D$ with $\rho(\mathfrak D)=\rho$ and the choices of profiles for their blocks, as in Proposition~\ref{prop:one-shot-format}\ref{item:format-combined}.  For each of them it applies $\textsc{Solve}$ to the combined system $\mathcal N_{\mathfrak D}$, and it stops when $\textsc{Solve}$ returns a solution.  Hence $\rho^*$ is the smallest resulting length $\rho(\mathfrak D)<R$ of a combined system which has a solution, and $\rho^*=R$ if there is no such system.  By Proposition~\ref{prop:one-shot}\ref{item:one-shot-rank-formula}, in both cases $\rho^*=\rankp(\T)$.

\proofstep[step:main-reconstruct]{exact reconstruction of a minimal decomposition}
This is Phase~II, line~\ref{line:reconstruct}.  If $\rho^*=R$, the algorithm outputs the supplied decomposition, which is minimal by Step~\ref{step:main-rank}. Otherwise, Step~\ref{step:main-rank} gives a combined system $\mathcal N_{\mathfrak D}$ with $\rho(\mathfrak D)=\rho^*$ and a solution $y$ of it whose coordinates lie in one field $\Q(\theta)$. The coordinates of $y$ are the values of the variables $y_{mij}$ and $\omega_1,\ldots,\omega_{q-1}$ of all blocks, so by \eqref{eq:simplex} and \eqref{eq:last-weight} the point $y$ gives the matrices $\bar Y_m$ and the weights $\omega_1,\ldots,\omega_q$ of every block.  $\textsc{Reconstruct}$ computes $\mathcal D^*$ from $y$ as follows.
\begin{enumerate}[label=\textup{(\arabic*)}]
\item\label{item:reconstruct-rows} For every block, every active mode $m$ of its profile, and every omitted row $i$, it computes a coefficient row $\ell_{m,i}$ with entries in $\Q(\theta)$.  These rows are not part of $y$, because they were eliminated in Section~\ref{sec:lifting}.  At the point $y$, condition \eqref{eq:row-lift} is a system of linear equations and inequalities in $\ell_{m,i}$ with coefficients in $\Q(\theta)$.  It has a real solution, because $y$ satisfies the conditions which replace \ref{cond:lifting}, so Lemma~\ref{lem:field-lifting}\ref{item:field-linear-systems} gives such a row.
\item\label{item:reconstruct-terms} For every block, it forms the terms \eqref{eq:reconstructed-term} of the localized replacement of Theorem~\ref{thm:representation}\ref{item:soundness}.  The matrix $\Lambda_m$ has the rows of the identity matrix on $S_m$ and the rows $\ell_{m,i}$ elsewhere, and
\[
\widetilde Q_j=\omega_j\Bigl(\bigotimes_{m\in H}\Lambda_m\bar Y_m[:,j]\Bigr)\otimes\Bigl(\bigotimes_{m\notin H}u_m\Bigr),\qquad j\in[q].
\]
Only additions and multiplications in $\Q(\theta)$ are used.
\item\label{item:reconstruct-list} It lists the terms of these replacements, followed by the unchanged supplied terms, as in \eqref{eq:combined-decomposition}.
\end{enumerate}
All entries of $\mathcal D^*$ lie in the same field $\Q(\theta)$.  By Proposition~\ref{prop:one-shot}\ref{item:one-shot-soundness} and Step~\ref{step:main-rank}, $\mathcal D^*$ has length $\rho(\mathfrak D)=\rho^*=\rankp(\T)$, so it is minimal.

\proofstep{minimality and essential uniqueness of the supplied decomposition}
These are lines~\ref{line:decide-min}--\ref{line:short-witness} of Algorithm~\ref{alg:main}.  By Step~\ref{step:main-rank}, the supplied decomposition is minimal exactly when $\rho^*=R$, and line~\ref{line:decide-min} outputs \textsc{minimal} exactly in this case.  If $\rho^*<R$, the decomposition $\mathcal D^*$ of Step~\ref{step:main-reconstruct} is an exact shorter competitor, and the supplied decomposition is not essentially unique by Lemma~\ref{lem:nonminimal-nonunique}.  If $\rho^*=R$, the algorithm applies $\textsc{Solve}$ to the systems $\mathcal N^{\mathrm{td}}_{I,|I|}$ for all $I\subseteq[R]$ with $|I|\ge2$ and all profiles of $(I,|I|)$.  If none of them has a solution, the supplied decomposition is essentially unique by Corollary~\ref{cor:criteria}\ref{item:criteria-uniqueness}.  If one of them has a solution $y'$, then $\textsc{Reconstruct}(y')$, as in Step~\ref{step:main-reconstruct}, gives a balanced localized replacement on $I$ with entries in a field $\Q(\theta)$.  It is term-disjoint by Theorem~\ref{thm:representation}\ref{item:soundness}, because $y'$ satisfies the inequalities \eqref{eq:td}, so it is a localized alternative.  Replacing the terms $P_r$, $r\in I$, by its terms gives an exact length-$R$ decomposition, which is not equivalent to the supplied decomposition by the proof of Theorem~\ref{thm:localized-witness}\ref{item:localized-uniqueness}.

\proofstep{running time and output length}
By Proposition~\ref{prop:one-shot-format}, Algorithm~\ref{alg:main} solves at most $G(R)(L+1)^{CR}$ polynomial systems, each with fewer than $2R^2$ variables, of length at most $G(R)(L+1)$, of degree at most $2R$, and of bit length at most $G(R)(L+1)^{C}$.  By Theorem~\ref{thm:sampling}, one call of $\textsc{Solve}$ costs at most $G(R)(L+1)^{CR^2}$ bit operations, and all calls together cost the same order after an increase of $C$.  Listing and forming the systems and the single call of $\textsc{Reconstruct}$ cost no more.  This gives \eqref{eq:running-time} and the bound on the output in Remark~\ref{rem:main-format}; the detailed accounting is in Appendix~\ref{app:complexity}.
\end{proof}

\begin{remark}
\label{rem:main-format}
The entries of the output decompositions in \ref{item:main-decomposition} and \ref{item:main-witnesses} can be irrational.  Every entry is given as a quotient $g_1(\theta)/g_0(\theta)$ as in Definition~\ref{def:exact-representation}\ref{item:quotient-bits}, with one real algebraic number $\theta$ for all entries of one decomposition, and $\theta$ is given by a polynomial $f$ and its Thom encoding. 
The degree of the polynomial $f$ of $\theta$ is at most $(2R)^{2C_{\mathrm{sa}}R^2}$ by \eqref{eq:output-degree}, with $C_{\mathrm{sa}}$ from Theorem~\ref{thm:sampling}.  The algorithm writes the output within the running time of Theorem~\ref{thm:main}, so this also bounds the bit length of the output.  These facts are established in Appendix~\ref{app:complexity}.  We have not tried to determine the best possible function $G$, see Section~\ref{sec:discussion}.  Every ingredient of Section~\ref{sec:one-shot} (the number of patterns and profiles, the length, degree, and coefficient size of the systems, and the arithmetic in $\Q(\theta)$) is bounded by $2^{\mathrm{poly}(R)}$ times a polynomial in $L$ of the stated form, see Remark~\ref{rem:field-bound}, so $G(R)=2^{\mathrm{poly}(R)}$ suffices.   
\end{remark}

\begin{remark}[Computational model]\label{rem:model}
Theorem~\ref{thm:main} is a statement in the bit model with rational input, and all running times count bit operations.  The algorithm of \citet{BasuPollackRoy2006} used in Theorem~\ref{thm:sampling} is stated by its number of arithmetic operations, and the exact output representation of Definition~\ref{def:exact-representation} and the degree bounds are bit-model notions.  The structural part of the reduction, namely the numbers of patterns, profiles, and variables, depends only on $R$, $d$, and the mode dimensions, and it does not depend on the model.  We do not carry out an analysis in the real-number model of \citet{BlumShubSmale1989}, in which the factor entries would be real numbers, and we make no claim there.
\end{remark}

\section{Sharpness of the effective-order and variable bounds}\label{sec:sharpness}

The exact reduction of Section~\ref{sec:normalized} separates the dimensions of the ambient tensor from the number of real variables. We now show that its effective-order and variable bounds are attained for both shortening and balanced exchanges at every supplied length $R\ge2$. Appendix~\ref{app:necessity} complements this by showing that the remaining ingredients of the reduction, namely the added rows $K_m$, irreducibility, nonnegative lifting, and term disjointness, cannot be dropped. Throughout this section we use the profiles and polynomial systems already defined in Section~\ref{sec:normalized}. For matrices, we identify $u\otimes v$ with $uv^{\mathsf T}$.

\subsection{Matrix families and cycle decompositions}

For $n\ge3$, let $P_n\in\Q^{n\times n}$ have the entries $P_n[i,i+1]=1$ for $i=1,\ldots,n-1$ and $P_n[n,1]=1$, and zeros elsewhere, that is,
\[
 P_n=\begin{pmatrix}0&1&&\\&\ddots&\ddots&\\&&0&1\\1&&&0\end{pmatrix}.
\]
We call the matrix
\begin{equation}\label{eq:Cn}
 C_n=\Id_n+P_n
\end{equation}
\emph{the incidence matrix of the cycle of length $n$} \citep[pp.~289--290]{GrossmanKulkarniSchochetman1994}.  Its rows correspond to the vertices $1,\ldots,n$ of the cycle and its columns to the edges $\{n,1\},\{1,2\},\ldots,\{n-1,n\}$.  An entry is $1$ if the vertex of its row is an endpoint of the edge of its column, and $0$ otherwise. We call
\begin{equation}\label{eq:cycle-row}
 C_n=\sum_{i=1}^{n}\be_iC_n[i,:]
 \qquad\text{and}\qquad
 C_n=\sum_{j=1}^{n}C_n[:,j]\be_j^{\mathsf T}
\end{equation}
the \emph{row decomposition} and the \emph{column decomposition} of $C_n$.

\begin{lemma}\label{lem:cycle}
For every $n\ge3$,
\begin{equation}\label{eq:Cn-ranks}
 \rankp(C_n)=n,
 \qquad
 \rank C_n=
 \begin{cases}
  n-1,&n\text{ even},\\
  n,&n\text{ odd}.
 \end{cases}
\end{equation}
Every minimal nonnegative decomposition of $C_n$ is equivalent to its row or its column decomposition.  These two decompositions are term-disjoint, and their exchange is irreducible.
\end{lemma}

\begin{proof} 
\pfpart{Ordinary rank}
The equations $C_nx=0$ read $x_i+x_{i+1}=0$. Hence
$x_i=(-1)^{i-1}x_1$, and the equation $x_n+x_1=0$ gives
$((-1)^{n-1}+1)x_1=0$. For even $n$ the kernel is the line spanned by
$(1,-1,\ldots,1,-1)^{\mathsf T}$; for odd $n$ it is zero. This proves the formula for $\rank C_n$.

\pfpart{Nonnegative rank and minimal decompositions}
In a nonnegative decomposition of $C_n$, every term $uv^{\mathsf T}$ has support $\supp(u)\times\supp(v)$ contained in $\supp(C_n)$. For $n\ge3$, two distinct rows of $C_n$ are both positive in at most one column.  Hence $\supp(u)\times\supp(v)$ lies in one row or in one column. Hence each term covers at most two of the $2n$ positive entries of $C_n$, so $\rankp(C_n)\ge n$, and the row decomposition gives equality.

Suppose $C_n=\sum_{j=1}^{n}u_jv_j^{\mathsf T}$ is a nonnegative decomposition with $n$ terms.  Nonnegativity gives
\[
 2n=|\supp(C_n)|
 =\Bigl|\bigcup_{j=1}^{n}\supp(u_jv_j^{\mathsf T})\Bigr|
\le\sum_{j=1}^{n}\bigl|\supp(u_jv_j^{\mathsf T})\bigr|\le2n.
\]
Equality forces each term to cover exactly two entries of one row or of one column and the supports
to be pairwise disjoint. The term covering $(1,1)$ covers row $1$ or column $1$.  If it covers row $1$, the term covering $(2,2)$ cannot cover $(1,2)$ as well, so it covers row $2$, and continuing in this way every term covers a row.  If it covers column $1$, the same argument shows that every term covers a column. Since each positive entry of $C_n$ lies in only one term, each term equals $1$ at its two positions.  Hence the decomposition is the row decomposition or the column decomposition, up to the order of its terms.  A term of the row decomposition has two nonzero columns, whereas a term of the column decomposition has only one, so the two decompositions have no common term, and therefore are term-disjoint.

\pfpart{Irreducibility}
A subexchange consists of index sets $I',J'\subseteq[n]$ with
\[
 \sum_{i\in I'}\be_iC_n[i,:]=\sum_{j\in J'}C_n[:,j]\be_j^{\mathsf T}.
\]
At a positive entry $(i,j)$, the left side is $1$ if $i\in I'$ and $0$ otherwise, and the right side is $1$ if $j\in J'$ and $0$ otherwise.  So for $i<n$ the entries $(i,i)$ and $(i,i+1)$ give $i\in J'\Leftrightarrow i\in I'\Leftrightarrow i+1\in J'$, and the entry $(n,n)$ gives $n\in I'\Leftrightarrow n\in J'$.  Hence $I'$ and $J'$ are both empty or both equal to $[n]$.
\end{proof}

For $R\ge2$ put
\[
 B_R=\begin{pmatrix}\Id_{R-1}\\ \mathbf1^{\mathsf T}\end{pmatrix}\in\Q^{R\times(R-1)},
 \qquad T_2=\begin{pmatrix}1&1\\0&1\end{pmatrix}.
\]
We supply the row decomposition $M=\sum_i\be_iM[i,:]$. A profile is \emph{feasible} if its polynomial system has a solution: $\mathcal N_{I,q}$ for shortening replacements and $\mathcal N^{\mathrm{td}}_{I,p}$ for term-disjoint balanced replacements.

\subsection{Sharpness of the effective-order and variable bounds}

\begin{theorem}\label{thm:attainment}
Let $R\ge2$.
\begin{enumerate}[label=\textup{(\roman*)}]
\item\label{item:attainment-short} Every shortening localized replacement for the row decomposition of $B_R$ has $I=[R]$ and $q=R-1$, and every feasible profile for this block has $\kappa=2R-3$.
\item\label{item:attainment-balanced} The row decomposition of $C_R$, or of $T_2$ if $R=2$, is minimal and not essentially unique.  Every term-disjoint balanced localized replacement for it has $I=[R]$, and every feasible profile for this block has $\kappa=2R-2$.
\end{enumerate}
In both parts $\kappa=p+q-2$, so the bounds of Theorem~\ref{thm:budget} and Remark~\ref{rem:counts} are attained.
\end{theorem}

\begin{proof}
\pfpart{Part~\ref{item:attainment-short}}
The column decomposition $B_R=\sum_{j=1}^{R-1}B_R[:,j]\be_j^{\mathsf T}$ has $R-1$ terms, and $\rank B_R=R-1$, so $\rankp(B_R)=R-1$.  Any $R-1$ rows of $B_R$ are linearly independent. Hence, for a proper subset $I$ of $[R]$, the sum $\sum_{i\in I}\be_iB_R[i,:]$ has rank $|I|$, so it has no shorter decomposition. Thus every shortening replacement has $I=[R]$ and therefore $q=R-1$. 

The exchange between the row and the column decomposition is irreducible, because a proper subexchange and its complement would split the $R$ terms of the row decomposition and the $R-1$ terms of the column decomposition into two parts, one of which would replace a proper subset of the former by fewer terms, which is impossible. Applied to this irreducible replacement, Theorem~\ref{thm:representation}\ref{item:completeness} gives a feasible profile of $([R],R-1)$.  In each mode the supplied factors span the whole space, that is, $(\sigma_1,\sigma_2)=(R,R-1)=(n_1,n_2)$, so the bounds $\sigma_m\le\delta_m\le n_m$ of \eqref{eq:profile-dimension-range} give $(\delta_1,\delta_2)=(R,R-1)$ for every profile.  Hence $\kappa=(R-1)+(R-2)=2R-3$.

\pfpart{Part~\ref{item:attainment-balanced} for $R\ge3$}
By Lemma~\ref{lem:cycle}, the row decomposition of $C_R$ is minimal, and the column decomposition is a term-disjoint minimal decomposition whose exchange with it is irreducible.  Hence the row decomposition is not essentially unique, and by Theorem~\ref{thm:representation}\ref{item:completeness} the block $([R],R)$ has a feasible profile.

Suppose that, for a nonempty proper subset $I$ of $[R]$, the terms $\be_iC_R[i,:]$ with $i\in I$ have a term-disjoint balanced replacement.  Replacing them gives a nonnegative decomposition of $C_R$ with $R$ terms, which is minimal and hence, by Lemma~\ref{lem:cycle}, equivalent to the row or to the column decomposition. It is not the row decomposition, because the replacement is term-disjoint.  It is also not the column decomposition, because it keeps a term $\be_iC_R[i,:]$ with $i\notin I$, and the row and column decompositions are term-disjoint.  Hence $I=[R]$.

Now take a feasible profile.  Then Theorem~\ref{thm:representation}\ref{item:soundness} gives a term-disjoint replacement of all terms of the row decomposition, that is, a minimal decomposition of $C_R$ term-disjoint from the row decomposition. By Lemma~\ref{lem:cycle} it is equivalent to the column decomposition, so its factors in mode~$2$ are positive multiples of distinct coordinate vectors. By \eqref{eq:reconstructed-term}, these factors are positive multiples of the columns of $\Lambda_2\bar Y_2$, hence
\[
 R=\rank(\Lambda_2\bar Y_2)\le\rank\bar Y_2\le\delta_2\le R.
\] 
The supplied factors span $\R^R$ in mode~$1$, so $\delta_1=R$ as well, and $\kappa=2R-2$.

\pfpart{Part~\ref{item:attainment-balanced} for $R=2$}
Since $\rank T_2=2$, its row decomposition $T_2=\be_1(1,1)+\be_2(0,1)$ is minimal. Its column decomposition $T_2=(1,0)^{\mathsf T}\be_1^{\mathsf T}+(1,1)^{\mathsf T}\be_2^{\mathsf T}$ is term-disjoint from it, and their exchange is irreducible, because a proper subexchange would consist of one term on each side with equal values, but the two decompositions have no common term. Hence the row decomposition is not essentially unique, and the block $([2],2)$ has a feasible profile. A single term has no term-disjoint replacement by one term, so $I=[2]$. The supplied factors span $\R^2$ in both modes, so $(\delta_1,\delta_2)=(2,2)$ and $\kappa=2$.

\pfpart{The bounds}
In part~\ref{item:attainment-short}, $p=R$ and $q=R-1$; in part~\ref{item:attainment-balanced}, $p=q=R$.  In both cases $\kappa=p+q-2$, which attains the bound of Theorem~\ref{thm:budget}.  By \eqref{eq:exact-count}, the systems have $q\kappa+q-1$ variables, namely $2R^2-4R+1$ and $2R^2-R-1$, which are the bounds of Remark~\ref{rem:counts}.
\end{proof}

\begin{remark}[Scope of the sharpness statement]\label{rem:sharpness-scope}
Theorem~\ref{thm:attainment} proves sharpness of the effective-order bound and of the variable counts for the semialgebraic representation developed here. It is not a lower bound on the number of variables required by every possible exact algorithm.
\end{remark}

\section{Moment tensors, quadrature, and algebraic output complexity}\label{sec:moments}

Moment tensors connect nonnegative decomposition with positive quadrature. This connection provides explicit families whose rank and uniqueness can be determined, illustrates exact reconstruction by Algorithm~\ref{alg:main}, and yields a factorial lower bound on the degree of a common field for exact outputs. The quadrature ingredients are classical; their tensor interpretation brings together rank, identifiability, and algebraic output complexity.

\subsection{Moment tensors and positive quadrature}\label{subsec:moment}

For $t\in\R$ put $v(t)=(1-t,\,t)^{\mathsf T}$.  For a finite positive Borel measure $\mu$ on $[0,1]$ and $d\ge2$, the \emph{moment tensor}
\begin{equation}\label{eq:moment-tensor}
 M_d(\mu)=\int v(t)^{\otimes d}\,d\mu(t)\in(\R^2)^{\otimes d}
\end{equation}
has the entry $\int(1-t)^{d-s}t^s\,d\mu(t)$ at every index $(i_1,\ldots,i_d)$ with exactly $s$ coordinates $i_m$ equal to $2$. Nodes $\xi_1,\ldots,\xi_q\in[0,1]$ and weights $w_1,\ldots,w_q>0$ form a \emph{positive quadrature rule of degree $d$} for $\mu$ if $\sum_jw_jh(\xi_j)=\int h\,d\mu$ for every polynomial $h$ of degree at most $d$.  The entries of $v(t)^{\otimes d}$ are polynomials of degree $d$ in $t$, and they span all polynomials of degree at most $d$.  Therefore every such rule gives the nonnegative decomposition
\begin{equation}\label{eq:quadrature}
 M_d(\mu)=\sum_{j=1}^{q}w_j\,v(\xi_j)^{\otimes d}
\end{equation}
of length $q$, and every nonnegative decomposition of $M_d(\mu)$ whose terms are symmetric tensors arises in this way.  We put $\mathcal P_d=\int_0^1v(t)^{\otimes d}\,dt$, the moment tensor of the Lebesgue measure on $[0,1]$. By Euler's Beta integral \citep[(5.12.1)]{DLMF}, the entry of $\mathcal P_d$ with $s$ coordinates equal to $2$ is
\[
 \int_0^1(1-t)^{d-s}t^s\,dt=\frac{\Gamma(s+1)\Gamma(d-s+1)}{\Gamma(d+2)}=\frac{s!\,(d-s)!}{(d+1)!}.
\]
Here $\Gamma$ is the Gamma function. This entry is the probability of a fixed sequence of $d$ draws, $s$ of them black, from P\'olya's urn started with one white and one black ball \citep[Section~V.2]{Feller1968}.

For $\rho\ge1$, the \emph{$\rho$-point Gauss--Legendre rule} is the Gaussian quadrature rule with $\rho$ nodes for the Lebesgue measure on $[0,1]$.  Its nodes are the zeros of the shifted Legendre polynomial of degree $\rho$, which are distinct and lie in $(0,1)$, and it is a positive quadrature rule of degree $2\rho-1$ \citep[\S3.5(v), \S18.2(vi), Table~18.3.1 and (18.7.10)]{DLMF}.

The following statement restates classical facts on Gaussian quadrature \citep[\S1.4]{Gautschi2004} in terms of nonnegative decompositions. Its odd-order uniqueness conclusion applies to all length-$\rho$ tensor decompositions, without imposing symmetry on the competing terms.

\begin{proposition}\label{prop:polya}
Let $d\ge2$ and $\rho=\lfloor d/2\rfloor+1$.  Then $\rankp(\mathcal P_d)=\rank(\mathcal P_d)=\rho$.
\begin{enumerate}[label=\textup{(\roman*)}]
\item\label{item:polya-odd} If $d=2\rho-1$, then $\mathcal P_d$ has exactly one nonnegative decomposition of length $\rho$.  It is \eqref{eq:quadrature} for the $\rho$-point Gauss--Legendre rule.
\item\label{item:polya-even} If $d=2\rho-2$, then $\mathcal P_d$ has infinitely many pairwise inequivalent nonnegative decompositions of length $\rho$.
\end{enumerate}
\end{proposition}

\begin{proof}[Proof sketch]
For the lower bound, a flattening of $\mathcal P_d$ is transformed by the evaluation map of the monomials into a Hankel matrix of moments of Lebesgue measure, whose leading $\rho\times\rho$ block is positive definite.  The upper bound is the Gauss--Legendre rule.  For odd $d=2\rho-1$, every set $I\subseteq[\rho]$ with $|I|\ge2$ has $\beta(I)=d\ge2|I|-1$, so uniqueness follows from the Lovitz--Petrov criterion.  For even $d=2\rho-2$, the quadrature conditions are $d+1=2\rho-1$ equations in $2\rho$ unknowns with a Jacobian of full rank at the Gauss--Legendre rule, and the implicit function theorem gives a curve of positive quadrature rules.  The full proof is in Appendix~\ref{app:proofs}.
\end{proof}

\subsection{Exact reconstruction and the binomial-mixture interpretation}\label{subsec:quadrature-reconstruction}

\paragraph{From Simpson's rule to Gaussian quadrature.}
Simpson's rule, with nodes $0,\frac12,1$ and weights $\frac16,\frac23,\frac16$, is a positive quadrature rule of degree $3$ for the Lebesgue measure \citep[(3.5.6)]{DLMF}.  By \eqref{eq:quadrature} it gives the rational supplied decomposition
\[
 \mathcal P_3=\tfrac16\,v(0)^{\otimes3}+\tfrac23\,v(\tfrac12)^{\otimes3}+\tfrac16\,v(1)^{\otimes3},
\]
with $R=3$, where the weights are absorbed into the mode-$1$ factors.  By Proposition~\ref{prop:polya}, $\rankp(\mathcal P_3)=2$ and the minimum-length decomposition is unique.  The following example runs Algorithm~\ref{alg:main} on this input.

\begin{example}[Simpson's rule]\label{ex:simpson}
\emph{Phase~I.}  In every mode any two of $v(0),v(\frac12),v(1)$ span $\R^2$, so $\sigma_m=2$ and $\beta(I)=3$ for every $I\subseteq[3]$ with $|I|\ge2$.  Every profile of a block $(I,q)$ satisfies $\beta(I)\le\kappa\le p+q-2$ by \eqref{eq:profile-dimension-range} and \eqref{eq:profile-budget}, so $p+q\ge5$, and among the shortening blocks only $([3],2)$ qualifies.  Since $\sigma_m\le\delta_m\le n_m=2$, its only profile has $\delta_m=2$, $H=[3]$, and $K_m=\varnothing$.  There are no omitted rows and no inactive modes, the core is $\mathcal P_3$ itself, and $\mathcal N_{[3],2}$ consists of the core equation
\[
 \mathcal P_3=\sum_{j=1}^{2}\omega_j\bigotimes_{m=1}^{3}\begin{pmatrix}y_{m1j}\\1-y_{m1j}\end{pmatrix},
 \qquad
 \omega_2=1-\omega_1,
\]
and the inequalities $0\le y_{m1j}\le1$ and $0<\omega_1<1$.  These are $8$ equations in $q\kappa+q-1=7$ variables.

\emph{Phase~I, solving.}  A lexicographic Gr\"obner basis \citep{Buchberger1965} of the core equation, for $y_{111}\succ y_{112}\succ y_{211}\succ y_{212}\succ y_{311}\succ y_{312}\succ\omega_1$, consists up to constant factors of
\[
 y_{m11}+y_{312}-1\ (m\in[3]),\qquad y_{m12}-y_{312}\ (m\in[2]),\qquad 6y_{312}^2-6y_{312}+1,\qquad 2\omega_1-1.
\]
Hence $\omega_1=\omega_2=\frac12$, and with $\xi_\pm:=\frac12\pm\frac{\sqrt3}{6}$ the core equation has exactly two solutions, $\bar Y_1=\bar Y_2=\bar Y_3=[\,v(\xi_+)\ \ v(\xi_-)\,]$ and the same matrices with the two columns interchanged.  Both satisfy the inequalities and give the same decomposition.  For the first one, $\theta=y_{312}=\xi_+$ is given by $f(T)=6T^2-6T+1$ and its Thom encoding $(\sgn f'(\theta),\sgn f''(\theta))=(+,+)$, and every coordinate is $g_i(\theta)/2$ with $g_i\in\{2-2T,\,2T,\,1\}$.  All coordinates lie in $\Q(\theta)=\Q(\sqrt3)$.

\emph{Phase~II.}  Since $n_m=\delta_m$, every coefficient matrix $\Lambda_m$ is the identity, and the output is
\[
 \mathcal P_3=\tfrac12\,v(\xi_+)^{\otimes3}+\tfrac12\,v(\xi_-)^{\otimes3}.
\]
Thus $\rankp(\mathcal P_3)=2$ and the supplied decomposition is not minimal, so Phase~III is not run and the algorithm reports \textsc{not unique}, with this decomposition as the shorter competitor.  The output is the two-point Gauss--Legendre rule, whose nodes are the zeros of the shifted Legendre polynomial $6t^2-6t+1=f(t)$.
\end{example}

\begin{remark}
Comparing the entries at an index with $s$ coordinates equal to $2$, and multiplying by $\binom ds$, we see that a decomposition $\mathcal P_d=\sum_{j=1}^{\rho}w_j\,v(\xi_j)^{\otimes d}$ as in~\eqref{eq:quadrature} amounts to
\[
 \sum_{j=1}^{\rho}w_j\binom ds\xi_j^s(1-\xi_j)^{d-s}=\binom ds\frac{s!\,(d-s)!}{(d+1)!}=\frac1{d+1},\qquad s=0,\ldots,d .
\]
That is, it represents the uniform distribution on $\{0,\ldots,d\}$ as a mixture of the binomial distributions with $d$ trials and success probabilities $\xi_j$, with weights $w_j$. Mixtures of $\rho$ binomial distributions with $d$ trials are identifiable if and only if $d\ge2\rho-1$ \citep{Teicher1963}.  For the uniform distribution, part~\ref{item:polya-odd} of Proposition~\ref{prop:polya} gives uniqueness at $d=2\rho-1$, and part~\ref{item:polya-even} gives non-uniqueness at $d=2\rho-2$.
\end{remark}

\begin{remark}
For $d=4$, for instance, the rule with nodes $0,\frac{6\pm\sqrt6}{10}$ and weights $\frac19,\frac49\mp\frac{\sqrt6}{36}$ is a positive quadrature rule of degree $4$.  It and the three-point Gauss--Legendre rule give two inequivalent minimal nonnegative decompositions of $\mathcal P_4$.  For $d=5$, Boole's rule \citep[(25.4.14)]{AbramowitzStegun1964}, with nodes $0,\frac14,\frac12,\frac34,1$ and weights $\frac7{90},\frac{32}{90},\frac{12}{90},\frac{32}{90},\frac7{90}$, is a positive quadrature rule of degree $5$ and gives a supplied decomposition of $\mathcal P_5$ with $R=5$.  By Proposition~\ref{prop:polya}\ref{item:polya-odd}, Algorithm~\ref{alg:main} returns the three-point Gauss--Legendre rule on this input.
\end{remark}

\subsection{Necessary algebraic output complexity}\label{subsec:output-lower}\label{subsec:degree}

The output decompositions of Theorem~\ref{thm:main} have entries in a field $\Q(\theta)$, and the polynomial $f$ of $\theta$ has degree at most $(2R)^{2C_{\mathrm{sa}}R^2}$, see Remark~\ref{rem:main-format}.  If $\T$ has only one minimal nonnegative decomposition, then $\Q(\theta)$ contains the entries of its terms.  Hence $\deg f$ is at least the degree of the field which these entries generate \citep[Proposition~1.20]{MilneFT}.

\begin{theorem}\label{thm:degree}
For every integer $\rho\ge2$ there is a tensor $\T\in(\Q^2)^{\otimes(2\rho-1)}$ such that
\begin{enumerate}[label=\textup{(\roman*)}]
\item $\rankp(\T)=\rho$, and $\T$ has exactly one nonnegative decomposition of length $\rho$.  The entries of the terms of this decomposition generate a field of degree $\rho!$ over $\Q$.
\item $\T$ has a rational nonnegative decomposition of length $R$ for every $R\ge2\rho$.
\end{enumerate}
\end{theorem}

\begin{proof}[Proof of the field degree in~(i)]
The tensor is the moment tensor $\T=M_{2\rho-1}(\mu)$ of the measure $\mu$ with unit mass at each of the points $\xi_j=\zeta_j/c$, where $\zeta_1<\cdots<\zeta_\rho$ are the zeros of the Laguerre polynomial $L_\rho(x)=\sum_{k=0}^{\rho}\binom{\rho}{k}(-x)^k/k!$ and $c$ is an integer larger than all of them.  Its moments are power sums of the zeros of a polynomial with rational coefficients, so they and the entries of $\T$ are rational.  The construction of $\T$, the rank and uniqueness assertions in~(i), and the rational decomposition in~(ii) are proved in Appendix~\ref{app:proofs}.  We prove here that the entries of the terms generate a field of degree $\rho!$.  The zeros of $L_\rho$ generate a field of degree $\rho!$ \citep[pp.~443--444]{Schur1930}.

  Put $d=2\rho-1$.  Every entry of the $j$th term $v(\xi_j)^{\otimes d}$ is a product of $d$ factors, each equal to $1-\xi_j$ or $\xi_j$, so it lies in $\Q(\xi_j)$. Conversely, the entries of this term whose first index is $2$ sum to
\[
 \sum_{i_2,\ldots,i_d\in\{1,2\}}\xi_j\,v(\xi_j)_{i_2}\cdots v(\xi_j)_{i_d}=\xi_j\bigl((1-\xi_j)+\xi_j\bigr)^{d-1}=\xi_j .
\]
Hence the entries of all terms and the numbers $\xi_1,\ldots,\xi_\rho$ generate the same field.  The zeros of $L_\rho$ are $c\xi_1,\ldots,c\xi_\rho$ with $c$ a positive integer, so they generate the same field, which has degree $\rho!$.

The same holds for the factor entries, which is the format of an output decomposition.  By Lemma~\ref{lem:product-equality}, the factors of the $j$th term of the unique length-$\rho$ decomposition are $b^{(m)}=\lambda_mv(\xi_j)$ for positive scalars $\lambda_m$, which may be irrational.  The ratio $b^{(m)}[2]/b^{(m)}[1]=\xi_j/(1-\xi_j)$ does not depend on $\lambda_m$, so every field containing the entries of these factors contains $\xi_j$.  Hence every field containing the entries of the factor matrices of a length-$\rho$ decomposition of $\T$ contains $\xi_1,\ldots,\xi_\rho$ and has degree at least $\rho!$.
\end{proof}

\begin{corollary}[Superpolynomial common-field degree]\label{cor:output-degree-lower}
For every $R\ge4$ there is a rational supplied decomposition of length $R$ such that every field containing all entries of the factor matrices of an exact minimum-length decomposition has degree at least $\lfloor R/2\rfloor!$ over $\Q$. Consequently, no polynomial in $R$ bounds the degree of the common output field in Theorem~\ref{thm:main}.
\end{corollary}

\begin{proof}
Take $\rho=\lfloor R/2\rfloor$ and a tensor from Theorem~\ref{thm:degree}, supplied by a rational decomposition of length $R$. By the proof of Theorem~\ref{thm:degree}, every field containing the entries of the factor matrices of its unique minimum-length decomposition has degree at least $\rho!$. Therefore the common field $\Q(\theta)$ of any exact output has degree at least $\rho!$. Together with \eqref{eq:output-degree},
\[
 2^{\Theta(R\log R)}=\lfloor R/2\rfloor!\le\deg f\le(2R)^{2C_{\mathrm{sa}}R^2}=2^{O(R^2\log R)}.
\]
With the normalization $\lambda_m=1$, the entries of each individual term lie in $\Q(\xi_j)$, whose degree is at most $\rho$; the factorial lower bound is caused by requiring one field to contain the entries of the entire decomposition.
\end{proof}

\section{Computational evidence}\label{sec:experiments}

\subsection{Setting and exact certification}

The results of this paper are theoretical; this section provides computational evidence that the bounded-variable formulation can resolve instances left inconclusive by existing sufficient criteria.  We also record structural features of the reduction, such as added rows and larger exchanges, that occur in these instances.  We do not claim that the implementation is competitive or that the timings reflect the worst-case bound.

We ran an implementation of Algorithm~\ref{alg:main}, with the general-purpose solver Z3 in the role of $\textsc{Solve}$, on the 114 supplied decompositions which remain inconclusive in \citet[Sections~8.2 and~8.3]{WangYuan2026}, all with $d=3$, $n_1=n_2=n_3=6$, and $R=5$.\footnote{The data are available at \url{https://github.com/hw3129/exact-algorithm}.} The support of each factor is drawn
coordinatewise from a Bernoulli distribution, and the entries on the
support are independent uniform integers in $\{1,\ldots,999\}$.
No competing decomposition was found for these 114 instances in
\citet[Section~8.3]{WangYuan2026}. We label them $001,\ldots,114$ as in the released data.

The implementation uses exact rational preprocessing and Z3 for the remaining real polynomial systems. Appendix~\ref{app:experiments} describes the rank bounds, redundant-profile elimination, support splitting, and rank-one-curve parametrization, with proofs that these simplifications preserve feasibility. The reported witnesses are exact outputs of the implementation's exact solver, and infeasibility is decided by that solver procedure, including the feasibility-preserving reductions described in the appendix.  Z3 is not the procedure of Theorem~\ref{thm:sampling}, so the running times below are empirical and are not governed by the bound \eqref{eq:running-time}; the experiments test the reduction, not the worst-case complexity.

\subsection{Results}\label{subsec:exp-results}

\begin{figure}[t]
\centering
\includegraphics[width=\linewidth]{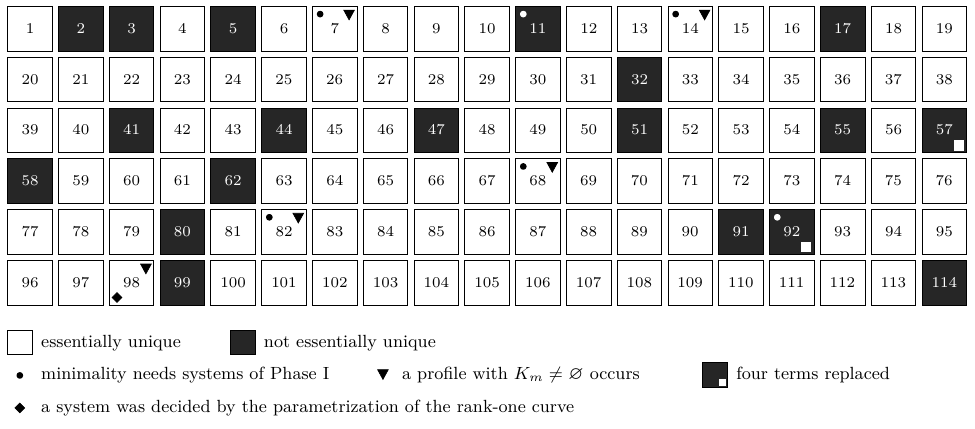}
\caption{The 114 instances, numbered without leading zeros; all are minimal.  White/black cells: essentially unique/nonunique.  Dots: Phase~I systems needed; triangles: $K_m\ne\varnothing$ occurs; small white squares: four rather than three terms replaced; diamond: Proposition~\ref{prop:exp-curve} used.}
\label{fig:exp-outcomes}
\end{figure}

Table~\ref{tab:exp-summary} summarizes the computation. Across the 114 instances, the implementation considered 525 profiles with $5$--$44$ variables per system; the maximum 44 attains the bound in Remark~\ref{rem:counts} for $p=q=5$. One redundant profile was skipped. Of the remaining 524 systems, 505 were proved infeasible and 19 yielded rational decompositions inequivalent to the corresponding supplied decompositions. Because Phase~III stops at the first feasible profile for an instance, these 19 systems came from 19 distinct inputs. Thus all 114 supplied decompositions are minimal, 95 are essentially unique, and 19 are nonunique; see Figure~\ref{fig:exp-outcomes}.
Notably, two of the 19 nonunique instances, 057 and 092, require four-term localized exchanges: every term-disjoint system on at most three terms is infeasible for these inputs.

\begin{table}[H]
\centering
\caption{Headline results for the 114 previously inconclusive supplied decompositions. The 525 profiles include one redundant profile skipped before solving.}
\label{tab:exp-summary}
\begin{tabular}{lr}
\toprule
Quantity & Result\\
\midrule
Supplied decompositions & 114\\
Profiles considered & 525\\
Redundant profiles skipped & 1\\
Infeasible systems & 505\\
Feasible systems & 19\\
Minimal supplied decompositions & 114\\
Essentially unique & 95\\
Minimal but nonunique & 19\\
Variables per system & $5$--$44$\\
Median running time & $0.12$ seconds\\
Maximum running time & $283$ seconds\\
\bottomrule
\end{tabular}
\end{table}

\subsection{Structural consequences}

The largest system has 44 variables, attaining the bound for $p=q=5$ in Remark~\ref{rem:counts}. Each of the 19 feasible systems yields a rational length-five decomposition inequivalent to the supplied one. Seventeen witnesses replace three terms and use 14 variables; the other two, for instances 057 and 092, replace four terms and use 27 variables. For these two inputs every term-disjoint system on at most three terms is infeasible. Thus witnesses to nonuniqueness can require larger localized exchanges even at this small supplied length.

The role of the inequalities is visible in the subsums illustrated in Figure~\ref{fig:exp-families}, whose coordinate descriptions and exact families remain in Appendix~\ref{subsec:exp-families}. The subsums for instances 098 and 041 have the same mode ranks up to permutation and the same numbers of variables and equations, yet the former is essentially unique while the latter admits a one-parameter family of inequivalent decompositions. The difference comes from the nonnegativity and lifting inequalities, rather than from ranks alone.

Only instance 098 required the rank-one-curve parametrization of Proposition~\ref{prop:exp-curve}. Without it, two systems were not decided within one hour; with it, each was proved infeasible in under 15 seconds. The detailed profile counts and timing analysis are in Appendix~\ref{subsec:exp-details}. These instances provide computational evidence that the bounded-variable formulation can decide cases left open by existing sufficient criteria while retaining the structural distinctions responsible for the exact decisions.

\section{Discussion, related work, and open problems}\label{sec:discussion}

\paragraph{Nonnegative tensor geometry and identifiability.}
\citet{QiComonLim2016} study the semialgebraic geometry of tensors of nonnegative rank at most $r$, including typical ranks, identifiability, and uniqueness. This geometry explains why real-algebraic methods are natural, but does not by itself control the number of variables in an exact decision procedure. Deterministic identifiability results such as the theorem of \citet{Kruskal1977}, the Lovitz--Petrov criterion \citep{LovitzPetrov2023}, and the positive-scattering criterion of \citet{WangYuan2026} give sufficient conditions under which a supplied decomposition is minimal or unique. Their role is complementary to the complete decision-and-search procedure developed here.

The connection with \citet{WangYuan2026} is especially direct. That paper certifies minimality when $\beta(I)+\tau(I)\ge2|I|-2$ for every $I\subseteq[R]$ with $|I|\ge2$, and uniqueness when $\beta(I)+\tau(I)\ge2|I|-1$. Here $\beta(I)=\sum_m(\sigma_m-1)$ is the Lovitz--Petrov budget, with $\sigma_m=\rank A_I^{(m)}$, and $\tau(I)$ is the positive-scattering term; see Remark~\ref{rem:scattering}. For every profile on $I$, \eqref{eq:profile-dimension-range} gives $\kappa=\beta(I)+\sum_m(\delta_m-\sigma_m)\ge\beta(I)$. Hence sufficiently large $\beta(I)$ already makes the relevant profile lists empty. With the additional condition $\sum_{m\in H}|K_m|\ge\tau(I)$ of Remark~\ref{rem:tau-profile}, the same argument gives $\kappa\ge\beta(I)+\tau(I)$. Thus whenever the sufficient criterion of \citet{WangYuan2026} applies, Algorithm~\ref{alg:main} certifies the corresponding conclusion without invoking real-algebraic solving; when it fails, the algorithm searches the remaining low-effective-order obstructions.

\paragraph{Nonnegative matrix rank.}
A nonnegative matrix factorization is exactly an order-two nonnegative CP decomposition:
\[
M=UV^{\mathsf T}=\sum_{r=1}^{R}U[:,r]V[:,r]^{\mathsf T}.
\]
Thus Theorem~\ref{thm:main} immediately gives the following specialization.

\begin{corollary}\label{cor:nmf}
Fix $R$. Given a rational factorization
\[
M=UV^{\mathsf T},
\qquad
U\in\Qpos^{m\times R},
\quad
V\in\Qpos^{n\times R},
\]
with no zero factor column, there is a deterministic algorithm running in
\[
G(R)\,\bits(U,V)^{CR^2}
\]
time that computes $\rankp(M)$, outputs an exact minimum-length NMF, decides whether the supplied factorization is minimal and essentially unique, outputs an exact shorter factorization when minimality fails, and, when the supplied factorization is minimal but not essentially unique, outputs an exact inequivalent length-$R$ factorization.
\end{corollary}

In particular, when the supplied factorization is minimal, Corollary~\ref{cor:nmf} allows $\rank M<\rankp(M)=R$; no full-rank assumption is imposed. The variable counts at small supplied lengths are given in Remark~\ref{rem:counts}.

Uniqueness of nonnegative matrix factorizations is classically studied through separability-type conditions \citep{DonohoStodden2003,LaurbergEtAl2008}; for $d=2$ the uniqueness criterion of \citet{WangYuan2026} is two-sided separability, and Corollary~\ref{cor:nmf} decides essential uniqueness without any such hypothesis.

For comparison, exact nonnegative matrix factorization is NP-hard \citep{Vavasis2009}, and \citet{CohenRothblum1993} observed its decidability by quantifier elimination over the reals. \citet{AroraGeKannanMoitra2012} gave an exact algorithm for every fixed target rank and proved, under the exponential time hypothesis, that deciding $\rankp(M)\le r$ cannot be done in time $(mn)^{o(r)}$ in their input model. \citet{Moitra2016} sharpened the upper bound to $(mn)^{O(r^2)}$, up to bit-complexity and parameter factors, by reducing a naive $r(m+n)$-variable formulation to one with $2r^2$ real variables. This variable-reduction viewpoint is particularly close to ours: real-algebraic feasibility becomes useful only after the number of variables has been separated from the ambient dimensions.

Corollary~\ref{cor:nmf} has an $R^2$ exponent in the supplied length while additionally deciding essential uniqueness and returning exact algebraic competitors. The input models are nevertheless different. The algorithms of \citet{AroraGeKannanMoitra2012} and \citet{Moitra2016} receive $M$ and a target rank $r$, whereas our algorithm receives a factorization of length $R$; its parameter is the supplied length, not the unknown value $\rankp(M)$. Applying our theorem to a trivial factorization of an arbitrary matrix therefore does not recover Moitra's parameter dependence.

\paragraph{Exact algebraic output.}
Irrationality is unavoidable already for rational nonnegative matrix factorization inputs \citep{ChistikovEtAl2017,Shitov2016}. Theorem~\ref{thm:degree} and Corollary~\ref{cor:output-degree-lower} show a stronger phenomenon in the supplied tensor setting: although each term in the unique minimal decomposition may lie in a field of degree only $O(R)$, a single field containing all terms can require degree $\lfloor R/2\rfloor!$. Thus the algebraic output representation in Theorem~\ref{thm:main} is not merely a technical replacement for rational output.

\paragraph{Unrestricted tensor decomposition.}
Here $\rank(\T)$ denotes ordinary tensor rank over the indicated field, whereas $\rankp(\T)$ denotes nonnegative tensor rank. Tensor rank and many related tensor problems are NP-hard in general \citep{HillarLim2013}, and the algorithmic landscape for ordinary tensor rank in the low-rank regime has changed substantially in recent years. \citet{BhargavaSarafVolkovich2021} gave reconstruction algorithms for constant-rank tensors; \citet{PelegShpilkaVolk2024} extended reconstruction beyond constant rank; and \citet{BhargavaShringi2025} gave a deterministic $2^{k^{O(1)}}\operatorname{poly}(n,d)$ algorithm for decomposing a $d$-way tensor of ordinary rank at most $k$ over $\mathbb R$ or $\mathbb C$. These results concern decomposition over a field and do not impose nonnegativity, so they do not decide nonnegative rank: in general $\rank(\T)\le\rankp(\T)$, possibly strictly, and an unrestricted rank decomposition need not yield a nonnegative competitor. Accordingly, the FPT result for unrestricted tensor decomposition and the XP result here concern different problems and input models; neither implies a parameterized complexity separation between them.

\paragraph{Open questions.}
Theorem~\ref{thm:attainment} establishes sharpness for the shortening and balanced families, but the general equality cases of $\kappa=p+q-2$ remain to be characterized. Can additional geometric or positivity assumptions strengthen the effective-order bound? The examples in Section~\ref{sec:sharpness} show that uniform improvements cannot follow from supplied length alone. Structural questions also remain about which localized exchanges control nonuniqueness and when the feasible sets contain positive-dimensional families, as in Sections~\ref{sec:moments} and~\ref{sec:experiments}.

Theorem~\ref{thm:main} places the supplied-decomposition problem in XP parameterized by $R$.  The following questions concern its parameterized complexity.
\begin{itemize}
\item[$\bullet$] Is the problem of Theorem~\ref{thm:main} fixed-parameter tractable in $R$, that is, can it be solved in time $G(R)L^{O(1)}$?
\item[$\bullet$] Can the exponent $O(R^2)$ in Theorem~\ref{thm:main} be reduced, for example to $O(R)$?  Is there a matching parameterized lower bound?
\item[$\bullet$] What is the smallest function $G$ for which Theorem~\ref{thm:main} holds?  The proof gives $G(R)=2^{\mathrm{poly}(R)}$.  Can $G$ be taken singly exponential in a small power of $R$?
\item[$\bullet$] The ETH lower bound of \citet{AroraGeKannanMoitra2012} applies to a different matrix input model and does not settle these questions.  A further question concerns inputs without a supplied decomposition.
Can $\rankp(\T)$ be computed with complexity parameterized by $\rankp(\T)$ alone, when only the tensor $\T$ is given?
\end{itemize}

A trivial decomposition can have length as large as $\min_m\prod_{m'\ne m}n_{m'}$, so Theorem~\ref{thm:main} does not yield a useful parameterization by $\rankp(\T)$.  Finally, can the special sparse, multilinear, and determinantal structure of the systems be exploited to reduce the variable count or to extend exact computation beyond the supplied lengths currently accessible in Section~\ref{sec:experiments}?

More broadly, the results suggest that exact certification is governed by the geometry of low-dimensional positive exchanges: sharp structural bounds lead to bounded-variable representations, and those representations lead to exact algorithms and computationally effective certificates.

\appendix

\section{Limits of matricization-based certification}\label{app:matricization}

Let $d=3$, and let $\T_{(1)}$ be the matricization of $\T$, as in Section~\ref{sec:introduction}, whose rows are indexed by mode~$1$ and whose columns are indexed by the pairs $(i_2,i_3)\in[n_2]\times[n_3]$ in lexicographic order.  In this section, for $x\in\R^{n_2}$ and $y\in\R^{n_3}$, we identify $x\otimes y$ with the column vector in $\R^{n_2n_3}$ whose entries $x[i_2]y[i_3]$ are ordered like the columns of $\T_{(1)}$.  Every nonnegative decomposition $\T=\sum_{j=1}^{q}b_j^{(1)}\otimes b_j^{(2)}\otimes b_j^{(3)}$ gives the nonnegative factorization
\begin{equation}\label{eq:mat-factorization}
\T_{(1)}=UV^{\mathsf T},
\qquad
U=[b_1^{(1)},\ldots,b_q^{(1)}],
\qquad
V=[b_1^{(2)}\otimes b_1^{(3)},\ldots,b_q^{(2)}\otimes b_q^{(3)}].
\end{equation}

Since a tensor is determined by its matricization, a shorter decomposition of $\T$ gives a shorter factorization of $\T_{(1)}$, and an inequivalent decomposition of $\T$ of length $R$ gives an inequivalent factorization of $\T_{(1)}$ of length $R$.  This gives the following proposition.

\begin{proposition}\label{prop:matricization}
If the factorization \eqref{eq:mat-factorization} given by the supplied decomposition is minimal or essentially unique, then so is the supplied decomposition.
\end{proposition}

However, we will see in Example~\ref{ex:matricization} that the supplied tensor decomposition can still be minimal and essentially unique even if the matrix factorization \eqref{eq:mat-factorization} which it gives is neither minimal nor essentially unique. We say that a nonnegative factorization $\T_{(1)}=UV^{\mathsf T}$ satisfies the \emph{rank-one condition} if every column of $V$, written as an $n_2\times n_3$ matrix, has rank at most one.

\begin{proposition}\label{prop:rank-one}
The following two statements hold.
\begin{enumerate}[label=\textup{(\roman*)}]
\item\label{item:rank-one-mat} Every factorization \eqref{eq:mat-factorization} satisfies the rank-one condition.
\item\label{item:rank-one-converse} If a nonnegative factorization $\T_{(1)}=UV^{\mathsf T}$ satisfies the rank-one condition, then its nonzero terms are the matricizations of the terms of a nonnegative decomposition of $\T$.
\end{enumerate}
\end{proposition}

\begin{proof}
Written as an $n_2\times n_3$ matrix, $b_j^{(2)}\otimes b_j^{(3)}$ is $b_j^{(2)}b_j^{(3)\mathsf T}$, which proves part~\ref{item:rank-one-mat}.  For part~\ref{item:rank-one-converse}, every nonzero column $V[:,j]$ has the form $x_j\otimes y_j$ with nonnegative $x_j$ and $y_j$ \citep[Lemma~1(ii)]{WangYuan2026}.  Hence every nonzero term $U[:,j]V[:,j]^{\mathsf T}$ is the matricization of the nonnegative rank-one tensor $U[:,j]\otimes x_j\otimes y_j$.  Since the sum of these matricizations is $\T_{(1)}$, the sum of these tensors is $\T$.
\end{proof}

Moitra decides whether $\rankp(M)\le q$ for a nonnegative matrix $M$ by a formulation with at most $2q^2$ real variables \citep[Section~2.4]{Moitra2016}.  Every solution of this formulation gives a nonnegative factorization $M=UV^{\mathsf T}$ in which $U$ and $V$ have $q$ columns \citep[Lemma~2.15]{Moitra2016}.  In this factorization, the $i$th row of $V$ is a function of the $i$th column of $M$, and this function is the same for all $i$ \citep[Claim~2.13]{Moitra2016}.  Hence two equal columns of $M$ give two equal rows of $V$.  For an arbitrary nonnegative factorization $M=UV^{\mathsf T}$, however, the $i$th and $k$th columns of $M$ can be equal while the $i$th and $k$th rows of $V$ differ.  Nevertheless, every nonnegative factorization $M=UV^{\mathsf T}$ in which $U$ and $V$ have $q$ columns can be transformed into another such factorization which is given by a solution \citep[Lemmas~2.6 and~2.14]{Moitra2016}.  Hence the formulation has a solution if and only if $M$ has a nonnegative factorization in which $U$ and $V$ have $q$ columns, that is, if and only if $\rankp(M)\le q$ \citep[Theorem~2.16]{Moitra2016}.  However, for $M=\T_{(1)}$, this transformation can break the rank-one condition, and then the new factorization is not given by any nonnegative decomposition of $\T$.

We consider two approaches which try to decide with this formulation whether the supplied decomposition is minimal.
\begin{enumerate}[label=\textup{(\alph*)}]
\item\label{item:first-approach} The first approach computes $\rankp(\T_{(1)})$.  If $\rankp(\T_{(1)})=R$, then the supplied decomposition is minimal by Proposition~\ref{prop:matricization}.
\item\label{item:second-approach} The second approach takes Moitra's formulation which decides whether $\rankp(\T_{(1)})\le R-1$, and requires in addition that the factorization $\T_{(1)}=UV^{\mathsf T}$ given by a solution satisfies the rank-one condition.  If the resulting system has a solution, then the supplied decomposition is not minimal by Proposition~\ref{prop:rank-one}\ref{item:rank-one-converse}.  The rank-one condition adds conditions but no variables, so this system has at most $2(R-1)^2$ real variables, independently of $n_1$, $n_2$, and $n_3$.
\end{enumerate}
The converses of these two statements fail, and even together the two approaches do not decide minimality, as the following example shows.

\begin{example}\label{ex:matricization}
Let $\be_1,\be_2$ be the standard basis of $\R^2$.  We call the tensor
\begin{equation}\label{eq:W}
\T=\be_1\otimes\be_1\otimes\be_2+\be_1\otimes\be_2\otimes\be_1+\be_2\otimes\be_1\otimes\be_1
\end{equation}
\emph{the W tensor}, as in \citet[Example~1]{WangYuan2026}, where the decomposition \eqref{eq:W} is shown to be minimal and essentially unique.  Nevertheless, the following two statements hold.
\begin{enumerate}[label=\textup{(\roman*)}]
\item\label{item:W-rank} Every matricization of $\T$ has nonnegative rank two.
\item\label{item:W-moitra} For $M=\T_{(1)}$ and every $q\le3$, the formulation of \citet[Section~2.4]{Moitra2016} with the rank-one condition added has no solution.
\end{enumerate}
By part~\ref{item:W-rank}, approach~\ref{item:first-approach} fails to verify that \eqref{eq:W} is minimal; and by part~\ref{item:W-moitra}, approach~\ref{item:second-approach} fails to verify that the decomposition with $R=4$ obtained by splitting the last term of \eqref{eq:W} into two halves is not minimal.
\end{example}

\begin{proof}
First, the support of $\T$ consists of $(1,1,2)$, $(1,2,1)$, and $(2,1,1)$. The support of a nonnegative rank-one term is a Cartesian product of factor supports. A box containing any two of these three indices also contains an index outside $\supp(\T)$, since each pair differs in two coordinates. Hence every nonzero term in a nonnegative decomposition of $\T$ is supported at just one of these indices. At least three terms are required; with exactly three, each term is the unit tensor at its index. This proves directly that \eqref{eq:W} is minimal and essentially unique.

For the decomposition \eqref{eq:W}, the factorization \eqref{eq:mat-factorization} is
\begin{equation}\label{eq:W-factorization}
\T_{(1)}=\begin{pmatrix}0&1&1&0\\1&0&0&0\end{pmatrix}=U_0V_0^{\mathsf T},
\qquad
U_0=\begin{pmatrix}1&1&0\\0&0&1\end{pmatrix},
\qquad
V_0=\begin{pmatrix}0&0&1\\1&0&0\\0&1&0\\0&0&0\end{pmatrix}.
\end{equation}
The columns of $V_0$ are $\be_1\otimes\be_2$, $\be_2\otimes\be_1$, and $\be_1\otimes\be_1$.

\pfpart{Part~\ref{item:W-rank}}
Since $\T_{(1)}$ has two rows and rank two, $\rankp(\T_{(1)})=2$.  The same holds for every matricization of $\T$, since $\T$ does not change when its modes are permuted.

\pfpart{Part~\ref{item:W-moitra}}
Suppose that, for some $q\le3$, the system with rank-one condition has a solution, and let $\T_{(1)}=UV^{\mathsf T}$ be the factorization which it gives.  By Proposition~\ref{prop:rank-one}\ref{item:rank-one-converse}, its nonzero terms are the matricizations of the terms of a nonnegative decomposition of $\T$ of length at most $q$.  Since \eqref{eq:W} is minimal and essentially unique, this decomposition has length $q=3$, and its terms are those of \eqref{eq:W}.  By Lemma~\ref{lem:product-equality}, $V$ is therefore obtained from $V_0$ by permuting its columns and multiplying them by positive numbers.  The second and third rows of $V_0$ are nonzero in different columns, so the second and third rows of $V$ differ.  On the other hand, the second and third columns of $\T_{(1)}$ are equal, so the second and third rows of $V$ are equal \citep[Claim~2.13]{Moitra2016}.  This contradiction proves part~\ref{item:W-moitra}.
\end{proof}

The first two terms of \eqref{eq:W-factorization} have the same column $\be_1$ of $U_0$, so they add up to a single term, and $\T_{(1)}=\be_1(0,1,1,0)+\be_2(1,0,0,0)$.  The term $\be_1(0,1,1,0)$ is not the matricization of a rank-one tensor, since $(0,1,1,0)$, written as a $2\times2$ matrix, has rank two.  Hence $\rankp(\T_{(1)})=2$, although $\T$ has no nonnegative decomposition of length two. The systems of Section~\ref{sec:normalized} do not merge terms in this way, since every irreducible localized replacement itself gives a solution by Theorem~\ref{thm:representation}\ref{item:completeness}.

\section{Necessity of the structural ingredients}\label{app:necessity}

This appendix shows that the ingredients of the reduction in Section~\ref{sec:normalized}, namely the added rows $K_m$ of a profile, the restriction to irreducible exchanges, the nonnegative lifting condition~\ref{cond:lifting}, and the term-disjointness inequalities \eqref{eq:td}, cannot be omitted.  It uses the cycle matrices $C_n$ and the matrices $B_R$ and $T_2$ of Section~\ref{sec:sharpness}.
\paragraph{Added rows and joint spans.}
The factors of a nonnegative decomposition cannot always be restricted to the column and row spaces of the matrix. Consequently the profile row sets $K_m$ in Definition~\ref{def:profile} cannot in general be omitted.

\begin{corollary}\label{cor:added-rows}
Let $R\ge4$ be even, and supply the row decomposition of $C_R$.
Every feasible profile of a term-disjoint balanced replacement uses $I=[R]$ and satisfies
\[
 \sigma_1=\delta_1=R,\qquad
 \sigma_2=R-1,\qquad \delta_2=R,
 \qquad K_1=\varnothing,\quad |K_2|=1.
\]
Therefore, if $K_m=\varnothing$ were required in every mode, Phase~III would find no competitor, although the row decomposition is not essentially unique.
\end{corollary}

\begin{proof}
The supplied factor matrices are $\Id_R$ and $C_R^{\mathsf T}$, so
Lemma~\ref{lem:cycle} gives their ranks $R$ and $R-1$. By Theorem~\ref{thm:attainment}\ref{item:attainment-balanced}, every feasible term-disjoint balanced profile has $I=[R]$ and $\kappa=(\delta_1-1)+(\delta_2-1)=2R-2$, which forces $\delta_1=\delta_2=R$ because $\delta_m\le R$.
The identities for $K_m$ follow from
$|K_m|=\delta_m-\sigma_m$ in Definition~\ref{def:profile}. If $K_m=\varnothing$ is required, Phase~III tries only profiles with $\delta_2=\sigma_2=R-1$, and none of them is feasible.  Since the column decomposition is an inequivalent competitor, the row decomposition is nevertheless not essentially unique.
\end{proof}

\begin{remark}\label{rem:compression}
For ordinary rank, let $r=\rank M$ and let $M=\sum_{j=1}^{r}u_jv_j^{\mathsf T}$ be a minimal decomposition.  The column space of $M$ is $\spanof\{u_1,\ldots,u_r\}$, and likewise for the rows and $v_j$. Let the columns of $U$ and $V$ be bases of these two spans.  Writing $u_j=Ua_j$ and $v_j=Vb_j$ gives
\[
 M=UGV^{\mathsf T},
 \qquad
 G=\sum_{j=1}^{r}a_jb_j^{\mathsf T}.
\]
Since $U$ and $V$ have linearly independent columns, $G$ is determined by $M$ alone, so every minimal decomposition of $M$ gives a decomposition of the same matrix $G$.  Conversely, every decomposition of $G$ gives one of $M$.  Thus one may restrict the factors to these spans and decompose the smaller matrix $G$.

For nonnegative decompositions this restriction can fail.  For $C_4$, both column and row spaces equal $\{x\in\R^4:x_1+x_3=x_2+x_4\}$, and the nonnegative vectors of this space form a cone $\mathcal C$.  Every nonzero $x\in\mathcal C$ has $x_1+x_3>0$ and satisfies
\[
 x=\sum_{a\in\{1,3\}}\sum_{b\in\{2,4\}}\frac{x_ax_b}{x_1+x_3}\,(\be_a+\be_b),
\]
so $\mathcal C$ is generated by the four vectors $\be_a+\be_b$ with $a\in\{1,3\}$ and $b\in\{2,4\}$, which are the rows of $C_4$.  A vector of $\mathcal C$ with support in $\{a,b\}$ has $x_a=x_b$, so it is a multiple of $\be_a+\be_b$.  Hence each of these four vectors spans an extreme ray of $\mathcal C$.

Now suppose $C_4=\sum_ju_jv_j^{\mathsf T}$ is a nonnegative decomposition whose factors $v_j$ lie in this space, so that $v_j\in\mathcal C$.  Row $i$ of $C_4$ is $\sum_ju_j[i]\,v_j^{\mathsf T}$ and has the form $\be_a+\be_b$.  Since every summand is nonnegative, each $v_j$ with $u_j[i]>0$ has support in $\{a,b\}$ and is therefore a multiple of $\be_a+\be_b$, that is, of row $i$. If some $u_j$ had two nonzero coordinates, say $i$ and $i'$, then $v_j$ would be parallel to both the $i$th and the $i'$th row of $C_4$, which is impossible since distinct rows of $C_4$ are not parallel.  Hence each $u_j$ has only one nonzero entry, so it does not lie in $\mathcal C$.  Thus every nonnegative decomposition of $C_4$ has a factor $u_j$ outside the column space or a factor $v_j$ outside the row space. This is why the profiles contain the rows $K_m$, see Corollary~\ref{cor:added-rows}.
\end{remark}

\paragraph{Irreducibility.}
\begin{example}\label{ex:reducible}
Without irreducibility, the effective order of an exchange is not
bounded by any function of $p$ and $q$ alone.
For $d\ge3$, put $u=\be_1+\be_2$ and $w=\be_2+\be_3$ in $\R^3$,
and define
\[
\begin{aligned}
 P_1&=\be_1\otimes u^{\otimes(d-1)},&
 P_2&=\be_3\otimes u^{\otimes(d-1)},\\
 P_3&=w\otimes\be_1\otimes w^{\otimes(d-2)},&
 P_4&=w\otimes\be_3\otimes w^{\otimes(d-2)},\\
 Q_1&=(\be_1+\be_3)\otimes u^{\otimes(d-1)},&
 Q_2&=w\otimes(\be_1+\be_3)\otimes w^{\otimes(d-2)}.
\end{aligned}
\]
Then the full exchange $P_1+P_2+P_3+P_4=Q_1+Q_2$ has $p=4$ and $q=2$.  Its factors span $\spanof\{\be_1,\be_3,w\}=\R^3$ in mode~$1$, $\spanof\{u,\be_1,\be_3\}=\R^3$ in mode~$2$, and $\spanof\{u,w\}$ in every later mode.  Hence $\delta_1=\delta_2=3$ and $\delta_m=2$ for $m\ge3$, so
\[
 \kappa=2+2+(d-2)=d+2,
\]
although $p+q-2=4$ is fixed.  The full exchange is not irreducible, since it splits into $P_1+P_2=Q_1$ and $P_3+P_4=Q_2$. 
\end{example}

\paragraph{Nonnegative lifting and term disjointness.}
The core equations alone do not ensure a nonnegative lift to the omitted rows. Likewise, excluding the supplied terms is necessary when searching for a witness to failure of essential uniqueness.

\begin{example}
\label{ex:lifting}
Consider the supplied decomposition $\T=\sum_{r=1}^{2}A^{(1)}[:,r]\otimes A^{(2)}[:,r]$ with the factor matrices
\begin{equation}\label{eq:lifting-input}
A^{(1)}=\begin{pmatrix}2&1\\1&1\\0&1\\1&0\end{pmatrix},\qquad A^{(2)}=\Id_2,
\end{equation}
so that $\T=A^{(1)}$.  This decomposition is minimal and essentially unique by the two-sided separability \citep[Corollary~3]{WangYuan2026}. Nevertheless, if either of the following two changes is made, the system $\mathcal N^{\mathrm{td}}_{[2],2}$ has a solution, so Phase~III would wrongly report that the decomposition is not essentially unique.
\begin{enumerate}[label=\textup{(\roman*)}]
\item\label{item:necessity-lifting}
condition~\ref{cond:lifting} is omitted but
\eqref{eq:td} and all other conditions are retained;
\item\label{item:necessity-td}
\eqref{eq:td} is omitted but all of
\ref{cond:simplex}--\ref{cond:lifting} are retained.
\end{enumerate}
\end{example}

\begin{proof}
Phase~III considers the block $([2],2)$, for which both supplied factor matrices have rank two, i.e. $\sigma_1=\sigma_2=2$.  Since $\delta_m\ge\sigma_m$ by \eqref{eq:profile-dimension-range}, the effective-order bound $\kappa=(\delta_1-1)+(\delta_2-1)\le2$ of Theorem~\ref{thm:budget} forces $\delta_1=\delta_2=2$, and hence $|K_m|=\delta_m-\sigma_m=0$. In both modes, the first row basis in the lexicographic order consists of rows $1$ and $2$.  The normalized supplied data are
\[
 \bar X_1=\begin{pmatrix}2/3&1/2\\1/3&1/2\end{pmatrix},
 \qquad \bar X_2=\Id_2,\qquad
 (\alpha_1,\alpha_2)=(3,2),\qquad \Omega_I=5.
\]
Hence the core \eqref{eq:core} is the matrix formed by the first two rows of $A^{(1)}$,
\[
 \bar X_1\operatorname{diag}(\alpha)\,\bar X_2^{\mathsf T}=\begin{pmatrix}2&1\\1&1\end{pmatrix}.
\]

\pfpart{Part~\ref{item:necessity-lifting}}
For $0<t\le1/2$,
\[
 \begin{pmatrix}2&1\\1&1\end{pmatrix}=\begin{pmatrix}2&1-2t\\1&1-t\end{pmatrix}\begin{pmatrix}1&0\\t&1\end{pmatrix}^{\mathsf T}=:Y_1Y_2^{\mathsf T}.
\]
Both factors are nonnegative and have positive column sums. Normalizing the columns of $Y_1$ and $Y_2$ gives
\begin{equation}\label{eq:false-core-competitors}
\begin{gathered}
 \bar Y_1=
 \begin{pmatrix}
 2/3&(1-2t)/(2-3t)\\
 1/3&(1-t)/(2-3t)
 \end{pmatrix},
 \qquad
 \bar Y_2=
 \begin{pmatrix}
 1/(1+t)&0\\ t/(1+t)&1
 \end{pmatrix},\\
 \omega_1=3(1+t),\qquad \omega_2=2-3t.
\end{gathered}
\end{equation}
They satisfy~\ref{cond:simplex}--\ref{cond:core}.
Moreover, $\omega_1>3$ and $0<\omega_2<2$, so neither competing
weight equals either supplied weight. Consequently every polynomial
on the left of~\eqref{eq:td} is strictly positive.

The omitted row $A^{(1)}[4,:]=(1,0)$ is the first selected row minus the
second. Since $\bar X_1$ has rank two, its coefficient row in
\eqref{eq:row-lift} is uniquely $\ell_{1,4}=(1,-1)$. For this $\bar Y_1$,
\[
 \ell_{1,4}\bar Y_1
 =\left(\frac13,\,-\frac{t}{2-3t}\right),
\]
which has a negative entry. Hence condition~\ref{cond:lifting}
excludes these points, whereas deleting it would make the
term-disjoint system feasible for an essentially unique input.

\pfpart{Part~\ref{item:necessity-td}}
The values $\bar Y_m=\bar X_m$ and $\omega=\alpha$ describe the supplied decomposition itself, so they satisfy~\ref{cond:simplex}--\ref{cond:core}.  For the omitted rows $i=3,4$, \eqref{eq:row-lift} gives $\ell_{1,i}\bar Y_1=\ell_{1,i}\bar X_1=\bar a_{1,i}\ge0$, so~\ref{cond:lifting} also holds. Only \eqref{eq:td} fails, since each competing term equals a supplied term.  Without \eqref{eq:td}, Phase~III would return the supplied decomposition itself as a competitor.
\end{proof}

These examples explain why the exact reduction must encode the joint spans, nonnegative lifting, and term disjointness, in addition to the core tensor identity.

\section{Deferred proofs}\label{app:proofs}
\subsection{Proof of Lemma~\ref{lem:block-decomposition}}

\begin{proof}
Choose a nonempty subexchange minimal under inclusion; it is irreducible. Remove it and repeat on the remaining exchange. Finiteness gives a partition into irreducible blocks. If both original decompositions are minimal, each block side is minimal by Lemma~\ref{lem:localized-minimal-subsums}, so both sides have length equal to the nonnegative rank of their common sum.
\end{proof}

\subsection{Proof of Lemma~\ref{lem:localized-minimal-subsums}}

\begin{proof}
The subsum has a nonnegative decomposition of length $|K|$.  If it had a decomposition of length $s<|K|$, replacing the terms $Z_i$, $i\in K$, by those $s$ terms would give a nonnegative decomposition of $\sum_{i=1}^{n}Z_i$ of length $n-|K|+s<n$, a contradiction.
\end{proof}

\subsection{Proof of Lemma~\ref{lem:nonminimal-nonunique}}

\begin{proof}
Let $\T=\sum_{j=1}^{\rho^*}Q_j$ be a minimum-length decomposition, where $\rho^*<R$.  No two terms of this decomposition are positive multiples of one another, because two such terms could be combined into a single nonzero nonnegative rank-one term, contradicting minimality.  Split $Q_1$ into $R-\rho^*+1\ge2$ positive multiples whose coefficients sum to one.  Two different multisets of positive coefficients give two length-$R$ decompositions of $\T$.  Since no other $Q_j$ is parallel to $Q_1$, these two decompositions have different multisets of rank-one terms and are therefore inequivalent.
\end{proof}

\subsection{Proof of Proposition~\ref{prop:blocks}}

\begin{proof}
Every block has both sides nonempty, so $q_\nu\ge1$.  Lemma~\ref{lem:localized-minimal-subsums}, applied to $Q_1,\ldots,Q_{\rho^*}$, gives $\rankp\bigl(\sum_{j\in J_\nu}Q_j\bigr)=q_\nu$.  By \eqref{eq:blocks-minimal}, this is the nonnegative rank of $\sum_{r\in I_\nu}P_r$, which has a decomposition with $p_\nu$ terms, so $q_\nu\le p_\nu$.  Part~\ref{item:blocks-count} follows from $\sum_\nu p_\nu=R$ and $\sum_\nu q_\nu=\rho^*$.  If $\rho^*<R$, part~\ref{item:blocks-count} forces $p_\nu-q_\nu>0$ for some $\nu$, and then $p_\nu>q_\nu\ge1$.
\end{proof}

\subsection{Proof of Lemma~\ref{lem:combine}}

\begin{proof}
The sets $I_\nu$ are pairwise disjoint.  Subtracting the supplied subsums indexed by these sets from $\T$ and inserting the equal right-hand sides in \eqref{eq:local-replacements} gives \eqref{eq:combined-decomposition}.  Every displayed term is a nonzero nonnegative rank-one tensor.  The number of terms is
\[
R-\sum_\nu p_\nu+\sum_\nu q_\nu
=R-\sum_\nu(p_\nu-q_\nu)
=\rho(\mathfrak D).
\]
\end{proof}

Figure~\ref{fig:pattern} illustrates Definition~\ref{def:pattern} and the decomposition \eqref{eq:combined-decomposition} for a pattern with two blocks.

\begin{figure}[t]
\centering
\begin{tikzpicture}[x=1cm,y=1cm]
\foreach \k/\style in {1/knownsel,2/knownsel,3/knownsel,4/knownsel,5/knownsel,6/known,7/known}{
\pgfmathsetmacro{\x}{1.35*(\k-1)}
\path[\style] (\x,2.6) rectangle ++(1.2,0.8);
\draw (\x,2.6) rectangle ++(1.2,0.8);
}
\foreach \k in {1,...,5}{ \pgfmathsetmacro{\x}{1.35*(\k-1)+0.6} \node[lab,text=white] at (\x,3.0) {$P_{\k}$}; }
\foreach \k in {6,7}{ \pgfmathsetmacro{\x}{1.35*(\k-1)+0.6} \node[lab] at (\x,3.0) {$P_{\k}$}; }
\draw[brace] (0,3.52) -- node[above=4pt,lab]{$I_1$,\ \ $p_1=3$} (3.9,3.52);
\draw[brace] (4.05,3.52) -- node[above=4pt,lab]{$I_2$,\ \ $p_2=2$} (6.6,3.52);
\foreach \x/\name in {0.675/{Q_{11}},2.025/{Q_{12}},4.725/{Q_{21}}}{
\path[unknown] (\x,0) rectangle ++(1.2,0.8);
\draw (\x,0) rectangle ++(1.2,0.8);
\node[lab,fill=white,inner sep=1.5pt] at (\x+0.6,0.4) {$\name$};
}
\foreach \k in {6,7}{
\pgfmathsetmacro{\x}{1.35*(\k-1)}
\path[known] (\x,0) rectangle ++(1.2,0.8);
\draw (\x,0) rectangle ++(1.2,0.8);
\node[lab] at (\x+0.6,0.4) {$P_{\k}$};
}
\draw[brace] (3.225,-0.12) -- node[below=4pt,lab]{$q_1=2$} (0.675,-0.12);
\draw[brace] (5.925,-0.12) -- node[below=4pt,lab]{$q_2=1$} (4.725,-0.12);
\draw[brace] (9.3,-0.12) -- node[below=4pt,lab]{unchanged} (6.75,-0.12);
\draw[arr] (1.95,2.5) -- (1.95,0.9);
\node[slab,anchor=west] at (2.05,1.7) {$\sum\limits_{r\in I_1}\!P_r=Q_{11}+Q_{12}$};
\draw[arr] (5.325,2.5) -- (5.325,0.9);
\node[slab,anchor=west] at (5.425,1.7) {$\sum\limits_{r\in I_2}\!P_r=Q_{21}$};
\draw[arr,densely dashed] (7.35,2.5) -- (7.35,0.9);
\draw[arr,densely dashed] (8.7,2.5) -- (8.7,0.9);
\node[lab,anchor=east,align=right] at (-0.25,3.0) {supplied decomposition,\\[-1pt] length $R=7$};
\node[lab,anchor=east,align=right] at (-0.25,0.4) {competitor,\\[-1pt] length $\rho(\mathfrak D)=5$};
\end{tikzpicture}
\caption{A disjoint shortening pattern $\mathfrak D=((I_1,2),(I_2,1))$ for $R=7$ and one realization.  The terms in $I_1$ and in $I_2$ are replaced at the same time by fewer terms, and the other terms are unchanged.  This gives a competitor of length $\rho(\mathfrak D)=R-(p_1-q_1)-(p_2-q_2)=5$.}
\label{fig:pattern}
\end{figure}
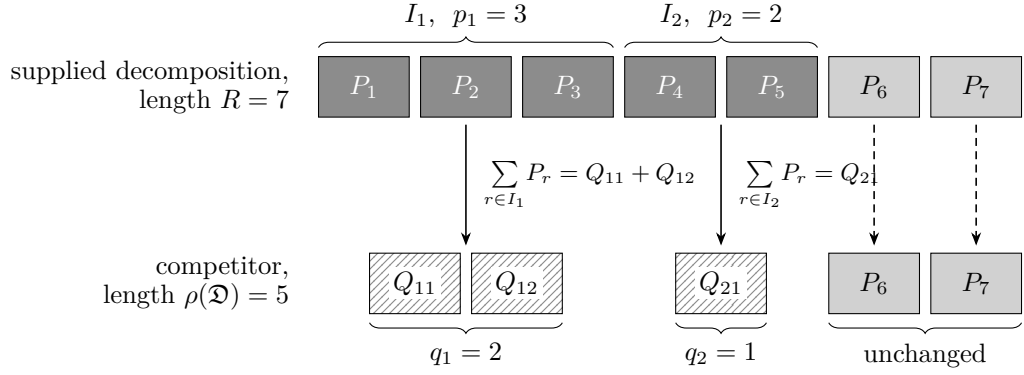

\subsection{Proof of Lemma~\ref{lem:circuit-elimination}}

\begin{proof}[Proof of Lemma~\ref{lem:circuit-elimination}]
\stepcounter{proofstepblock}
\proofstep[step:circuit-kernel]{the vector $\Delta^{E,J}$ spans $\ker\Cm_J$}
Fix $k\ge2$, a set $J=\{j_1<\cdots<j_k\}$, and a set $E\subseteq[e]$ with $|E|=k-1$, and put $\Delta=\Delta^{E,J}$.  Suppose $\rank \Cm_J=k-1$ and $\rank \Cm[E,J]=k-1$.  By Lemma~\ref{lem:cofactor}, the vector $\Delta$ is nonzero and spans $\ker\Cm[E,J]$.  Since $\rank\Cm_J=k-1$, the kernel of $\Cm_J$ is one-dimensional, and it is contained in $\ker\Cm[E,J]$ because the rows of $\Cm[E,J]$ are rows of $\Cm_J$.  Hence $\ker\Cm_J=\R\Delta$.
For $k=1$, $\Delta=(-1)$ by convention, and $\ker \Cm_J=\R\Delta$ holds exactly when the single column $\Cm[:,j_1]$ is zero, that is, when $\rank \Cm_J=0=k-1$.

\proofstep[step:circuit-forward]{condition \eqref{eq:farkas-condition} implies \eqref{eq:positive-circuit} and \eqref{eq:negative-circuit}}
Assume \eqref{eq:farkas-condition}, and fix $(k,J,E)$ such that the hypothesis of \eqref{eq:positive-circuit} holds.  All $\Delta_t$ are positive, so some minor $\det \Cm[E,J\setminus\{j_t\}]$ is nonzero, and $\rank \Cm_J\ge\rank \Cm[E,J]\ge k-1$. Since all minors of order $k$ of $\Cm_J$ are zero, $\rank \Cm_J\le k-1$. Hence $\rank \Cm_J=k-1$ and, by Step~\ref{step:circuit-kernel}, $\ker \Cm_J=\R\Delta$ with $\Delta$ strictly positive.  Let $\eta\in\R^q$ be the extension of $\Delta$ by zero outside $J$.  Then $\eta\ge0$ and $\Cm\eta=\Cm_J\Delta=0$, so \eqref{eq:farkas-condition} gives
\[
0\le g\eta=\sum_{t=1}^{k}g_{j_t}\Delta_t,
\]
which is the conclusion of \eqref{eq:positive-circuit}.  If instead the hypothesis of \eqref{eq:negative-circuit} holds, all $\Delta_t$ are negative.  The same argument applied to $-\Delta$ gives $g(-\eta)\ge0$, that is, $\sum_t g_{j_t}\Delta_t\le0$, which is the conclusion of \eqref{eq:negative-circuit}.

\proofstep[step:circuit-reverse]{the conditions \eqref{eq:positive-circuit} and \eqref{eq:negative-circuit} imply \eqref{eq:farkas-condition}}
Assume that \eqref{eq:positive-circuit} and \eqref{eq:negative-circuit} hold for all $(k,J,E)$ as in the lemma.  By Proposition~\ref{prop:lifting}\ref{item:lifting-rays}, it is enough to prove $g\eta\ge0$ for every $\eta$ which generates an extreme ray of $\mathcal K(\Cm)$.  Let $\eta$ be such a vector, with support $J$ of size $k$.  By Proposition~\ref{prop:lifting}\ref{item:lifting-support}, $k\le q$ and $k\le\rank \Cm+1\le e+1$, $\rank \Cm_J=k-1$, $\ker \Cm_J=\R\eta_J$, and $\eta_J$ is strictly positive.  In particular, all minors of order $k$ of $\Cm_J$ are zero. If $k=1$, take $E=\varnothing$; then $\Delta_1=-1<0$, the hypothesis of \eqref{eq:negative-circuit} holds, and the conclusion gives $-g_{j_1}\le0$, that is, $g\eta=g_{j_1}\eta_{j_1}\ge0$.  If $k\ge2$, choose $E\subseteq[e]$ with $|E|=k-1$ and $\rank \Cm[E,J]=k-1$; such $E$ exists because $\rank \Cm_J=k-1$.  By Step~\ref{step:circuit-kernel}, $\Delta$ spans $\ker \Cm_J=\R\eta_J$, so $\Delta=\lambda\,\eta_J$ for some $\lambda\ne0$, and all entries of $\Delta$ have the sign of $\lambda$.  If $\lambda>0$, the hypothesis of \eqref{eq:positive-circuit} holds, and its conclusion gives
\[
0\le\sum_{t=1}^{k}g_{j_t}\Delta_t=\lambda\sum_{t=1}^{k}g_{j_t}\eta_{j_t}=\lambda\,g\eta,
\]
so $g\eta\ge0$.  If $\lambda<0$, the hypothesis of \eqref{eq:negative-circuit} holds, and its conclusion gives $\lambda\,g\eta\le0$, again implying $g\eta\ge0$.  Thus $g$ is nonnegative on every extreme ray, and \eqref{eq:farkas-condition} holds.

\proofstep{number and degrees of the polynomials}
Regard the entries of $g$ and $\Gamma$ as variables.  The number of triples $(k,J,E)$ is at most $\sum_k\binom{q}{k}\binom{e}{k-1}\le 2^{q+e}$, which depends only on $(e,q)$, and every triple gives two implications.  Their polynomials are of three kinds.  The minors of order $k$ of $\Gamma_J$ have degree $k$, and they occur only for $k\le e$.  The numbers $\Delta_t$ are up to sign minors of order $k-1$, and they have degree $k-1\le e$.  The sums $\sum_tg_{j_t}\Delta_t$ have degree $1+(k-1)\le e+1$.  Hence every polynomial has degree at most $e+1$.
\end{proof}

\subsection{Proof of Proposition~\ref{prop:polya}}

\begin{proof}[Proof of Proposition~\ref{prop:polya}]
\pfpart{Lower bound} For $k\ge1$ put $\psi_k(t)=(1,t,\ldots,t^k)^{\mathsf T}$.  The entries of $v(t)^{\otimes k}$ span the polynomials of degree at most $k$, so there is a linear map $\Psi_k:(\R^2)^{\otimes k}\to\R^{k+1}$ with $\Psi_kv(t)^{\otimes k}=\psi_k(t)$ for all $t$. Let $F=\int_0^1v(t)^{\otimes(\rho-1)}\bigl(v(t)^{\otimes(d-\rho+1)}\bigr)^{\mathsf T}dt$ be the flattening of $\mathcal P_d$ whose rows are indexed by the first $\rho-1$ modes and whose columns by the others, with the tensors written as column vectors.  Then $\rank F\le\rank(\mathcal P_d)$, since flattening does not increase the rank \citep[\S3.4]{Landsberg2012}.  Moreover,
\[
 \Psi_{\rho-1}F\,\Psi_{d-\rho+1}^{\mathsf T}=\int_0^1\psi_{\rho-1}(t)\,\psi_{d-\rho+1}(t)^{\mathsf T}\,dt=\Bigl[\int_0^1t^{i+l}\,dt\Bigr]_{0\le i\le\rho-1,\ 0\le l\le d-\rho+1}.
\]
Since $2\rho-2\le d$, this matrix contains the block with $0\le i,l\le\rho-1$, and the block is positive definite because for every nonzero $c\in\R^{\rho}$,
\[
 \sum_{i,l=0}^{\rho-1}c_ic_l\int_0^1t^{i+l}\,dt=\int_0^1\Bigl(\sum_{i=0}^{\rho-1}c_it^i\Bigr)^2dt>0.
\]
Hence $\rho\le\rank\bigl(\Psi_{\rho-1}F\,\Psi_{d-\rho+1}^{\mathsf T}\bigr)\le\rank F\le\rank(\mathcal P_d)\le \rankp(\mathcal P_d)$.

\pfpart{Upper bound} Let $\xi_1,\ldots,\xi_\rho$ be the nodes of the $\rho$-point Gauss--Legendre rule.  Since its degree $2\rho-1$ is at least $d$, \eqref{eq:quadrature} gives a nonnegative decomposition of $\mathcal P_d$ of length $\rho$, so $\rankp(\mathcal P_d)\le\rho$.

\pfpart{Part~\ref{item:polya-odd}} Let $I\subseteq[\rho]$ with $|I|\ge2$.  The nodes $\xi_j$ are distinct, so the factors $v(\xi_j)$, $j\in I$, span $\R^2$ in every mode.  Hence $\beta(I)=d=2\rho-1\ge2|I|-1$, since $|I|\le\rho$.  By the criterion of \citet[Theorem~2]{LovitzPetrov2023}, the decomposition given by the Gauss--Legendre rule is the only decomposition of $\mathcal P_d$ of length $\rho$, even among decompositions with real factors of arbitrary sign.

\pfpart{Part~\ref{item:polya-even}}  For $w,\xi\in\R^\rho$ put $\Phi(w,\xi)=\bigl(\sum_jw_j\xi_j^i\bigr)_{i=0}^{d}$.  By linearity, nodes $\xi_j\in[0,1]$ and weights $w_j>0$ form a positive quadrature rule of degree $d$ if and only if $\sum_jw_jh(\xi_j)=\int_0^1h(t)\,dt$ for $h=1,t,\ldots,t^d$, that is,
\[
 \Phi(w,\xi)=\Bigl(\int_0^1t^i\,dt\Bigr)_{i=0}^{d}.
\]
These are $d+1=2\rho-1$ equations in $2\rho$ unknowns, and the Gauss--Legendre rule $(w^*,\xi^*)$ solves them, since its degree $2\rho-1$ exceeds $d$.  The Jacobian of $\Phi$ at $(w^*,\xi^*)$ has rank $d+1$.  Indeed, let $y\in\R^{d+1}$ be orthogonal to all its columns, and put $p(t)=\sum_{i=0}^{d}y_it^i$.  Then
\[
 p(\xi_j^*)=y^{\mathsf T}\frac{\partial\Phi}{\partial w_j}=0,\qquad w_j^*\,p'(\xi_j^*)=y^{\mathsf T}\frac{\partial\Phi}{\partial\xi_j}=0,\qquad j\in[\rho].
\]
Since $w_j^*>0$, each node $\xi_j^*$ is a double root of $p$, so $p$ has at least $2\rho>d$ roots counted with multiplicity.  Hence $p=0$, that is, $y=0$. By the implicit function theorem \citep[Theorem~9.28]{Rudin1976}, near $(w^*,\xi^*)$ one of the $2\rho$ unknowns determines the others as continuously differentiable functions, so the solutions form a curve through $(w^*,\xi^*)$. Choose the nodes of the initial rule in increasing order and restrict each node to a small neighborhood disjoint from the neighborhoods of the other nodes. Near $(w^*,\xi^*)$ the weights stay positive and the nodes stay distinct and ordered in $(0,1)$, so the points of this curve are positive quadrature rules of degree $d$ with $\rho$ nodes. For fixed distinct nodes, the equations with $i<\rho$ form a square linear system for $w$.  Its matrix $[\xi_j^{\,i}]_{0\le i<\rho,\ j\in[\rho]}$ is a Vandermonde matrix with distinct nodes and is therefore invertible.  Hence the nodes determine the weights, so different points of this restricted curve have different unordered node sets. By~\eqref{eq:quadrature}, each point of the curve gives a nonnegative decomposition $\mathcal P_d=\sum_jw_jv(\xi_j)^{\otimes d}$ of length $\rho$.  The term $w_jv(\xi_j)^{\otimes d}$ has the entry $w_j(1-\xi_j)^d$ at the index $(1,\ldots,1)$ and the entry $w_j\xi_j(1-\xi_j)^{d-1}$ at the index $(2,1,\ldots,1)$.  Their ratio $\xi_j/(1-\xi_j)$ determines $\xi_j$.  Hence the infinitely many points of the curve give pairwise inequivalent decompositions.
\end{proof}

\subsection{Proof of Theorem~\ref{thm:degree}}

The proof uses the following classical fact.  Let $a_\rho x^\rho+\cdots+a_1x+a_0$ be a polynomial with $a_0,\ldots,a_\rho\in\Q$ and $a_\rho\ne0$, and let $x_1,\ldots,x_\rho$ be its complex zeros, listed with multiplicity.  Then
\[
 a_\rho x^\rho+\cdots+a_1x+a_0=a_\rho\prod_{j=1}^{\rho}(x-x_j)=a_\rho\sum_{k=0}^{\rho}(-1)^ke_k\,x^{\rho-k},
\]
where $e_0=1$ and $e_k$ is the sum of all products of $k$ distinct $x_j$.  Comparing coefficients shows that every $e_k=(-1)^ka_{\rho-k}/a_\rho$ is rational.  The power sum $x_1^i+\cdots+x_\rho^i$ does not change when the $x_j$ are permuted, so it is a polynomial with integer coefficients in $e_1,\ldots,e_\rho$ \citep[Theorem~2.2]{MilneANT}.  Hence it is rational.

\begin{proof}[Proof of the construction, the rank, and the uniqueness in Theorem~\ref{thm:degree}, and of part~(ii)]
\pfpart{The tensor}  Let
\[
 L_\rho(x)=\sum_{k=0}^{\rho}\binom{\rho}{k}\frac{(-x)^k}{k!}
\]
be the Laguerre polynomial of degree $\rho$ \citep[(18.5.12)]{DLMF}.  Its zeros are distinct and positive \citep[\S18.2(vi) and Table~18.3.1]{DLMF}.  Fix an integer $c$ which is larger than all these zeros, and let $\xi_1<\cdots<\xi_\rho$ be the zeros divided by $c$.  Then $\xi_j\in(0,1)$. Let $\mu$ be the measure on $[0,1]$ with mass one at every $\xi_j$.  Since the $\xi_j$ are the zeros of the polynomial $L_\rho(cx)$, whose coefficients are rational, the moments $\mu_i=\sum_j\xi_j^{\,i}$ of $\mu$ are rational by the fact recalled above. Put $\T=M_{2\rho-1}(\mu)=\sum_jv(\xi_j)^{\otimes(2\rho-1)}$.  For $d=2\rho-1$, the entry of $\T$ at an index with exactly $s$ coordinates equal to $2$ is
\[
 \int(1-t)^{d-s}t^s\,d\mu(t)=\sum_{k=0}^{d-s}(-1)^k\binom{d-s}{k}\mu_{s+k},
\]
an integer combination of moments of $\mu$, and hence rational.

\pfpart{A rational supplied decomposition}  Let $\mathcal U$ be the convex cone generated by the vectors $\psi_{2\rho-1}(t)=(1,t,\ldots,t^{2\rho-1})^{\mathsf T}$ with $t\in[0,1]$. 
The cone $\mathcal U$ is closed and full-dimensional. Indeed, the convex hull $K$ of the compact curve $\{\psi_{2\rho-1}(t):t\in[0,1]\}$ is compact, and every point of $K$ has first coordinate one. Thus $\mathcal U=\{a x:a\ge0,\ x\in K\}$ is closed: for a convergent sequence $a_nx_n$, its first coordinates $a_n$ converge, and compactness of $K$ supplies a convergent subsequence of $x_n$. Moreover, the vectors at any $2\rho$ distinct points of $[0,1]$ are linearly independent by the Vandermonde determinant, so $\mathcal U$ has nonempty interior in $\R^{2\rho}$.

We claim that the moment vector $\eta:=(\mu_0,\ldots,\mu_{2\rho-1})^{\mathsf T}=\sum_j\psi_{2\rho-1}(\xi_j)$ is an interior point of $\mathcal U$.  Otherwise $\eta$ lies on the boundary of $\mathcal U$, so there is a nonzero $\lambda\in\R^{2\rho}$ with $\lambda^{\mathsf T}u\ge\lambda^{\mathsf T}\eta$ for all $u\in\mathcal U$ \citep[\S2.5.2]{BoydVandenberghe2004}.  Since $0$ and $2\eta$ lie in the cone $\mathcal U$, we have $\lambda^{\mathsf T}\eta=0$. Hence $h(t):=\lambda^{\mathsf T}\psi_{2\rho-1}(t)$ is a nonzero polynomial of degree at most $2\rho-1$ with $\sum_jh(\xi_j)=\lambda^{\mathsf T}\eta=0$, and $h\ge0$ on $[0,1]$ because every $\psi_{2\rho-1}(t)$ with $t\in[0,1]$ lies in $\mathcal U$. Since every $h(\xi_j)\ge0$, this gives $h(\xi_j)=0$ for all $j$.  A zero of odd multiplicity at an interior point of $[0,1]$ would make $h$ change sign there, so every $\xi_j\in(0,1)$ is a zero of even multiplicity, and $h(t)$ is divisible by
\[
 (t-\xi_1)^2(t-\xi_2)^2\cdots(t-\xi_\rho)^2 .
\]
This polynomial has degree $2\rho>2\rho-1\ge\deg h$, which is impossible since $h\ne0$.

Let $\mathcal U_\Q\subseteq\mathcal U$ be the convex cone generated by the vectors $\psi_{2\rho-1}(t)$ with rational $t\in[0,1]$.  Every element of $\mathcal U$ is a nonnegative combination of finitely many vectors $\psi_{2\rho-1}(t)$ with $t\in[0,1]$ \citep[\S2.1.5]{BoydVandenberghe2004}.  Replacing each such $t$ by a nearby rational number changes this combination arbitrarily little, because $\psi_{2\rho-1}$ is continuous, so $\mathcal U$ lies in the closure of $\mathcal U_\Q$. Hence $\mathcal U$ and $\mathcal U_\Q$ have the same relative interior \citep[Proposition~1.3.5(c)]{Bertsekas2009}, and the interior point $\eta$ of $\mathcal U$ lies in it, so $\eta\in\mathcal U_\Q$. 

Carath\'eodory's theorem \citep[Proposition~1.2.1(a)]{Bertsekas2009} gives
\[
\eta=\sum_{i=1}^{R'}w'_i\psi_{2\rho-1}(\xi'_i),\qquad R'\le2\rho,
\]
where $\xi'_i\in[0,1]$ are rational, $w'_i>0$, and the vectors $\psi_{2\rho-1}(\xi'_1),\ldots,\psi_{2\rho-1}(\xi'_{R'})$ are linearly independent. By linear independence, the weights $w'_i$ are the unique solution of the linear system $\sum_iw'_i\psi_{2\rho-1}(\xi'_i)=\eta$, whose coefficients and right-hand side are rational.  Gaussian elimination computes this solution with rational arithmetic, so the $w'_i$ are rational. Since $\psi_{2\rho-1}(t)$ lists the monomials $1,t,\ldots,t^{2\rho-1}$, the identity 
\[
\sum_{i=1}^{R'}w'_i\psi_{2\rho-1}(\xi'_i)=\eta=\sum_{j=1}^\rho\psi_{2\rho-1}(\xi_j)
\]
says that the nodes $\xi'_i$ and weights $w'_i$ form a positive quadrature rule of degree $2\rho-1$ for $\mu$, see the proof of Proposition~\ref{prop:polya}\ref{item:polya-even}.  By \eqref{eq:quadrature}, $\T=\sum_iw'_i\,v(\xi'_i)^{\otimes(2\rho-1)}$, and this is a rational nonnegative decomposition of length $R'$, since $\xi'_i\in[0,1]$ and $w'_i>0$ are rational. Repeatedly replacing a term by two terms, each with half its weight, we obtain such a decomposition of every length $R\ge R'$, in particular of every length $R\ge2\rho$, since $R'\le2\rho$.

\pfpart{The minimal decomposition} Flatten $\T$ between its first $\rho-1$ modes and its remaining $\rho$ modes, and apply the maps $\Psi_{\rho-1}$ and $\Psi_\rho$ from the proof of Proposition~\ref{prop:polya}. The resulting matrix is $[\mu_{i+l}]_{0\le i\le\rho-1,\ 0\le l\le\rho}$. Its leading $\rho\times\rho$ block is positive definite.  Indeed, for every nonzero $y\in\R^\rho$,
\[
 \sum_{i,l=0}^{\rho-1}y_iy_l\int t^{i+l}\,d\mu=\int\Bigl(\sum_{i=0}^{\rho-1}y_it^i\Bigr)^2d\mu=\sum_{j=1}^{\rho}\Bigl(\sum_{i=0}^{\rho-1}y_i\xi_j^{\,i}\Bigr)^2>0,
\]
since a nonzero polynomial of degree at most $\rho-1$ cannot vanish at all $\rho$ points $\xi_j$. Thus the flattening has rank at least $\rho$. Together with the displayed length-$\rho$ decomposition, this proves $\rank(\T)=\rankp(\T)=\rho$. For every subset $I\subseteq[\rho]$ with $|I|\ge2$, the distinct nodes imply that $v(\xi_j)$, $j\in I$, span $\R^2$ in every mode. Consequently $\beta(I)=2\rho-1\ge2|I|-1$. The criterion of \citet[Theorem~2]{LovitzPetrov2023} therefore makes $\T=\sum_jv(\xi_j)^{\otimes(2\rho-1)}$ the unique length-$\rho$ decomposition, even when competing real factors need not be nonnegative or symmetric.
\end{proof}

\section{Complexity analysis}\label{app:complexity}

This appendix contains the bookkeeping behind Proposition~\ref{prop:one-shot-format} and the running time and output length in Theorem~\ref{thm:main}.

\subsection{Proof of Proposition~\ref{prop:one-shot-format}}

\begin{proof}
\stepcounter{proofstepblock}
\proofstep[step:format-profiles]{the number of profiles}
A profile consists of the set of active modes $H$, the numbers $\delta_m$, and the sets $K_m$, $m\in H$.  We count the choices of each.  The set $H$ is a subset of $[d]$ with $|H|\le\kappa\le p+q-2$, because every $m\in H$ has $\delta_m-1\ge1$.  By Section~\ref{sec:model}, $d\le L$, so there are at most $(d+1)^{p+q-2}\le(L+1)^{p+q-2}$ choices of $H$.  For a fixed $H$, every $\delta_m$, $m\in H$, is an integer with $2\le\delta_m\le p+q-1$, so there are at most $(p+q-2)^{p+q-2}$ choices of the numbers $\delta_m$.  Since $\sigma_m\ge1$ for every mode,
\[
\sum_{m\in H}|K_m|
=\sum_{m\in H}(\delta_m-\sigma_m)
\le\kappa\le p+q-2,
\]
and by Section~\ref{sec:model}, $\sum_mn_m\le L$, so there are at most $(L+1)^{p+q-2}$ choices of the sets $K_m$.  Hence the number of profiles is at most
\[
(L+1)^{p+q-2}(p+q-2)^{p+q-2}(L+1)^{p+q-2}
\le G(p,q)L^{O(p+q)}.
\]
The same bound holds for the number of candidates $(H,\delta,K)$ which are examined.  Listing a candidate and checking \eqref{eq:profile-dimension-range}--\eqref{eq:profile-rows}, which needs the ranks $\sigma_m$, takes time polynomial in $L$, so the profiles are listed in time $G(p,q)L^{O(p+q)}$.

\proofstep[step:format-data]{the rational data}
The coefficients of the systems are computed from $\bar X_m$, $N_m$, and $w_{m,i}$, $m\in H$, and from the weights $\alpha_r$.  By Lemma~\ref{lem:core} and Section~\ref{subsec:explicit}, with $S_m=S_m^0\cup K_m$ and $D_{X,m}=\operatorname{diag}(s_{mr}:r\in I)$ as in Section~\ref{subsec:derivation},
\[
\bar X_m=A_I^{(m)}[S_m,:]\,D_{X,m}^{-1},
\qquad
N_m\bar X_m=0,
\qquad
w_{m,i}\bar X_m=A_I^{(m)}[i,:]\,D_{X,m}^{-1},\quad i \notin S_m,
\]
where the rows of $N_m$ form a basis of the left kernel of $\bar X_m$.  The row basis $S_m^0$, the matrix $N_m$, and the rows $w_{m,i}$ are computed by Gaussian elimination over $\Q$.  This takes time polynomial in $L$, and all intermediate numbers have bit length at most $(L+1)^{C}$ \citep[Theorem~3.3 and Corollary~3.3a]{Schrijver1986}.  By \eqref{eq:supplied-weight}, $\alpha_r$ is the product of the numbers $c_{mr}$, $m\notin H$, and $s_{mr}$, $m\in H$.  By \eqref{eq:inactive-ratios} and Lemma~\ref{lem:core},
\[
c_{mr}=\frac{a^{(m)}_r[i]}{a^{(m)}_{r_0}[i]}\quad(m\notin H),
\qquad
s_{mr}=\sum_{i\in S_m}a^{(m)}_r[i]\qquad m\in H,
\]
where $i$ is any row index with $a^{(m)}_{r_0}[i]\ne0$.  The bit length of a sum, a product, or a quotient of rational numbers is at most $C$ times the sum of their bit lengths~\citep[pp.~30--31]{GLS1988}.  By the formulas above, $c_{mr}$ and $s_{mr}$ are computed from entries of $A^{(m)}$ alone, so
\[
\bits(\alpha_r)
\le C\Bigl(\sum_{m\notin H}\bits(c_{mr})+\sum_{m\in H}\bits(s_{mr})\Bigr)
\le C\sum_{m=1}^{d}\bits(A^{(m)})
\le C\bits(A)=CL.
\]
Since $CL<(L+1)^{C}$ for $C\ge1$, all these data have bit length at most $(L+1)^{C}$.

\proofstep[step:format-length]{the length}
The system $\mathcal N_{I,q}$ consists of \ref{cond:simplex}--\ref{cond:core} and of the conditions of Lemma~\ref{lem:circuit-elimination} which replace \ref{cond:lifting}, and $\mathcal N^{\mathrm{td}}_{I,p}$ contains in addition the inequalities \eqref{eq:td}.  The numbers of equations and inequalities in \ref{cond:simplex}, \ref{cond:weights}, \ref{cond:core}, and \eqref{eq:td} are $q\sum_{m\in H}\delta_m$, $q$, $\prod_{m\in H}\delta_m$, and $p^2$.  Since $|H|\le\kappa\le p+q-2$ by Step~\ref{step:format-profiles} and $\delta_m\le2^{\delta_m-1}$,
\[
\sum_{m\in H}\delta_m=|H|+\kappa\le2(p+q-2),
\qquad
\prod_{m\in H}\delta_m\le2^{\kappa}\le2^{p+q-2}.
\]
For one omitted row $i\notin S_m$, the number of polynomials in the conditions of Lemma~\ref{lem:circuit-elimination} depends only on $(e_m,q)$.  Since $e_m=\delta_m-\sigma_m\le\delta_m-1\le p+q-2$, this number is at most $G(p,q)$.  By Section~\ref{sec:model}, there are at most $\sum_mn_m\le L$ omitted rows.  Hence the length is at most
\[
2q(p+q-2)+q+2^{p+q-2}+p^2+G(p,q)L\le G(p,q)(L+1).
\]

\proofstep[step:format-degree]{the degrees}
By \eqref{eq:simplex} and \eqref{eq:last-weight}, every entry of $\bar Y_m$ and every weight $\omega_j$ is a polynomial of degree at most one in the variables.  Hence \ref{cond:simplex} and \ref{cond:weights} consist of inequalities of degree at most one.  By \eqref{eq:core-equation}, in every equation of \ref{cond:core} the left-hand side is a number, and every term on the right-hand side is the product of one weight and of one entry of $\bar Y_m$ for every $m\in H$.  Hence these equations have degree at most $|H|+1\le p+q-1$, see Step~\ref{step:format-profiles}.  The entries of $g=w_{m,i}\bar Y_m$ and $\Gamma=N_m\bar Y_m$ in \eqref{eq:g-C-system} also have degree at most one, so by Lemma~\ref{lem:circuit-elimination} the polynomials in its conditions have degree at most $e_m+1\le p+q-1$, see Step~\ref{step:format-length}.  The polynomials in \eqref{eq:td} have degree two.  Since $p+q-1\ge2$, every polynomial of the system has degree at most $p+q-1$.

\proofstep[step:format-bits]{the bit lengths}
By Corollary~\ref{cor:count}, the system has at most $q(p+q-2)+q-1$ variables, and by Step~\ref{step:format-degree} its polynomials have degree at most $p+q-1$.  Hence every polynomial has at most $G(p,q)$ coefficients.  We multiply out \eqref{eq:core-equation}, \eqref{eq:simplex}--\eqref{eq:td}, \eqref{eq:g-C-system}, and the determinants in Lemma~\ref{lem:circuit-elimination}, which have order at most $p+q-1$.  Then every coefficient is a sum of at most $G(p,q)$ terms, and every term is, up to sign, a product of at most $p+q-1$ of the data of Step~\ref{step:format-data}.  By Step~\ref{step:format-data}, such a product has bit length at most $(p+q-1)(L+1)^{C}$, so every coefficient has bit length at most
\[
C\,G(p,q)\,(p+q-1)(L+1)^{C}\le G(p,q)(L+1)^{C}.
\]
Multiplying a polynomial by the product of the denominators of its coefficients gives integer coefficients, as required in Definition~\ref{def:polysystem}, of bit length at most $G(p,q)(L+1)^{C}$.  This proves part~\ref{item:format-single}.

\proofstep{combined systems}
By Lemma~\ref{lem:patterns}, there are at most $G(R)$ disjoint shortening patterns on $[R]$, and they are listed in time $G(R)$.  Fix a nonempty pattern.  The sets $I_\nu$ are pairwise disjoint subsets of $[R]$, and every block has $q_\nu<p_\nu$, so
\begin{equation}\label{eq:total-budget}
\sum_\nu(p_\nu+q_\nu-2)\le2\sum_\nu p_\nu-2c\le2R-2.
\end{equation}
By Step~\ref{step:format-profiles}, the number of choices of profiles for all blocks is at most
\[
\prod_\nu(p_\nu+q_\nu-2)^{p_\nu+q_\nu-2}(L+1)^{2(p_\nu+q_\nu-2)}
\le(2R-2)^{2R-2}(L+1)^{4R-4}
\le G(R)L^{O(R)},
\]
and they are listed in time $G(R)L^{O(R)}$.  By Definition~\ref{def:combined}, the combined system consists of the systems $\mathcal N_{I_\nu,q_\nu}$ of the blocks.  Since $p_\nu,q_\nu\le R$ and $c\le R$, part~\ref{item:format-single} shows that its length is at most $cG(R)(L+1)\le G(R)(L+1)$, and that its polynomials have degree at most $p_\nu+q_\nu-1\le2R-2$ and bit length at most $G(R)(L+1)^{C}$.  By Proposition~\ref{prop:dimension}, it has fewer than $2R^2$ variables.  This proves part~\ref{item:format-combined}.
\end{proof}

\subsection{Arithmetic in a real algebraic number field}

The next lemma collects the facts on computations in a field $\Q(\theta)$ which are needed for $\textsc{Reconstruct}$.  For two elements $g_1(\theta)/g_0(\theta)$ and $g_1'(\theta)/g_0'(\theta)$ of $\Q(\theta)$, the sum or difference, the product, and the quotient are
\[
\frac{(g_1g_0'\pm g_1'g_0)(\theta)}{(g_0g_0')(\theta)},\qquad
\frac{(g_1g_1')(\theta)}{(g_0g_0')(\theta)},\qquad
\frac{(g_1g_0')(\theta)}{(g_0g_1')(\theta)},
\]
so these operations need only additions, subtractions, and multiplications of integer polynomials.  We say that a quotient $g_1(\theta)/g_0(\theta)$ is \emph{bounded} by $M$ if $g_0$ and $g_1$ have degree at most $M$ and all their coefficients have bit length at most $M$.  If two elements are bounded by $M$, then their sum, difference, product, and quotient are bounded by $4M$, because the degrees at most double and every new coefficient is a sum of at most $2(M+1)$ products of two old coefficients.

\begin{lemma}\label{lem:field-lifting}
Let $\theta$ be a real algebraic number which is given by $f\in\mathbb Z[T]$ and its Thom encoding, let $h$ be the sum of the bit lengths of all coefficients of $f$ and of all entries of its Thom encoding, and let $t\ge1$.  Let the given elements of $\Q(\theta)$ have bit length at most $h'$.
\begin{enumerate}[label=\textup{(\alph*)}]
\item\label{item:field-arithmetic} Starting from given elements, a sequence of at most $t$ additions, subtractions, multiplications, divisions by nonzero elements, and comparisons in $\Q(\theta)$ costs at most
\[
G(t)(\deg f+1)^C(h+h'+1)^C
\]
bit operations, and every result has bit length at most this number.
\item\label{item:field-linear-systems} Consider at most $t$ linear equations and inequalities of the form $\ge$ in at most $t$ unknowns, whose coefficients are given elements of $\Q(\theta)$.  If they have a real solution, then they have a solution with coordinates in $\Q(\theta)$, and such a solution can be computed within bounds of the same form as in \ref{item:field-arithmetic}.
\end{enumerate}
\end{lemma}

\begin{proof}
For part \ref{item:field-arithmetic}, note that the given elements are bounded by $h'+1$, because each of them has at most $h'$ coefficients, each of bit length at most $h'$.  Hence every element which occurs in a sequence of at most $t$ operations is bounded by $M=4^t(h'+1)$.  Such an element has at most $2(M+1)$ coefficients of bit length at most $M$, so its bit length is at most
\[
2(M+1)\cdot M\le4M^2=4^{2t+1}(h'+1)^2.
\]
To compare two elements, we compute the sign of their difference.  By the first formula before the lemma,
\[
\sgn\Bigl(\frac{g_1(\theta)}{g_0(\theta)}-\frac{g_1'(\theta)}{g_0'(\theta)}\Bigr)
=\sgn\bigl((g_1g_0'-g_1'g_0)(\theta)\bigr)\cdot\sgn\bigl((g_0g_0')(\theta)\bigr).
\]
By \citet[Algorithms~8.1 and~8.2]{BasuPollackRoy2006}, an addition or multiplication of polynomials of degree at most $M$ with coefficients of bit length at most $M$ costs a number of bit operations which is polynomial in $M$.  By \citet[Algorithms~10.14 and~10.15]{BasuPollackRoy2006}, the sign of $g(\theta)$ for such a polynomial $g$ costs at most $(\deg f+h+M)^C$ bit operations, so the whole sequence costs at most
\[
t\bigl(\deg f+h+4^t(h'+1)\bigr)^{C}\le G(t)(\deg f+1)^C(h+h'+1)^C
\]
bit operations.

For part \ref{item:field-linear-systems}, we replace every equation by two inequalities of the form $\ge$ and apply Fourier--Motzkin elimination with back-substitution \citep[Section~12.2]{Schrijver1986}.  This method uses only additions, subtractions, multiplications, divisions, and comparisons of the coefficients, so it computes a solution with coordinates in $\Q(\theta)$ whenever the system has a real solution. Since the method starts with at most $2t$ inequalities and eliminates at most $t$ unknowns, the number of operations and comparisons in $\Q(\theta)$ which it uses depends only on $t$, and part \ref{item:field-linear-systems} follows from part \ref{item:field-arithmetic}.
\end{proof}

\begin{remark}\label{rem:field-bound}
If a system as in Lemma~\ref{lem:field-lifting}\ref{item:field-linear-systems} has a real solution, then there is a set of inequalities of the system such that the equations of the system and the inequalities of this set with $\ge$ replaced by~$=$ have a real solution, and every such solution is a solution of the system \citep[Theorem~2.4]{Schrijver1995}. For each of the at most $2^t$ sets of inequalities, we solve the equations of the system together with the inequalities of the set with $\ge$ replaced by~$=$, by Gaussian elimination over $\Q(\theta)$, and test whether the solution found, if any, satisfies all inequalities of the system.  This takes at most $Ct^3$ operations and comparisons in $\Q(\theta)$ for each set.  By the proof of Lemma~\ref{lem:field-lifting}\ref{item:field-arithmetic}, a sequence of at most $s$ operations and comparisons costs at most $2^{O(s)}(\deg f+1)^C(h+h'+1)^C$ bit operations, and every result has bit length at most this number.  Hence the computation for one set costs at most $2^{O(t^3)}(\deg f+1)^C(h+h'+1)^C$ bit operations, the at most $2^t$ sets together cost at most $2^t\cdot2^{O(t^3)}(\deg f+1)^C(h+h'+1)^C$ bit operations, and Lemma~\ref{lem:field-lifting}\ref{item:field-linear-systems} holds with $G(t)=2^{O(t^3)}$.
\end{remark}
\subsection{Running time and output length in Theorem~\ref{thm:main}}

\begin{proof}[Proof of the running time and output bounds]
We bound the cost of every line of Algorithm~\ref{alg:main} in three parts and then add them.

\costpart{item:cost-forming}~\emph{Listing and forming.}  By Proposition~\ref{prop:one-shot-format} and Steps~\ref{step:format-data} and~\ref{step:format-bits} of its proof, listing all choices of a pattern and of profiles costs $G(R)L^{O(R)}$ bit operations, and forming one system costs $G(R)(L+1)^{C}$ bit operations.

\costpart{item:cost-solve}~\emph{Solve.}  The algorithm calls $\textsc{Solve}$ for the combined systems of Phase~I and for the term-disjoint systems $\mathcal N^{\mathrm{td}}_{I,p}$ of Phase~III, where $q=p=|I|\le R$.  For a system with $k$ variables and length $s$ whose polynomials have degree at most $D$ and bit length at most $b$, one call costs at most $(s+D+b+2)^{C_{\mathrm{sa}}(k+1)}$ bit operations by Theorem~\ref{thm:sampling}.  For a combined system, Proposition~\ref{prop:one-shot-format}\ref{item:format-combined} gives
\[
k<2R^2,
\qquad
s\le G(R)(L+1),
\qquad
D+1\le2R,
\qquad
b+1\le G(R)(L+1)^{C}.
\]
For a term-disjoint system, Proposition~\ref{prop:one-shot-format}\ref{item:format-single} with $q=p$ gives the same bounds, because
\begin{gather*}
k\le p(2p-2)+p-1=2p^2-p-1<2R^2,
\qquad
D+1\le(p+q-1)+1=2p\le2R,
\\
s\le G(p,p)(L+1)\le G(R)(L+1),
\qquad
b+1\le G(p,p)(L+1)^{C}\le G(R)(L+1)^{C},
\end{gather*}
where $G(R)$ is the largest value of $G(p,p)$ for $p\le R$.  Hence one call of $\textsc{Solve}$ costs at most
\[
T_{\mathrm{one}}\le(s+D+b+2)^{C_{\mathrm{sa}}(k+1)}\le G(R)(L+1)^{CR^2},
\]
after an increase of the absolute constant $C$.  By Proposition~\ref{prop:one-shot-format}\ref{item:format-combined}, there are at most $G(R)L^{O(R)}$ combined systems.  There are at most $2^R$ sets $I$, and by Proposition~\ref{prop:one-shot-format}\ref{item:format-single} each set has at most $G(R)L^{O(R)}$ profiles, so there are at most $2^RG(R)L^{O(R)}$ term-disjoint systems.  Thus the number of calls of $\textsc{Solve}$ is at most $N_{\mathrm{sys}}\le G(R)(L+1)^{CR}$, and all calls together cost at most
\[
N_{\mathrm{sys}}T_{\mathrm{one}}
\le G(R)(L+1)^{CR^2+CR}
\le G(R)(L+1)^{CR^2},
\]
after a further increase of $C$.

\costpart{item:cost-reconstruct}~\emph{Reconstruct.}  The algorithm calls $\textsc{Reconstruct}$ at most once, for $y$ in Phase~II or for $y'$ in Phase~III, and we write $y$ in both cases.  It first computes the coefficient rows as in item~\ref{item:reconstruct-rows} of Step~\ref{step:main-reconstruct}.  Let $f$ be the polynomial of $y$, and let $h$ be the bit length of the exact representation of $y$.  Since $D+1\le2R$ and $k+1\le2R^2$, Theorem~\ref{thm:sampling} gives
\begin{gather}
\deg f\le(D+1)^{C_{\mathrm{sa}}(k+1)}\le(2R)^{2C_{\mathrm{sa}}R^2},\label{eq:output-degree}\\
h\le(s+D+b+2)^{C_{\mathrm{sa}}(k+1)}\le G(R)(L+1)^{CR^2}.\nonumber
\end{gather}
Fix a block, an active mode $m$, and an omitted row $i$.  The coefficients of the system \eqref{eq:row-lift} are the entries of $\bar X_m$ and $\bar a_{m,i}$, which have bit length at most $(L+1)^{C}$ by Step~\ref{step:format-data} of the proof of Proposition~\ref{prop:one-shot-format}, and the entries of $\bar Y_m$, which by \eqref{eq:simplex} are coordinates of $y$ or one minus a sum of at most $2R$ coordinates of $y$.  Hence every coefficient is obtained by at most $2R$ additions and subtractions from given elements of bit length at most $h+(L+1)^{C}$, and by Lemma~\ref{lem:field-lifting}\ref{item:field-arithmetic} with $t=2R$ it has bit length at most
\[
h':=G(2R)(\deg f+1)^{C}\bigl(2h+(L+1)^{C}+1\bigr)^{C}\le G(R)(L+1)^{CR^2}.
\]
The system has $\delta_m\le p+q-1\le2R$ unknowns, $p\le R$ equations, and $q\le R$ inequalities, and it has a real solution by Step~\ref{step:main-reconstruct}.  Hence Lemma~\ref{lem:field-lifting}\ref{item:field-linear-systems} applies with $t=4R$.  There is one such system for every block, every active mode, and every omitted row, so there are at most $R\sum_mn_m\le RL$ systems.  Therefore the computation of all coefficient rows costs at most
\[
RL\cdot G(4R)(\deg f+1)^{C}(h+h'+1)^{C}\le G(R)(L+1)^{CR^2}
\]
bit operations.

Then $\textsc{Reconstruct}$ forms the output as in items~\ref{item:reconstruct-terms} and~\ref{item:reconstruct-list} of Step~\ref{step:main-reconstruct}.  The output has at most $R\sum_mn_m\le L$ factor entries, see Section~\ref{sec:model}.  Each of them is computed by $G(R)$ operations in $\Q(\theta)$ from the coordinates of $y$, the coefficient rows, and rational data, which have bit length at most $G(R)(L+1)^{CR^2}$ by the previous paragraph.  By Lemma~\ref{lem:field-lifting}\ref{item:field-arithmetic}, each output entry therefore has bit length at most $G(R)(L+1)^{CR^2}$ and costs at most $G(R)(L+1)^{CR^2}$ bit operations, so all output entries cost at most $L\cdot G(R)(L+1)^{CR^2}$ bit operations.

\textup{(d)}~\emph{Total.} Hence the whole algorithm costs at most
\[
\underbrace{G(R)L^{O(R)}+N_{\mathrm{sys}}G(R)(L+1)^{C}}_{\text{\ref{item:cost-forming} listing, forming}}
+\underbrace{N_{\mathrm{sys}}T_{\mathrm{one}}}_{\text{\ref{item:cost-solve} }\textsc{Solve}}
+\underbrace{\underbrace{G(R)(L+1)^{CR^2}}_{\text{coefficient rows}}
+\underbrace{L\cdot G(R)(L+1)^{CR^2}}_{\text{output entries}}}_{\text{\ref{item:cost-reconstruct} }\textsc{Reconstruct}}
\le G(R)L^{CR^2}
\]
bit operations, after an increase of $C$ and since $L\ge1$.  This is \eqref{eq:running-time}.  All output entries lie in the field $\Q(\theta)$ of $y$, where $\deg f\le(2R)^{2C_{\mathrm{sa}}R^2}$, and the cost of writing them is included in this bound.  This proves Remark~\ref{rem:main-format}.
\end{proof}

\section{Implementation and detailed experimental evidence}\label{app:experiments}

Section~\ref{sec:experiments} gives the experimental setting and headline results. This appendix records the exact simplifications, equivalent solver formulations, and detailed instance analyses.

For a profile of a block $(I,q)$, we denote the core \eqref{eq:core} by $\mathcal C$.  We write $\T_I=\sum_{r\in I}P_r$ for $I\subseteq[R]$, and $\T_{(m)}$ for the mode-$m$ matricization of $\T$, and likewise for $\T_I$ and $\mathcal C$.  As in Section~\ref{sec:normalized}, the supplied columns $\bar X_m[:,r]$ are indexed by $r\in I$, and the competing columns $\bar Y_m[:,j]$ by $j\in[q]$.

\subsection{Implementation}\label{subsec:exp-setup}

All computations outside $\textsc{Solve}$ are exact over $\Q$. To reduce the number and size of the systems passed to
$\textsc{Solve}$, we use the following 
simplifications. They preserve the minimality and
uniqueness decisions.
\begin{enumerate}[label=\textup{(\alph*)}]
\item\label{item:exp-bounds} \emph{Rank lower bounds.}  Put $\rho_0=\max_m\rank\T_{(m)}$ and $\rho_I=\max_m\rank(\T_I)_{(m)}$. 
By \eqref{eq:unfolding-rank}, $\rho_0\le\rankp(\T)$ and $\rho_I\le\rankp(\T_I)$.  A solution on $(I,q)$ gives decompositions of $\T_I$ and $\T$ of lengths $q$ and $R-p+q$ by Theorem~\ref{thm:representation}\ref{item:soundness} and Lemma~\ref{lem:combine}, so Phase~I need only consider 
\[
\max\{\rho_I,p-R+\rho_0\}\le q<p.
\]
If $\rho_0=R$, the supplied decomposition is minimal, as in Proposition~\ref{prop:matricization}.
\item\label{item:exp-skip} \emph{Redundant profiles.}  Profiles with an active mode $m$ such that $K_m\ne\varnothing$ and $\rank\mathcal C_{(m)}=q$ are omitted in both phases; they are redundant by Lemma~\ref{lem:exp-redundant} below.
\end{enumerate}

Fix a profile of $(I,q)$ with active modes $H$.  For $m\in H$, write $M_{m,j}=\bigotimes_{m'\in H\setminus\{m\}}\bar Y_{m'}[:,j]$ as a column in the order of $\mathcal C_{(m)}$, and put $\mathbf M_m=[M_{m,1}\ \cdots\ M_{m,q}]$.  Then the core equation \eqref{eq:core-equation} gives
\begin{equation}\label{eq:exp-matricized}
\mathcal C_{(m)}=\sum_{j=1}^q\omega_j\bar Y_m[:,j]M_{m,j}^{\mathsf T}=\bar Y_m\operatorname{diag}(\omega_1,\ldots,\omega_q)\mathbf M_m^{\mathsf T},
\quad\text{and}\quad
\col\mathcal C_{(m)}\subseteq\col\bar X_m.
\end{equation}
The inclusion follows by matricizing the left side of \eqref{eq:core-equation} in the same way.

\begin{lemma}\label{lem:exp-rank-q}
If $\rank\mathcal C_{(m)}=q$, every solution satisfies $\col\bar Y_m=\col\mathcal C_{(m)}\subseteq\col\bar X_m$, and the columns of $\mathbf M_m$ form a basis of $\col\mathcal C_{(m)}^{\mathsf T}$.
\end{lemma}

\begin{proof}
By \eqref{eq:exp-matricized}, $q=\rank\mathcal C_{(m)}
\le\min\{\rank\bar Y_m,\rank\mathbf M_m\}\le q$. Thus both $\bar Y_m$ and $\mathbf M_m$ have rank $q$.
The factorization \eqref{eq:exp-matricized} and its transpose give
\[
\col\mathcal C_{(m)}\subseteq\col\bar Y_m,
\qquad
\col\mathcal C_{(m)}^{\mathsf T}\subseteq\col\mathbf M_m.
\]
Both inclusions are equalities by dimension. Since $\mathbf M_m$ has $q$ linearly independent columns,
these columns form the claimed basis.
\end{proof}

\begin{lemma}\label{lem:exp-redundant}
If a profile has an active mode $m_0$ with $K_{m_0}\ne\varnothing$ and $\rank\mathcal C_{(m_0)}=q$, every solution gives a solution for another profile of the same $(I,q)$ with smaller $\sum_{m=1}^{d}\delta_m$, preserving \eqref{eq:td} when $q=p$.  Such profiles may therefore be omitted without changing Corollary~\ref{cor:criteria} or Proposition~\ref{prop:one-shot}\ref{item:one-shot-rank-formula}.
\end{lemma}
\begin{proof}
By Theorem~\ref{thm:representation}\ref{item:soundness}, the solution gives a localized replacement \eqref{eq:lifted-exchange}. For each mode $m$, let $\delta'_m$ be the dimension of the span of its supplied and new mode-$m$ factors, as in Section~\ref{subsec:derivation}.  In an active mode $m$, the new factors are positive multiples of the columns of $\bar B_m=\Lambda_m\bar Y_m$ by \eqref{eq:normalized-competing-factor}--\eqref{eq:reconstructed-term}, and \eqref{eq:normalized-supplied-factor} gives $A_I^{(m)}=\Lambda_m\bar X_mD_{X,m}$.  Since $\Lambda_m$ has $\delta_m$ columns, it follows that
\[
\delta'_m=\rank[\,A_I^{(m)}\ \ \bar B_m\,]=\rank\bigl(\Lambda_m[\,\bar X_mD_{X,m}\ \ \bar Y_m\,]\bigr)\le\delta_m.
\]
In an inactive mode $m$, all factors are multiples of $u_m$ by \eqref{eq:inactive-ratios} and \eqref{eq:reconstructed-term}, so $\delta'_m=1=\delta_m$.  For $m=m_0$, Lemma~\ref{lem:exp-rank-q} gives $\col\bar Y_{m_0}\subseteq\col\bar X_{m_0}$, hence $\col\bar B_{m_0}\subseteq\col(\Lambda_{m_0}\bar X_{m_0})=\col A_I^{(m_0)}$ and $\delta'_{m_0}=\rank A_I^{(m_0)}=\sigma_{m_0}$, which is less than $\delta_{m_0}$ because $K_{m_0}\ne\varnothing$.  Thus $\sum_m\delta'_m<\sum_m\delta_m$.

Theorem~\ref{thm:representation}\ref{item:completeness} assumes that the replacement is irreducible, which need not hold here.  Its proof, through Section~\ref{subsec:derivation}, uses irreducibility only to show that the profile of the replacement satisfies \eqref{eq:profile-budget}, and here $\sum_m(\delta'_m-1)<\sum_m(\delta_m-1)\le p+q-2$ shows this directly.  Hence Section~\ref{subsec:derivation} gives a solution for the profile of this replacement.  If the original solution satisfies \eqref{eq:td}, the replacement is term-disjoint by Theorem~\ref{thm:representation}\ref{item:soundness}, and Step~1 of the proof of Theorem~\ref{thm:representation} shows that the new solution satisfies \eqref{eq:td} as well.
\end{proof}

In Phase~I, Algorithm~\ref{alg:main} runs through $\rho=1,\ldots,R-1$, the patterns $\mathfrak D$ with $\rho(\mathfrak D)=\rho$, and the choices of profiles for their blocks, and solves the combined system $\mathcal N_{\mathfrak D}$ for each of them until one has a solution.  The implementation forms no combined system.  It removes blocks by \ref{item:exp-bounds} and profiles by \ref{item:exp-skip}, and then solves the block system $\mathcal N_{I,q}$ separately for every remaining block $(I,q)$ and every remaining profile.  The blocks removed by \ref{item:exp-bounds} have no feasible system by \ref{item:exp-bounds}, and a profile removed by \ref{item:exp-skip} can have a solution only if a remaining profile of the same block has one, by repeated application of Lemma~\ref{lem:exp-redundant}; the positive integer $\sum_m\delta_m$ strictly decreases at each application, so the process terminates.  Hence, if none of the solved systems has a solution, no block system has one, so by Definition~\ref{def:combined} no combined system has one, and Proposition~\ref{prop:one-shot}\ref{item:one-shot-rank-formula} gives $\rho^*=R$.  This occurred on all 114 inputs.  Phase~III follows line~\ref{line:loop-subset} of Algorithm~\ref{alg:main} with the remaining profiles, stopping at its first feasible system.

$\textsc{Solve}$ uses Z3~5.0.0 \citep{deMouraBjorner2008}, a widely used solver developed at Microsoft Research for Boolean combinations of polynomial equations and inequalities over the reals.  It returns a rational or real algebraic solution, or decides infeasibility.  This solver does not inherit Theorem~\ref{thm:sampling}'s complexity bound. After reconstruction, every returned solution was verified exactly, including the tensor equality at all $6\times 6\times 6$ indices, nonnegative nonzero factors, and inequivalence in Phase~III.

\subsection{Splitting and parametrizing the systems}\label{subsec:exp-solve}

Some systems are not decided by Z3 in reasonable time.  We therefore use two reformulations of the systems, which preserve feasibility.  First, using the supports of the competing terms, we replace a system by finitely many systems with additional linear conditions (Lemma~\ref{lem:exp-support}); if there are none, the system is infeasible.  Second, for two systems of instance 098, we reduce the number of variables by a parametrization (Proposition~\ref{prop:exp-curve}). Section~\ref{subsec:exp-setup} determines which systems are passed to $\textsc{Solve}$.  Section~\ref{subsec:exp-solve} describes the equivalent reformulations which $\textsc{Solve}$ applies to such a system before calling Z3.

\pfpart{Support} Fix a profile of $(I,q)$ with active modes $H$.  The systems $\mathcal N_{I,q}$ and $\mathcal N^{\mathrm{td}}_{I,p}$ do not depend on the order of the competing terms, so these terms can be relabeled.  A \emph{box} is a set $\mathcal B=\prod_m\mathcal B_m$ with nonempty $\mathcal B_m\subseteq[n_m]$.  A box $\mathcal B\subseteq\supp(\T_I)$ is \emph{maximal} if there is no box $\mathcal B'$ with $\mathcal B\subsetneq\mathcal B'\subseteq\supp(\T_I)$.

\begin{lemma}\label{lem:exp-support}
If $\T_I=\sum_{j=1}^qQ_j$ with nonzero nonnegative rank-one tensors $Q_j=\bigotimes_m b_j^{(m)}$, then $\supp(Q_j)=\prod_m\supp(b_j^{(m)})\subseteq\supp(\T_I)$, and there are maximal boxes $\mathcal B^1,\ldots,\mathcal B^q$, not necessarily distinct, with
\[
\supp(Q_j)\subseteq\mathcal B^j\quad \text{for $j\in[q]$},
\qquad
\supp(\T_I)=\bigcup_{j=1}^{q}\mathcal B^j.
\]
\end{lemma}

\begin{proof}
First, $\supp(Q_j)=\prod_m\supp(b_j^{(m)})$, which is a box because every $b_j^{(m)}$ is nonzero.  All $Q_j$ are nonnegative, so $\supp(\T_I)=\bigcup_{j=1}^{q}\supp(Q_j)$, and $\supp(Q_j)\subseteq\supp(\T_I)$.  Among the boxes $\mathcal B$ with $\supp(Q_j)\subseteq\mathcal B\subseteq\supp(\T_I)$, let $\mathcal B^j$ be one with the largest number of elements.  Then $\mathcal B^j$ is maximal, and $\supp(\T_I)=\bigcup_j\supp(Q_j)\subseteq\bigcup_j\mathcal B^j\subseteq\supp(\T_I)$, so $\supp(\T_I)=\bigcup_j\mathcal B^j$.
\end{proof}

A solution of the system gives a decomposition $\T_I=\sum_{j=1}^q\widetilde Q_j$ by Theorem~\ref{thm:representation}\ref{item:soundness}.  So, by Lemma~\ref{lem:exp-support}, the system has no solution if no $q$ maximal boxes have union $\supp(\T_I)$.  This is the tensor version of the rectangle covering bound for nonnegative matrix rank \citep{Yannakakis1991,FioriniEtAl2013}.  Otherwise, $\textsc{Solve}$ runs through all multisets $\{\mathcal B^1,\ldots,\mathcal B^q\}$ of maximal boxes with union $\supp(\T_I)$.  Write $\mathcal B^j=\prod_m\mathcal B^j_m$ and $\mathcal P_j=\mathcal B^j\setminus\bigcup_{k\ne j}\mathcal B^k$, and suppose that $\supp(\widetilde Q_j)\subseteq\mathcal B^j$ for $j\in[q]$.  For $m\in H$, the mode-$m$ factor of $\widetilde Q_j$ is a positive multiple of $\bar B_m[:,j]$ by \eqref{eq:normalized-competing-factor}--\eqref{eq:reconstructed-term}.  By Lemma~\ref{lem:exp-support}, it follows that $\bar B_m[i,j]=0$ for $i\notin\mathcal B^j_m$.  The other terms vanish on $\mathcal P_j$, so $\widetilde Q_j[c]=\T_I[c]>0$ for $c=(c_1,\ldots,c_d)\in\mathcal P_j$.  This gives $\bar B_m[c_m,j]>0$, and 
\[
\bar B_m[c_m,j]\,\T_I[c']=\bar B_m[c'_m,j]\,\T_I[c]
\]
if $c'\in\mathcal P_j$ differs from $c$ only in the $m$th coordinate.  Row $i$ of $\bar B_m=\Lambda_m\bar Y_m$ is a row of $\bar Y_m$ if $i\in S_m$, and it is $w_{m,i}\bar Y_m$ if $i\notin S_m$ and $K_m=\varnothing$, by the proof of Proposition~\ref{prop:lifting}.  So if $K_m=\varnothing$, or if all rows in a condition lie in $S_m$, then the condition is linear in the variables.  For each multiset, $\textsc{Solve}$ adds these linear conditions to the system and passes the new system to Z3.  By Lemma~\ref{lem:exp-support}, after relabeling, every solution of the system solves one of the new systems.  Hence the system has a solution if and only if one of the new systems has one.

\pfpart{Parametrization}  We now consider profiles with
\begin{equation}\label{eq:exp-curve-case}
H=\{m_0,a,b\},\qquad\delta_a=2,\qquad\rank\mathcal C_{(m_0)}=q=\delta_b+1.
\end{equation}
For such a profile, $M_{m_0,j}=\bar Y_a[:,j]\otimes\bar Y_b[:,j]\in\R^{2\delta_b}$, and $\dim \ker\mathcal C_{(m_0)}=2\delta_b-q=\delta_b-1$.  Let $\Xi\in\R^{(\delta_b-1)\times 2\delta_b}$ be a matrix whose rows form a basis of $\ker\mathcal C_{(m_0)}$. Partition $\Xi=[\,\Xi_1\ \Xi_2\,]$, where
$\Xi_1,\Xi_2\in\R^{(\delta_b-1)\times\delta_b}$. For $s\in\R$, let $L(s)=s\Xi_1+(1-s)\Xi_2\in\R^{(\delta_b-1)\times\delta_b}$, equivalently $L(s)v=\Xi\bigl((s,1-s)^{\mathsf T}\otimes v\bigr)$ for $v\in\R^{\delta_b}$. Here $L(s)$ is a matrix pencil \citep{JaJa1979}; see also \citet[Chapter~XII]{Gantmacher1959}. We call $s$ \emph{special} if $\rank L(s)<\delta_b-1$.  If $\delta_b=2$, then $L(s)$ is one row, write as $(\alpha(s),\beta(s))$.  Since $\delta_a=2$, we write $\bar Y_a[:,j]=(s_j,1-s_j)^{\mathsf T}$ with $s_j\in[0,1]$.

\begin{proposition}\label{prop:exp-curve}
Assume \eqref{eq:exp-curve-case}.  Then every solution of $\mathcal N_{I,q}$ satisfies the following.
\begin{enumerate}[label=\textup{(\roman*)}]
\item\label{item:exp-curve-dual} The columns of $\mathbf M_{m_0}$ form a basis of $\col\mathcal C_{(m_0)}^{\mathsf T}$, and
\begin{equation}\label{eq:exp-cramer}
\bar Y_{m_0}\operatorname{diag}(\omega_1,\ldots,\omega_q)=\mathcal C_{(m_0)}\mathbf M_{m_0}(\mathbf M_{m_0}^{\mathsf T}\mathbf M_{m_0})^{-1}.
\end{equation}
\item\label{item:exp-curve-other} $\bar Y_{m_0}[:,j]$ depends only on the $M_{m_0,k}$ with $k\ne j$.
\item\label{item:exp-curve-param} $L(s_j)\bar Y_b[:,j]=0$.  If $s_j$ is not special, then $\dim\ker L(s_j)=1$; if moreover $\delta_b=2$, this gives
\begin{equation}\label{eq:exp-curve}
\bar Y_b[:,j]=\frac{1}{\beta(s_j)-\alpha(s_j)}\begin{pmatrix}\beta(s_j)\\-\alpha(s_j)\end{pmatrix}.
\end{equation}
\end{enumerate}
Conversely, let $\bar Y_a$ and $\bar Y_b$ have nonnegative columns with entry sum one, such that the columns of $\mathbf M_{m_0}$ form a basis of $\col\mathcal C_{(m_0)}^{\mathsf T}$ and every column of the right-hand side of
\eqref{eq:exp-cramer} is nonnegative and nonzero.  Let $\omega_j$ be the entry sum of the $j$th column
of the right-hand side of \eqref{eq:exp-cramer}, and define $\bar Y_{m_0}$ by \eqref{eq:exp-cramer}.  Then \ref{cond:simplex}--\ref{cond:core} hold.
\end{proposition}

\begin{proof}
\pfpart{Part~\ref{item:exp-curve-dual}} By \eqref{eq:exp-curve-case} and Lemma~\ref{lem:exp-rank-q} with $m=m_0$, the $M_{m_0,j}$ form a basis of $\col\mathcal C_{(m_0)}^{\mathsf T}$.  Hence $\mathbf M_{m_0}$ has rank $q$, and $\mathbf M_{m_0}^{\mathsf T}\mathbf M_{m_0}$ is invertible.  Multiplying \eqref{eq:exp-matricized} from the right by $\mathbf M_{m_0}(\mathbf M_{m_0}^{\mathsf T}\mathbf M_{m_0})^{-1}$ gives \eqref{eq:exp-cramer}.

\pfpart{Part~\ref{item:exp-curve-other}} By \eqref{eq:exp-cramer}, $\bar Y_{m_0}[:,j]$ is a positive multiple of $\mathcal C_{(m_0)}$ times the $j$th column of
\[
\mathbf M_{m_0}(\mathbf M_{m_0}^{\mathsf T}\mathbf M_{m_0})^{-1}.
\]
This column lies in $\col\mathbf M_{m_0}$, which has dimension $q$, and is orthogonal to the $M_{m_0,k}$ with $k\ne j$ since $\mathbf M_{m_0}^{\mathsf T}\mathbf M_{m_0}(\mathbf M_{m_0}^{\mathsf T}\mathbf M_{m_0})^{-1}=\Id_q$.  Therefore $\mathbf M_{m_0}(\mathbf M_{m_0}^{\mathsf T}\mathbf M_{m_0})^{-1}[:,j]$ is determined up to a nonzero scalar, which cancels by normalizing $\bar Y_{m_0}[:,j]$. Hence $\bar Y_{m_0}[:,j]$ depends only on the $M_{m_0,k}$ with $k\ne j$.

\pfpart{Part~\ref{item:exp-curve-param}} By Lemma~\ref{lem:exp-rank-q}, the vector $M_{m_0,j}=(s_j,1-s_j)^{\mathsf T}\otimes\bar Y_b[:,j]$ lies in $\col\mathcal C_{(m_0)}^{\mathsf T}$, which is orthogonal to $\ker\mathcal C_{(m_0)}$. Hence $L(s_j)\bar Y_b[:,j]=\Xi M_{m_0,j}=0$, so $\bar Y_b[:,j]\in\ker L(s_j)$. If $s_j$ is not special, then $\dim \ker L(s_j)=1$, and $\bar Y_b[:,j]$ is its vector with entry sum one.  For $\delta_b=2$, $\ker L(s_j)$ is spanned by $(\beta(s_j),-\alpha(s_j))^{\mathsf T}$, which gives \eqref{eq:exp-curve}.

For the converse, the rows of $\mathcal C_{(m_0)}$ lie in $\col\mathbf M_{m_0}$, so multiplying $\mathcal C_{(m_0)}$ from the right by the projection $\mathbf M_{m_0}(\mathbf M_{m_0}^{\mathsf T}\mathbf M_{m_0})^{-1}\mathbf M_{m_0}^{\mathsf T}$ onto this space does not change it.  By \eqref{eq:exp-cramer}, this gives $\mathcal C_{(m_0)}=\bar Y_{m_0}\operatorname{diag}(\omega_1,\ldots,\omega_q)\mathbf M_{m_0}^{\mathsf T}$, which is \ref{cond:core} by \eqref{eq:exp-matricized}.  The entry sum of $\mathcal C_{(m_0)}$ is $\Omega_I$ by \eqref{eq:core}, and each $M_{m_0,j}$ has entry sum one.  Hence $\sum_j\omega_j=\Omega_I$, as required by \eqref{eq:last-weight}. Since $\bar Y_a$, $\bar Y_b$, and the right side of \eqref{eq:exp-cramer} are nonnegative, and its columns are nonzero, \ref{cond:simplex} and \ref{cond:weights} hold.
\end{proof}

The \emph{rank-one curve} consists of all pairs
$(u,v)$ of nonnegative vectors in $\R^2$ and
$\R^{\delta_b}$, respectively, each with entry sum
one and with
$u\otimes v\in\col\mathcal C_{(m_0)}^{\mathsf T}$. By Proposition~\ref{prop:exp-curve}\ref{item:exp-curve-dual}, 
for every solution of $\mathcal N_{I,q}$, the pairs
$(\bar Y_a[:,j],\bar Y_b[:,j])$ lie on this curve,
with $\bar Y_a[:,j]=(s_j,1-s_j)^{\mathsf T}$
and $s_j\in[0,1]$.
If $s_j$ is not special, then $\bar Y_b[:,j]$
is uniquely determined by
$\bar Y_b[:,j]\in \ker L(s_j)$ and its entry sum one.
By \eqref{eq:exp-cramer}, non-special $s_1,\dots, s_q$ determine
all remaining variables. Thus every solution of
$\mathcal N_{I,q}$ is represented in the
parametrized systems.

Conversely, choose $s_j\in[0,1]$ and
$\bar Y_b[:,j]\in\ker L(s_j)$ with nonnegative
entries and entry sum one.
Put $\bar Y_a[:,j]=(s_j,1-s_j)^{\mathsf T}$
and $M_{m_0,j}=\bar Y_a[:,j]\otimes\bar Y_b[:,j]$.
Since $\Xi M_{m_0,j}=L(s_j)\bar Y_b[:,j]=0$,
we have $M_{m_0,j}\in
\col\mathcal C_{(m_0)}^{\mathsf T}$.
If $\rank\mathbf M_{m_0}=q$ and every column
of the right-hand side of \eqref{eq:exp-cramer}
is nonnegative and nonzero, the converse in
Proposition~\ref{prop:exp-curve} gives
$\bar Y_{m_0}$ and $\omega_j$ satisfying
\ref{cond:simplex}--\ref{cond:core}.
The additional conditions \ref{cond:lifting}
and, when $q=p$, \eqref{eq:td} then give
a solution of the original system.
Therefore the parametrized systems preserve
feasibility. When $s_j$ is special, $\textsc{Solve}$ keeps
the entries of $\bar Y_b[:,j]$ as variables.

\subsection{Detailed results}\label{subsec:exp-details}

\pfpart{Minimality}  For 66 inputs, $\rho_0=5$ certifies minimality by the rank lower bound~\ref{item:exp-bounds} in
Section~\ref{subsec:exp-setup}.  The other 48 have $\rho_0=4$, allowing only $q=p-1$: 42 have no such profile; 007, 011, 014, 068, 082, and 092 yield 10 blocks with one profile each.  None admits a cover by $q$ maximal boxes, so Lemma~\ref{lem:exp-support} proves all 10 systems infeasible without Z3.

\pfpart{Added rows}  Only 007, 014, 068, 082, and 098 encounter $K_m\ne\varnothing$, all in Phase~III and infeasible. For these profiles, \ref{cond:lifting} is handled by Lemma~\ref{lem:circuit-elimination}; elsewhere \ref{cond:lifting} is linear, $w_{m,i}\bar Y_m\ge0$.

\pfpart{Witnesses} Each of the 19 feasible systems yields a rational length-five decomposition inequivalent to the supplied one. Seventeen replace three terms, i.e. $p=q=3$, with $\sigma$ a permutation of $(3,2,2)$, where $\sigma_m=\rank A_I^{(m)}$. 

The other two, 057 and 092, replace four terms, i.e. $p=q=4$, with $\sigma=(2,4,3)$ and $(3,3,3)$, respectively. In all cases, $\sum_m(\sigma_m-1)=p+q-2$ forces $\delta_m=\sigma_m$ by \eqref{eq:profile-budget}. Thus each of the first 17 systems has 14 variables, and each of the other two has 27. For 057 and 092, Z3 found every term-disjoint system on at most three terms infeasible.

\pfpart{Instance 098} Only this instance required Proposition~\ref{prop:exp-curve}, for the blocks $(\{1,2,4\},3)$ and $(\{1,2,4,5\},4)$, with $\delta=(3,2,2)$ and $(4,2,3)$, and 14 and 27 variables, respectively. Without parametrization, neither system was decided within one hour.  With parametrization, Z3 proved each system infeasible in under 15 seconds.  The only profile skipped by Lemma~\ref{lem:exp-redundant} occurred in instance 098, with $I=[5]$, $q=5$, $\delta=(6,2,3)$, $K_1=\{6\}$, and $\rank\mathcal C_{(1)}=5$.

\subsection{Decompositions of three subsums}\label{subsec:exp-families}

\begin{figure}[p]
\centering
\includegraphics[width=0.90\linewidth]{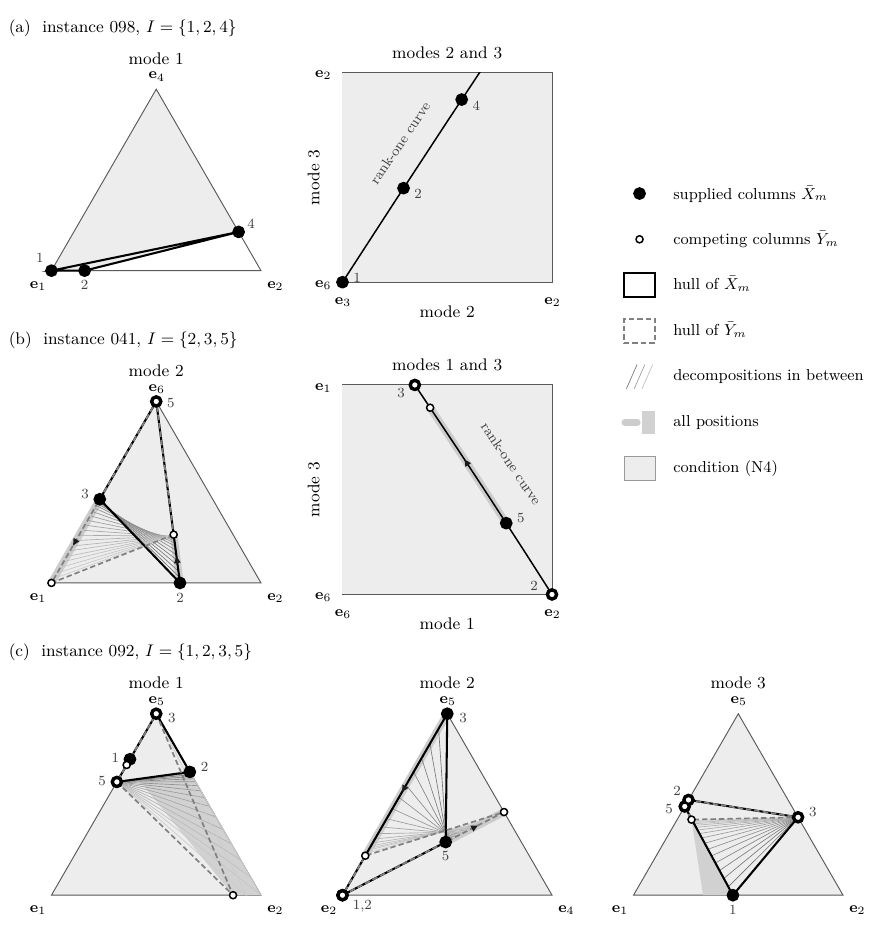}
\caption{Nonnegative length-$|I|$ decompositions of $\T_I$ in selected-row coordinates.  Triangles: $\delta_m=3$; squares: pairs of $\delta_m=2$ columns.  Black points/solid hulls: supplied decomposition; white points/dashed hulls: one other decomposition.  Thin grey hulls: a sequence of other decompositions; thick bands and dark regions: all positions of one point, with arrows away from the supplied position.  Light grey: condition~\ref{cond:lifting}; black curve in (a) and (b): rank-one curve.  (a)~098, $I=\{1,2,4\}$: only the supplied decomposition.  (b)~041, $I=\{2,3,5\}$: decompositions with one parameter.  (c)~092, $I=\{1,2,3,5\}$: decompositions with two parameters.}
\label{fig:exp-families}
\end{figure}

For instances 098, 041, and 092, Figure~\ref{fig:exp-families} shows inequivalent nonnegative decompositions of subsum $\mathcal T_I$ with $|I|$ terms. For each subsum, the effective-order bound \eqref{eq:profile-budget} forces $\delta_m=\sigma_m$ and $K_m=\varnothing$ in all three active modes, so there is only one profile. For a mode $m$ with $\delta_m=3$, each column $\bar X_m[:,r]$, $r\in I$, is represented by a point in a triangle, since its entries sum to one. The vertices are labelled by the selected row indices in $S_m$.

For two modes $a,b$ with $\delta_a=\delta_b=2$, each pair of columns $(\bar X_a[:,r],\bar X_b[:,r])$, $r\in I$, is represented by a point in a square, with coordinates given by their first entries. The ends of the bottom and left edges are labelled by the selected row indices in $S_a$ and $S_b$, respectively. For example, in Figure~\ref{fig:exp-families}(a), $S_2=\{2,3\}$ and $S_3=\{2,6\}$. The black point labelled $1$ represents $\bar X_2[:,1]=\bar X_3[:,1]=(0,1)^{\mathsf T}$. The competing columns $\bar Y_m[:,j]$, $j\in[q]$, are represented in the same way.  

The light grey regions show the columns $y$ satisfying $w_{m,i}y\ge0$ for every omitted row $i$, as required by \ref{cond:lifting}. Only in Figure~\ref{fig:exp-families}(a), the mode-$1$ triangle is truncated at $\be_2$ by \ref{cond:lifting}, with supplied columns 1, 2, and 4 on the boundary of the remaining region. If a competing column equals the supplied column, the white and black points coincide.

In each triangle, the hull of a decomposition is the convex hull of its points.  The solid hull is given by the black points, and the dashed hull by the white points.  The thin grey hulls are given by a sequence of other decompositions.  These decompositions are specified below.  A thick band or a dark region shows all positions of one point, and an arrow on a band points away from the supplied position.  For 098 and 041, put $(m_0,a,b)=(1,2,3)$ and $(2,1,3)$, respectively.  In both cases, $\rank\mathcal C_{(m_0)}=3$, so \eqref{eq:exp-curve-case} holds, and the black curve in each square is the rank-one curve of Section~\ref{subsec:exp-solve}.

Instance 098 is essentially unique (Figure~\ref{fig:exp-outcomes}), so in (a), the supplied decomposition is the only one.  In (b), for every $s$ with $s_{\min}\le s\le s_{\max}$, the parameters $(s_1,s_2,s_3)=(\bar X_1[1,2],\bar X_1[1,3],s)$ of Proposition~\ref{prop:exp-curve} give a decomposition, where $s_{\min}=4491371/10735830$ and $s_{\max}=\bar X_1[1,5]=211/270$.  We verified this exactly by real-root isolation.  Distinct $s$ give inequivalent decompositions, and $s=s_{\max}$ gives the supplied one.  The white points are given by $s=s_{\min}$, and the thin grey hulls by a sequence of values of $s$ in $(s_{\min},s_{\max})$, darker for larger $s$.  In (c), the white points show another decomposition, and the thin grey hulls show a sequence of decompositions between the supplied decomposition and this one, darker nearer the supplied one.  The subsums in (a) and (b) have the same ranks $\sigma_m$ up to the order of the modes, so their systems have the same numbers of variables and equations.  Hence the ranks alone do not decide uniqueness here, and the difference between (a) and (b) comes from the inequalities \ref{cond:simplex}, \ref{cond:weights}, and \ref{cond:lifting}.

\end{document}